\pdfoutput=1

\documentclass[12pt, english]{article}
\usepackage[T1]{fontenc}
\usepackage[latin9]{inputenc}
\usepackage[margin=1in,bottom=1in,top=1in]{geometry}
\usepackage{lscape}
\usepackage{setspace}
\usepackage{footmisc}
\usepackage{amsmath}
\usepackage{amsthm}
\usepackage{bm}
\usepackage{amssymb}
\usepackage{stmaryrd}
\usepackage{graphicx}
\usepackage{float}
\usepackage{rotating}
\usepackage{subcaption}
\usepackage{ragged2e}
\usepackage{esint}
\usepackage{ulem}
\usepackage{babel}
\usepackage{longtable}
\usepackage{booktabs}
\usepackage{caption,booktabs,array}

\usepackage[round]{natbib}
\usepackage[hypertexnames=false]{hyperref}

\usepackage{colortbl}
\definecolor{darkblue}{rgb}{0,0,.6}
\hypersetup{
colorlinks=true,
linkcolor=darkblue,
filecolor=darkblue,
urlcolor=darkblue,
citecolor=darkblue,
}
\usepackage[toc,page]{appendix}

\usepackage[sc]{mathpazo}
\usepackage{colortbl}

\newtheorem{assumption}{Assumption}
\newtheorem{definition}{Definition}
\newtheorem{remark}{Remark}
\newtheorem{example}{Example}
\newtheorem{theorem}{Theorem}

\newtheorem{lemma}{Lemma}

\usepackage{dsfont}
\usepackage{pifont}
\usepackage{enumerate}

\usepackage{tikz}
\usetikzlibrary{positioning, arrows.meta}
\usetikzlibrary{arrows}
\usetikzlibrary{patterns}
  \pgfdeclarepatternformonly{north east lines wide}%
        {\pgfqpoint{-1pt}{-1pt}}%
        {\pgfqpoint{10pt}{10pt}}%
        {\pgfqpoint{9pt}{9pt}}%
        {
            \pgfsetlinewidth{0.3pt}
            \pgfpathmoveto{\pgfqpoint{0pt}{0pt}}
            \pgfpathlineto{\pgfqpoint{9.0pt}{9.0pt}}
            \pgfusepath{stroke}
        }

        \pgfdeclarepatternformonly{north west lines wide}%
        {\pgfqpoint{-1pt}{1pt}}%
        {\pgfqpoint{10pt}{-10pt}}%
        {\pgfqpoint{9pt}{-9pt}}%
        {
            \pgfsetlinewidth{0.3pt}
            \pgfpathmoveto{\pgfqpoint{0pt}{0pt}}
            \pgfpathlineto{\pgfqpoint{9.0pt}{-9.0pt}}
            \pgfusepath{stroke}
        }

\definecolor{green}{rgb}{0.0,0.7,0.0}

\begin{document}
\doublespacing
\normalem

\title{Sensitivity Analysis in Unconditional Quantile Effects
\textcolor{red}{}}
\author{Juli\'an Mart\'inez-Iriarte\thanks{%
\begin{doublespace}
Email: jmart425@ucsc.edu. I am deeply indebted to my advisor Yixiao Sun for his constant support and guidance. I thank Jeffrey Clemens, Dimitris Christelis, Antonio Galvao, Peter Hull, Ying-Ying Lee, Michael Leung, Jessie Li, Xinwei Ma, Augusto Nieto-Barthaburu, Joel Sobel, Pietro Spini, Kaspar Wuthrich, Ying Zhu, and seminar participants at various universities for helpful comments. Two anonymous referees provided excellent comments that greatly improved the paper. All errors remain my own. \end{doublespace}
} \\
Department of Economics,
UC Santa Cruz}
\date{\textbf{\today}}
\maketitle
\thispagestyle{empty}
\vspace{-2em}

\begin{abstract}

This paper proposes a framework to analyze the effects of counterfactual policies on the unconditional quantiles of an outcome variable. For a given counterfactual policy, we obtain identified sets for the effect of both marginal and global changes in the proportion of treated individuals. To conduct a sensitivity analysis, we introduce the quantile breakdown frontier, a curve that $(i)$ indicates whether a sensitivity analysis is possible or not, and $(ii)$ when a sensitivity analysis is possible, quantifies the amount of selection bias consistent with a given conclusion of interest across different quantiles. To illustrate our method, we perform a sensitivity analysis on the effect of unionizing low income workers on the quantiles of the distribution of (log) wages.

\end{abstract}

\bigskip

\clearpage
\pagenumbering{arabic} 
\section{Introduction}

In this paper we propose a sensitivity analysis on the effect of counterfactual policies that change the proportion of treated individuals. Consider a situation where a policy maker is interested in treating non-treated individuals.
The key identification challenge is that we do not observe the counterfactual outcome 
of individuals who switch groups, that is, the \emph{newly} treated individuals. In some cases, however, it is still possible to recover the distribution of the unobserved counterfactual outcome. For example, suppose that treatment status is randomly assigned, and a policy maker increases the proportion of treated individuals by randomly selecting non-treated individuals.\footnote{We assume full compliance in both randomizations.} Although we do not observe the counterfactual outcome of the \emph{newly} treated individuals, we know it is drawn from the same distribution as the \emph{already} treated individuals. Hence, we identified the counterfactual distribution of \emph{newly} treated individuals.

When treatment status is \emph{not} randomly assigned in the first place, the identification strategy previously described breaks down. The reason is that due to the selection bias in the original treatment status, a random selection of individuals from the control group will be drawn from a different distribution. Thus, in the presence of selection bias, identification of the counterfactual distribution requires that the policy maker has enough information to devise a policy such that the (unobservable) distribution of the \emph{newly} treated ``matches'' the distribution of the \emph{already} treated individuals. This is usually infeasible. Even if the policy maker has this information, such as when treatment status is randomly assigned, they might not be interested in a policy that merely selects the \emph{newly} treated individuals at random. 

The previous discussion highlights that identification of counterfactual distributions results in either very stringent information requirements, or in policies that might not be interesting. In both cases, the distribution of the \emph{newly} treated individuals is restricted. From the point of view of the policy maker, this can rule out many interesting policies. To see this, consider the following example. A policy maker might like to know if an increase in the unionization rate reduces inequality. If unionized workers are relatively high-skilled, and a policy expands unionization with low-skilled workers, then the distribution of wages conditional on being in the union, is likely to change. 

In order to analyze a richer set of counterfactual policies, we drop the restrictions on the distribution of the \emph{newly} treated individuals and provide partial identification results for two effects. The first one is a global effect that compares the quantiles of the observed outcome, to those of the counterfactual outcome, where the proportion of treated individuals has been increased by $\delta$. The second one is a marginal effect where we let $\delta$ go to zero, and analyze its limiting effect on the unconditional quantiles of the outcome. 

Another important contribution of this paper is to propose a framework for a sensitivity analysis on certain conclusions of interest. To do this, we quantify the departure from point identification by the vertical distance between the distributions of the \emph{newly} treated individuals and the \emph{already} treated individuals. We introduce a curve called the \emph{quantile breakdown frontier}, which first indicates whether a sensitivity analysis is possible, and second, it quantifies the maximum departure from point identification such that a given set of conclusions holds across different quantiles. Using this curve, we bound the global effects curve using this maximum departure derived from the quantile breakdown frontier. In this way, we obtain an identified region for the global effect curve consistent with the desired conclusions. Estimation of both the quantile breakdown frontier and the bounds on the global effect are based on empirical distribution functions and empirical quantiles, and are $\sqrt n$-consistent. 

The departure from point identification is due to the selection bias induced by the counterfactual policy. We call this the \emph{policy selection bias}. The usual selection bias states that treated and non-treated individuals are different in a sense, and that is what explains the selection in the first place. Instead, the policy selection bias is the difference between the distributions of the \emph{newly} treated individuals and the \emph{already} treated individuals. Returning to the unionization example, the policy selection bias arises because the union wages of \emph{newly} unionized workers may not be drawn from distribution of the \emph{already} unionized workers. We do not know the distribution of union wages of newly unionized workers, hence we can only partially identify the global and marginal effects.

The policy selection bias can be non-negligible even if the original selection into treatment is randomly assigned.\footnote{See Appendix \ref{app_pol_bias} for an analysis of the policy selection bias in the context of an RCT with perfect compliance.} The reason is that, for the policy selection bias, what matters is who the \emph{newly} treated individuals are. Conversely, if there is selection bias initially, but the distribution of the \emph{newly} treated ``matches'' the distribution of the \emph{already} treated individuals, then there will be no policy selection bias. Thus, the policy selection bias depends on the particular counterfactual policy being analyzed, not whether there is selection bias in the original selection mechanism.

We apply these methods to the study of unions and inequality, which has long been of interest to labor economics. A recent comprehensive review of this extensive literature is provided by \cite{Farber2020}. Using the data in \cite{Firpo2009}, our empirical application considers the effect of expanding unionization on the quantiles of the distribution of (log) wages. Our approach allows us to tackle the question from a different perspective. Using the tools developed in this paper, we can quantify the amount of policy selection bias that is consistent with a policy that increases the unionization rate by unionizing low earnings workers. By looking at the global effect in the $30$\textsuperscript{th} quantile of the distribution of wages we investigate the amount of policy selection bias consistent with unions reducing overall inequality. To this end, we examine the following conclusion: whether the $30$\textsuperscript{th} quantile increases by more than 10\%. We find that this is consistent with moderate values of policy selection bias. The bounds on the global effect for other quantiles reveals that this policy can have a bigger effect on quantiles below the $30$\textsuperscript{th} quantile, without hurting those at the middle and top of the distribution of income.

\textbf{Related Literature} There is an extensive literature devoted to the analysis of counterfactual distributions. A good reference is \cite{Firpo2011}. In this paper, we focus on counterfactual distributions that arise as a result of a counterfactual policy that changes the proportion of treated individuals. The Policy Relevant Treatment Effect (PRTE) of \cite{heckman_prte, Heckman2005}, and the Marginal PRTE (MPRTE) of \cite{Carneiro2010,Carneiro2011} are examples of the aforementioned global and marginal effects. The difference is that they analyze the unconditional mean of the outcome. Identification relies on a separable threshold model for the selection equation, and the availability of a continuous instrumental variable. In this setting, the proportion of treated individuals is changed by manipulating the instrumental variable. Our analysis does not make any assumptions on the selection equation. We do not require an instrumental variable either.

The marginal effect on the unconditional quantiles of an outcome was first studied by \cite{Firpo2009}. The identification arguments of \cite{Firpo2009} are based on a distributional invariance assumption: the distribution of the outcome for the original treatment group (under the original policy regime) is the same as that for the new treatment group (under the new policy regime), and this also holds for the control groups under the two policy regimes.\footnote{See the proof to Corollary 3 of the working paper version \cite{Firpo2007}.} For the case of an endogenous binary covariate, where distributional invariance might not hold, \cite{yixiao2020} achieve identification by generalizing the Marginal Treatment Effect framework. \cite{Kasy2016} also analyzes counterfactual policies which assign a binary treatment, but focuses on a welfare ranking. %

\cite{Rothe2012} provides a general treatment for functionals of the unconditional distribution of the outcome. What we call a global effect, \cite{Rothe2012} refers to as a \emph{Fixed Partial Policy Effect}, and what we call a marginal effect, \cite{Rothe2012} refers to as a \emph{Marginal Partial Distributional Policy}. However, \cite{Rothe2012} imposes different identifying assumptions, namely a form of conditional exogeneity, which also yield a partial identified set. We do not impose such assumptions in order to broaden the types of policies we can analyze. 

It is important to highlight that we do not estimate a quantile treatment effect. The quantile treatment effect is the difference between the $\tau$-quantile under treatment and the $\tau$-quantile under control, and depends on the distribution of the covariates. In a recent contribution, \cite{Lieli2020} investigate the changes in this effect when the distribution of the covariates is manipulated. Aside from treatment status, we do not manipulate the distribution of covariates.

Our sensitivity analysis is based on the breakdown analysis of \cite{Kline2013} and \cite{Masten2020}.  \cite{Kline2013} perform a sensitivity analysis in a different context: departures from a missing (data) at random assumption. In a manner similar to us, this departure is measured as the Kolmogorov-Smirnov distance between the distribution of observed outcomes and the (unobserved) distribution of missing outcomes. Our quantile breakdown frontier builds on the breakdown frontier introduced by \cite{Masten2020}. However, instead of relaxing two parameters, we relax just one, and plot it against different quantiles. Another recent application of the breakdown analysis is \cite{noack} in the context of LATE. 

\textbf{Notation} All the CDFs are denoted by $F$ with a subscript indicating the random variable. So, the CDF of $Y$ is $F_Y(y)$. Conditional CDFs are denoted similarly. For example, the CDF of $Y$ conditional on $D=1$ and $X=x$ is denoted by $F_{Y|D=1,X=x}(y)$. The $\tau$-quantile of $Y$ is denoted by $F^{-1}_Y(\tau)$. Weak convergence is denoted by $\rightsquigarrow$.

\section{Counterfactual Policies and Unconditional Effects}\label{section_uncond_effects}

We will work with the potential outcomes framework. For some unknown functions $h_0$ and $h_1$ 
\begin{align*}
Y(0) &= h_0(X,U_0),\\
Y(1) &= h_1(X,U_1),
\end{align*} 
where $X$ are observed covariates and $U_0$ and $U_1$ consist of unobservables. We do not impose any restriction on the dimension of the unobservables. The observed outcome is thus
\begin{align*}
Y &= D\cdot h_1(X,U_1) + (1-D)\cdot h_0(X,U_0).\notag\\
:&=h(D,X,U),
\end{align*} 
for a general nonseparable function $h$, where $D$ is a binary random variable taking values $0$ and $1$, and $U:=(U_0,U_1)'$. The variable $D$ can be interpreted as the treatment status, and $p:=\Pr(D=1)$ is the proportion of treated individuals.

A counterfactual policy is an alternative assignment of individuals to treatment. It is given by a binary random variable $D_\delta$, such that $\Pr(D_\delta=1)=p+\delta$ for a fixed $\delta \in(-p,1-p)$. It is called counterfactual because it may assign $D_\delta=1$ to an individual whose $D=0$. As $\delta$ varies over $(-p,1-p)$, we obtain a collection of (counterfactual) policies which is denoted by $\mathcal D$. When a particular counterfactual policy $D_\delta$ belongs to $\mathcal D$ we write $D_\delta\in\mathcal D$. The counterfactual outcome we would observe for a given $D_\delta\in\mathcal D$ is
\begin{align*}
Y_{D_\delta}&=h(D_\delta,X,U),
\end{align*} 
where we implicitly assumes that the potential outcomes are not affected by the manipulation of $D$. 

Strictly speaking, the counterfactual outcome $Y_{D_\delta}$ is not well defined until we define $\mathcal D$, the collection of counterfactual policies. We will restrict ourselves to policies that shift a portion of individuals in the control group to the treatment group. We refer to such individuals as \emph{newly treated}. This means that for every individual, $D_\delta-D\geq 0$. This is shown in Figure \ref{figure_introduction}.

\begin{figure}
\centering
\begin{tikzpicture}[
    scale=1.5]

\draw [green, thin, fill=green, fill opacity=0.2] plot [smooth] coordinates {(0.5,2) (0,3) (-1.3,3) (-2,2) (-0.5,1)};
\draw [green, thin, fill=cyan!20, fill opacity=0.7] plot [smooth] coordinates {(-0.5,1) (0,0)  (1.355,1) (1.5, 1.5) (0.5,2) };

\draw [green, thin, fill=green, fill opacity=0.2] plot [smooth] coordinates {(6.5,2) (6,3) (4.7,3) (4,2) (5.5,1)};
\draw [green, thin, fill=cyan!20, fill opacity=0.7] plot [smooth] coordinates {(5.5,1) (6,0) (7.355,1)};
\draw [green, fill=red, fill opacity=0.2] plot [smooth] coordinates {(7.355,1) (7.5, 1.5) (6.5,2) };
\fill[fill=red, fill opacity=0.2] (6.5,2)--(7.355,1)--(5.5,1);

 \draw [->, thick] (0.6,2.3)-- (3.8,2.3);
\node [anchor=south] at (2.2,2.3) {\footnotesize\emph{counterfactual policy}};

 \draw [solid, thick] (-0.5,1)-- (0.5,2);
      \node [anchor=north] at (-1,2.3) {\footnotesize$D=1$};

        \node [anchor=north]  at (0.5,1.3) {\footnotesize$D=0$};
        
         \draw [dotted, thick] (5.5,1)-- (6.5,2);
         \draw [solid, thick] (5.5,1)-- (7.355,1);
      \node [anchor=north] at (5,2.3) {\footnotesize$D=1$};
        \node [anchor=north] at (5.05,2) {\footnotesize$D_\delta=1$};
            \node [anchor=north] at (6.7,1.7) {\footnotesize$D=0$};
                        \node [anchor=north] at (6.75,1.4) {\footnotesize$D_\delta=1$};

        \node [anchor=north]  at (6.25,0.9) {\footnotesize$D=0$};
                \node [anchor=north]  at (6.3,0.63) {\footnotesize$D_\delta=0$};

\end{tikzpicture}
\caption{\emph{A counterfactual policy where $D_\delta-D\geq 0$.}}
\label{figure_introduction}
\end{figure}

\begin{assumption}[Counterfactual Policies]\label{assumption_policy}
The collection of policies $\mathcal D$ satisfies
\begin{enumerate}
\item $\Pr(D_\delta=1)=p+\delta$ for $\delta\in[0,1-p)$ and $D_\delta\in\mathcal D$;
\item Monotonicity: $D_\delta-D\geq 0$;
\end{enumerate}
\end{assumption}

The monotonicity assumption $D_\delta-D\geq 0$ turns out to be key for the construction of our main tool for the sensitivity analysis, namely, the quantile breakdown frontier introduced below. In the context of the empirical application, where $D$ is union status, and $D=1$ denotes unionized individuals, monotonicity requires that we increase the unionization rate by unionizing previously nonunionized workers. Relaxing the assumption results not only in a much more complicated expression for the breakdown analysis, but also would entail a policy that unionizes and deunionizes different workers. This is, perhaps, impractical.

Another way to look at the monotonicity assumption is by inspecting the joint distribution of $D$ and $D_\delta$ it induces:
\begin{equation*}
\centering
\begin{tabular}{l|*{2}{c}r}
              & $D_{\delta }=0$ & $D_{\delta }=1$ \\
\hline
$D=0$ & $1-p-\delta$& $\delta$ \\
$D=1$            & $0$ &  $p$ \\
\end{tabular}
\end{equation*}

In other words, Assumption \ref{assumption_policy} rules out the presence of \emph{newly untreated} individuals, which would be the newly deunionized workers in the context of the empirical application. Also, in the limit, when $\delta=0$, we return to the original distribution of individuals. We will evaluate the effect of a counterfactual policy with two parameters: the global and the marginal effects. Let $F_Y^{-1}(\tau)$ and $F_{Y_{D_\delta}}^{-1}(\tau)$ denote the $\tau$-quantiles of $Y$ and $Y_{D_\delta}$ respectively.
\begin{definition}[Global and Marginal Effects]
For a given collection of policies $\mathcal D$, the unconditional global effect at the $\tau$-quantile of $D_\delta\in\mathcal D$ is
\begin{align*}
G_{\tau, D_\delta}&:=F_{Y_{D_\delta}}^{-1}(\tau)-F_{Y}^{-1}(\tau),
\end{align*}
and the unconditional marginal effect at the $\tau$-quantile is
\begin{equation*}
M_{\tau,\mathcal D}:=\lim_{\delta\to 0}\frac{F_{Y_{D_\delta}}^{-1}(\tau)-F_{Y}^{-1}(\tau)}{\delta}
\end{equation*}
whenever this limit exists.
\end{definition}

The global effect $G_{\tau, D_\delta}$ is the comparison of quantiles of the counterfactual distribution vs. the observed distribution for a fixed policy $D_\delta$. Naturally, for a collection $\mathcal D$, we have a corresponding collection on global effects. The marginal effect $M_{\tau,\mathcal D}$ can be interpreted as an ordinary derivative: for small $\delta$, it provides an approximation to the direction of the change in a given $\tau$-quantile. The main text will focus on the global effect, while the marginal effect is treated in detail in Appendix \ref{app_maginal_effect}.

\begin{remark}[\cite{Firpo2009}]
The marginal effect $M_{\tau,\mathcal D}$ was originally studied by \cite{Firpo2009}. Instead of Assumption \ref{assumption_policy}, \cite{Firpo2009} assume a form of distributional invariance: $F_{Y_{D_\delta}|D_\delta=d}=F_{Y|D=d}$ and obtain point identification. See the proof to Corollary 3 of the working paper version \cite{Firpo2007}. When both $D$ and $D_\delta$ are independent of $U$ and $X$, then distributional invariance will be satisfied. In this particular case, a policy maker can randomize $D_{\delta}$  so that for a given $\delta$, a fraction $p+\delta$ of individuals is randomly assigned to treatment. However, if we allow for $D$ to be endogenous, and if, as is usually the case, the structural form of endogeneity is unknown, then it may be impossible for the policy maker to design a sequence $\mathcal D$, such that for every $D_\delta\in\mathcal D$, $F_{Y_{D_\delta}|D_\delta=d}$ ``matches'' $F_{Y|D=d}$. From the point of view of the policy maker, this is a significant restriction on the types of counterfactual policies they can consider.
\end{remark}

\begin{remark}[Policy Relevant Treatment Effects]
\cite{heckman_prte, Heckman2005} and \cite{Carneiro2010,Carneiro2011} investigate the effect on the unconditional mean of the outcome. Using our notation, the Policy Relevant Treatment Effect (PRTE) of \cite{heckman_prte, Heckman2005} is 
\begin{equation*}
\text{\textrm{PRTE}}_{D_\delta }=\frac{E(Y_{D_\delta })-E(Y)}{\delta}
\end{equation*}%
and taking the limit $\delta \rightarrow 0$ yields the Marginal PRTE (MPRTE) of \cite{Carneiro2010,Carneiro2011}:%
\begin{equation*}
\mathrm{MPRTE_\mathcal D}=\lim_{\delta \rightarrow 0}\text{\textrm{PRTE}}_{\delta }.
\end{equation*}
\cite{yixiao2020} show how to generalize the MPRTE to cover the case of \cite{Firpo2009} as well.
\end{remark}

\begin{remark}[\cite{Rothe2012}]
\cite{Rothe2012} also studies the global and marginal effects but under a different identifying assumption, namely a form of conditional exogeneity. This assumption also yields an identified set. Let the outcome be $Y=h(D,X,U)$. For uniformly distributed random variables $\tilde U_1$ and $\tilde U_2$, the outcome can be represented as $Y=h(Q_D(\tilde U_1),Q_X(\tilde U_2),U)$ where $Q_D$ and $Q_X$ are the quantile functions. Then $Q_D$ is changed to another quantile function $Q_D^*$, generating a counterfactual distribution, which is identified when $\tilde U_1\perp U\| X$ and $D$ is continuous. When $D$ is discrete, $\tilde U_1$ is not uniquely determined, so that a range of possible counterfactual distributions is possible resulting in partial identification.
\end{remark}

The next task is to define \emph{who} are the \emph{newly treated} individuals, that is, how does $D_\delta$ determine who receives treatment among the individuals whose $D=0$? In this paper we will focus on two types of policies: a policy that simply chooses individuals whose $D=0$ at random and assigns them to $D_\delta=1$, and a policy that chooses individuals based on a user-specified criterion. We will refer to these two types of policies as \emph{randomized policy} and \emph{non-randomized policy} respectively. 

\begin{example}[Randomized policy]\label{example_rand_policy}
A randomized policy satisfies: for any $\delta\in[0,1-p)$
\begin{align*}
D_\delta = 
\left\{
\begin{aligned}
1 &\quad \text{if } D = 1 \\
0 \text{ or } 1 &\quad \text{if } D = 0
\end{aligned}
\right.
\end{align*}
and the \emph{newly treated} are selected at random. Using the conditional independence notation %
we write $D_\delta\perp Y(1), Y(0)\|D=0$: 
\begin{align}\label{eqn_random_cia}
Pr(D_\delta=1|D=0)=\Pr(D_\delta=1|D=0, Y(1), Y(0))=\frac{\delta}{1-p}.
\end{align}
\end{example}

\begin{example}[Non-randomized policy]\label{example_non_rand_policy}
An example of a non-randomized policy is the following: for any $\delta\in[0,1-p)$
\begin{align}\label{eqn_non_random}
D_\delta = 
\left\{
\begin{aligned}
1 &\quad \text{if } D = 1 \\
1 &\quad \text{if } D = 0 \text{ and } Z \le F^{-1}_{Z|D=0}\left( \frac{\delta}{1 - p} \right) \\
0 &\quad \text{otherwise}
\end{aligned}
\right.
\end{align}

for some observable random variable $Z$. In this case, the individuals in the group $\left\{D=0\right\}$ whose $Z$ is less than the $\frac{\delta}{1-p}$-quantile of this group are shifted to $D_\delta=1$. This rule guarantees that, in expectation, a proportion $\delta$ of individuals is shifted.
\end{example}

For a collection of policies $\mathcal D$ that satisfies Assumptions \ref{assumption_policy}, the counterfactual distribution $F_{Y_{D_\delta}}(y)$ can be decomposed as
\begin{align}\label{count_dist_dec}
F_{Y_{D_\delta}}(y) &= pF_{Y|D=1}(y)  + (1-p-\delta) F_{Y|D_\delta=0}(y)  +\delta F_{Y(1)|D=0,D_\delta=1}(y)
\end{align}
for each $D_\delta$. Here, $F_{Y(1)|D=0,D_\delta=1}(y)$ corresponds to the distribution of the \emph{newly treated}. This distribution cannot be identified from the data because it requires observing $Y(1)$ for a subpopulation for which we only observe their $Y(0)$. Consequently, $F_{Y_{D_\delta}}(y)$ is not identified either. The goal is to bound the quantiles of $F_{Y_{D_\delta}}(y)$. To that end, we make the following regularity assumptions.

\begin{assumption}[Regularity Assumptions]\label{assumption_regularity}

\begin{enumerate}
\item\label{assumption_yd_regular} For $d=0,1$, $F_{Y|D=d,X=x}(y)$ is continuous and strictly increasing for all $y$ such that $0<F_{Y|D=d,X=x}(y)<1$, and for all $x\in\mathcal X$.
\item\label{assumption_yd_regular_2} For every $D_\delta\in\mathcal D$, $F_{Y|D_\delta=0}(y)$  is continuous and strictly increasing for all $y$ such that $0<F_{Y|D_\delta=0}(y)<1$.
\item For every $D_\delta\in\mathcal D$, $F_{Y(1)|D=0,D_\delta=1}(y)$  is continuous and strictly increasing for all $y$ such that $0<F_{Y(1)|D=0,D_\delta=1}(y)<1$.
\item\label{assumption_x_support} For every $D_\delta\in\mathcal D$, the support of $X$ conditional on $D=0$ and $D_\delta=1$ is included in the support of $X$ conditional on $D=1$.
\end{enumerate}
\end{assumption}

The next assumption is our main working assumption to obtain the bounds on the quantiles of $F_{Y_{D_\delta}}(y)$.
\begin{assumption}[KS-distance]\label{assumption_ks_distance}
For a given $D_\delta\in\mathcal D$, there exists a known $c\in[0,1]$ such that
\begin{align}\label{eq:ks_distance}
\sup_{y\in\mathbb R}\left |F_{Y(1)|D=0,D_\delta=1}(y)-\int_{\mathcal X}F_{Y|D=1,X=x}(y)dF_{X|D=0,D_\delta=1}(x)\right |\leq c
\end{align}
\end{assumption}

We refer to the left hand side of \eqref{eq:ks_distance} as the \emph{policy selection bias}. The idea is that in the absence of policy selection bias, we can take $c=0$ and $F_{Y(1)|D=0,D_\delta=1}(y)$ can be recovered by matching \emph{already} unionized individuals with \emph{newly} unionized individuals and integrating against the characteristics of the \emph{newly} unionized individuals, using $F_{X|D=0,D_\delta=1}(x)$. This is akin to a \textit{joint} unconfoundedness assumption: $D\perp Y(1)\| X$ \textit{and} $D_\delta\perp Y(1)\| X$.\footnote{If $D, D_\delta\perp Y(1)\| X$, then $F_{Y|D=1, X=x}(y)  = F_{Y(1)|D=0,D_\delta = 1, X=x}(y)$,
so that integrating against $F_{X|D=0,D_\delta=1}(x)$ yields $F_{Y(1)|D=0,D_\delta=1}(y) $.} By allowing $c$ to be non-zero, we are relaxing this particular type of conditional independence assumption (\cite{Masten2018}). Here, Assumption \ref{assumption_regularity}.\ref{assumption_x_support} becomes relevant.
\cite{Kline2013} also use the Kolmogorov-Smirnov distance to bound the quantiles of the outcome to allow for the possibility that data might not be missing at random.

\begin{remark}
For now, we take $c$ as known. In the next section it will be the quantity on which we will base the sensitivity analysis.
\end{remark}

\begin{theorem}[Bounds on the Global Effect]\label{thm_bounds_global}
If a policy $D_\delta$ satisfies assumptions \ref{assumption_policy}, \ref{assumption_regularity}, and \ref{assumption_ks_distance}, then for $\tau\in(\delta,1-\delta)$
\begin{align*}
\max\left\{\tilde F_A^{-1}(\tau-\delta) ,F_A^{-1}(\tau-\delta c)\right\} - F_{Y}^{-1}(\tau) \leq G_{\tau, D_\delta}\leq \min\left\{\tilde F_A^{-1}(\tau) ,F_A^{-1}(\tau+\delta c)\right\}-F_{Y}^{-1}(\tau)
\end{align*}
where $\tilde F_A(y):= p F_{Y|D=1}(y)  + (1-p-\delta) F_{Y|D_\delta=0}(y),$ and %
\begin{align*}
F_A(y)&:=\tilde F_A(y)  +\delta  \int_{\mathcal X}F_{Y|D=1,X=x}(y)dF_{X|D=0,D_\delta=1}(x).
\end{align*}
For a fixed $\tau\in(\delta,1-\delta)$, these bounds are sharp.
\end{theorem}

\begin{remark}
The bounds are sharp in the sense that for any possible value of the global effect within the bounds, we can construct a corresponding newly treated distribution $F_{Y(1)|D=0,D_\delta=1}(y)$ which delivers precisely that same global effect. This construction depends on $\tau$, implying that, rather than uniform, sharpness is pointwise in $\tau\in(\delta,1-\delta)$.
\end{remark}

\begin{remark}
When $c=0$, $G_{\tau, D_\delta}$ is point identified because the counterfactual distribution in \eqref{count_dist_dec} becomes $F_A(y)$, which is identified. We refer to $F_A(y)$ as the apparent distribution, hence the subscript $A$, because it is the distribution that ``appears'' to be the counterfactual distribution. Naturally, for $c=0$, we have that the lower and upper bound are identical, and equal to $F_A^{-1}(\tau)$.\footnote{\label{footnote}
For $c=0$, the upper bound is $F_A^{-1}(\tau)$ because $F_A^{-1}(\tau)\leq \tilde F_A^{-1}(\tau) $. For the lower bound we need to show that $\max\left\{\tilde F_A^{-1}(\tau-\delta),F_A^{-1}(\tau)\right\}=F_A^{-1}(\tau)$. Note that $\tilde F_A^{-1}(\tau-\delta)$ satisfies:
\begin{align*}
p F_{Y|D=1}(\tilde F_A^{-1}(\tau-\delta))  + (1-p-\delta) F_{Y|D_\delta=0}(\tilde F_A^{-1}(\tau-\delta)) = \tau-\delta
\end{align*}
On the other hand, $F_A^{-1}(\tau)$ when $c=0$ satisfies
\begin{align*}
pF_{Y|D=1}(F_A^{-1}(\tau))  + (1-p-\delta) F_{Y|D_\delta=0}(F_A^{-1}(\tau))  +\delta  \int_{\mathcal X}F_{Y|D=1,X=x}(F_A^{-1}(\tau))dF_{X|D=0,D_\delta=1}(x)=\tau
\end{align*} 
Comparing these two expression, we obtain
\begin{align*}
p F_{Y|D=1}(\tilde F_A^{-1}(\tau-\delta))  + (1-p-\delta) F_{Y|D_\delta=0}(\tilde F_A^{-1}(\tau-\delta)) \leq 
pF_{Y|D=1}(F_A^{-1}(\tau))  + (1-p-\delta) F_{Y|D_\delta=0}(F_A^{-1}(\tau))  
\end{align*}
which implies that we must have $\tilde F_A^{-1}(\tau-\delta)\leq F_A^{-1}(\tau)$. 
} Thus, for $c=0$, $G_{\tau, D_\delta}= F_A^{-1}(\tau)-F_{Y}^{-1}(\tau)$.
\end{remark}

\begin{remark}
By construction, we have that $\tilde F_A(y)\leq F_A(y)$, which implies, for $\tau\in(\delta,1-\delta)$, that $ F_A^{-1}(\tau)\leq \tilde F_A^{-1}(\tau)$. Therefore, whether $F_A$ or $\tilde F_A$ bounds the global effect depends on $\tau$, $\delta$, and $c$. In other words, the $\min$ and $\max$ are not redundant.
\end{remark}

\begin{remark}\label{remark_tau}
The range of $\tau$ is restricted to $(\delta,1-\delta)$ in order to ensure that $c$ to plays a role in the bounds. For $\tau$ outside this range, the bounds might not depend on $c$, if $c$ is close enough to 1.
\end{remark}

\section{Quantile Breakdown Frontier}

The quantile breakdown frontier is a curve that allows to perform to a sensitivity analysis with respect to $c$, the policy selection bias. Suppose we are interested in a target conclusion $G_{\tau,D_\delta}\geq g_{\tau,L} $ for some $g_{\tau,L}$. If the conclusion does not hold when $c=0$, then there is no point in performing a sensitivity analysis. On the other hand, if the conclusion \textit{does} hold under $c=0$, we would like to know the maximum amount of $c$ such that the conclusion continues to hold. This is the sensitivity analysis. 
The quantile breakdown frontier for the global effect tackles both of these issues. The frontier is the map
\begin{align}\label{qbf_def}
\tau\mapsto c_{\tau,L}= \frac{\tau -F_A\left ( F_{Y}^{-1}(\tau)+ g_{\tau,L} \right )}{\delta}
\end{align}
for $\tau\in(\delta,1-\delta)$.\footnote{We can also look at conclusions of the type $G_{\tau,D_\delta}\leq g_{\tau,U}$. In this case, the quantile breakdown frontier is the map
\begin{align*}
\tau\mapsto c_{\tau,U}= \frac{F_A\left ( F_{Y}^{-1}(\tau)+ g_{\tau,U} \right )-\tau}{\delta}.
\end{align*}
} 
An explanation of the derivation is in section \ref{app_frontier_der} of the appendix. When $c_{\tau,L}<0$, the desired target conclusion does not hold. If $c_{\tau,L}>0,$ then for $c\leq  c_{\tau,L}$ then target conclusion holds under point identification. If $c> c_{\tau,L}>0$, then the target conclusion might not hold. In this sense, when $c_{\tau,L}>0,$ the frontier provides an amount of policy selection bias which is compatible with the conclusion.

The next lemma contains some properties of the quantile breakdown frontier.

\begin{lemma}\label{qbf_global_lemma}
Under the assumptions of Theorem \ref{thm_bounds_global}, the quantile breakdown frontier for $G_{\tau,D_\delta}$ defined in \eqref{qbf_def} satisfies the following: (i) if $\tau\mapsto g_{\tau,L}$ is continuous, then $\tau\mapsto c_{\tau,L}$ is continuous; (ii) if $c_{\tau,L}<0$, then the conclusion does not hold for $c=0$; (iii) if $c_{\tau,L}\geq 0$, then for any $c\geq 0$ such that $c\leq  c_{\tau,L}$, the conclusion holds, and for $c> c_{\tau,L}$ the conclusion might not hold; and (iv) if $c_{\tau,L}\geq 0$ and $\tilde F_A^{-1}(\tau-\delta)-F_Y^{-1}(\tau)\geq g_{\tau,L} $, then the conclusion holds for any $c\in[0,1]$.
\end{lemma}

Part $(i)$ of the lemma puts a restriction on the types of conclusion by requiring continuity of the family of target conclusions. This is not essential, but illustrates the fact that the smoothness of the frontier depends, among other things, on the conclusions. For notational simplicity, later on we assume $g_{\tau,L}\equiv g$. Part $(iv)$ addresses a potential ``conservativeness'' in the breakdown analysis. This stems from the fact there is a possibility that the lower bound for the global effect does not depend on $c$. In the notation of Theorem \ref{thm_bounds_global}, this means that we cannot rule out $\max\left\{\tilde F_A^{-1}(\tau-\delta) ,F_A^{-1}(\tau-\delta c_{\tau,L})\right\} = \tilde F_A^{-1}(\tau-\delta).$ This is related to part $(iii)$. Indeed, if the opposite is true, $\tilde F_A^{-1}(\tau-\delta)< F_A^{-1}(\tau-\delta c_{\tau,L})$, then we can modify the language of part $(iii)$
 to say that for $c> c_{\tau,L}$ the conclusion \textit{will} not hold.\footnote{In the empirical application we test for $\tilde F_A^{-1}(\tau-\delta)-F_Y^{-1}(\tau) < g_{\tau,L}$.} A more in-depth discussion of parts $(ii)$, $(iii)$ and $(iv)$ of Lemma \ref{qbf_global_lemma} can be found in Appendix \ref{understanding_qbf_global_lemma}.

 \subsection{Some comparative statics}
 
 It is possible to examine the behavior of $c_{\tau,L}$ with respect to $g_{\tau,L}$. Differentiating \eqref{qbf_def}, we get, for a fixed $\tau$:
\begin{align}\label{der_qbf_g}
\frac{\partial c_{\tau,L}}{\partial g_{\tau,L}}= -\frac{f_A\left ( F_{Y}^{-1}(\tau)+ g_{\tau,L} \right )}{\delta}<0.
\end{align} 
provided $\delta>0$. This means that if the conclusion is less stringent, \textit{i.e.}, $g_{\tau,L}$ is reduced, then the frontier moves upwards, indicating that more policy selection bias can be allowed before the conclusion breaks down. In this sense, $c_{\tau,L}$ provides a minimum amount of policy selection bias for all $g$ such that $g\leq g_{\tau,L}$.

Now consider the derivative of \eqref{qbf_def} with respect to $\tau$:
\begin{align*}
\frac{\partial c_{\tau,L}}{\partial \tau}= \frac{1}{\delta}-\frac{f_A\left ( F_{Y}^{-1}(\tau)+ g_{\tau,L} \right )}{\delta}\left(\frac{\partial F_{Y}^{-1}(\tau)}{\partial \tau} +  \frac{\partial g_{\tau,L}}{\partial \tau} \right)
\end{align*} 
It is not possible to sign the slope of the frontier \textit{a priori}. It can depend largely on $\tau$, that is, whether we are the center or the tails of the distribution of $Y$, and whether the outcome is bounded or not, among other reasons. 

The derivative of \eqref{qbf_def} with respect to $\delta$ is more complicated due to the fact that $F_A$ depends on $\delta$. A general expression is the following:
\begin{align*}
\frac{\partial c_{\tau,L}}{\partial \delta}= -\frac{c_{\tau,L}}{\delta}-\frac{1}{\delta}\frac{\partial F_A\left ( F_{Y}^{-1}(\tau)+ g_{\tau,L} \right )}{\partial \delta}
\end{align*} 
where, following the definition of $F_A ( y)$ given in \eqref{complete_apparent}, we get (heuristically)
\begin{align*}
\frac{\partial F_A ( y)}{\partial \delta} &=  \int_{\mathcal X}F_{Y|D=1,X=x}(y)dF_{X|D=0,D_\delta=1}(x)-F_{Y|D_\delta=0}(y)\\& +  \delta \int_{\mathcal X}F_{Y|D=1,X=x}(y)\frac{\partial f_{X|D=0,D_\delta=1}(x)}{\partial \delta} dx - \delta \frac{\partial F_{Y|D_\delta=0}(y)}{\partial \delta}.
\end{align*}

\subsection{Bounds derived from the QBF}
Suppose that a policy maker choose a particular quantile $\tau^*$ and a target conclusion $g_{\tau^*,L}$. Then, provided that $c_{\tau^*,L}\in [0,1]$, we can obtain bounds on the global effect other quantiles $\tau\neq \tau^*$. The interpretation is that, as long as the policy selection bias satisfies $c\leq c_{\tau^*,L}$, the global effect will be bounded across quantiles by the bounds of Theorem \ref{thm_bounds_global} evaluated at $c_{\tau^*,L}$. In a sense, this allows us to extend the sensitivity analysis to other quantiles. That is for $\tau\in(\delta,1-\delta)$, we have
\begin{align}\label{eq:bounds_derived}
\max\left\{\tilde F_A^{-1}(\tau-\delta) ,F_A^{-1}(\tau-\delta c_{\tau^*,L})\right\} - F_{Y}^{-1}(\tau) &\leq G_{\tau, D_\delta}\notag\\
&\leq \min\left\{\tilde F_A^{-1}(\tau) ,F_A^{-1}(\tau+\delta c_{\tau^*,L})\right\}-F_{Y}^{-1}(\tau).
\end{align}

\section{Estimation and Inference}

We work in the space $\ell^{\infty}(\delta,1-\delta)$ of bounded real-valued functions defined on $(\delta,1-\delta)$. As usual, we endow this space with the supremum norm: $\|x\|_\infty:=\sup_{t\in(\delta,1-\delta)}|x(t)|$.\footnote{The reason we restrict the space to be $\ell^{\infty}(\delta,1-\delta)$ and not $\ell^{\infty}(0,1)$ is due to the fact that for a given $\delta$, we cannot reach quantiles below $\delta$ or above $1-\delta$. See Remark \ref{remark_tau} above.} In order to simplify notation, and ensure the continuity of the quantile breakdown frontier, we are going to focus on the case where the threshold is constant across $\tau$.

\begin{assumption}[Constant Threshold]\label{uniform_threshold}
For some scalar $g$, the threshold $g_{\tau,L}$ satisfies $g_{\tau,L}=g$ for any $\tau\in(\delta,1-\delta)$. 
\end{assumption}

This assumption can be relaxed at the expense of more complicated notation. However, we still require smoothness in the map $\tau\mapsto g_\tau.$ For the case of the quantile breakdown for the sign of the marginal effect, we will set $g=0$. 

\begin{assumption}[DGP]\label{dgp}
We observe an i.i.d. sample $\{ Y_i, D_i, X_i \}_{i=1}^n$, where the support of $X$ is finite.
\end{assumption}

\subsection{Global Effect - Estimation}
To estimate $ c_{\tau,L}$ given in \eqref{qbf_def} we need to specify two components: $(i)$ $\delta$, the increase in the treated proportion, and $(ii)$
$D_{i,\delta}$, the counterfactual treatment assignment, as in examples \ref{example_rand_policy} and \ref{example_non_rand_policy}. Once this has been done, 
we use sample analogs. The estimator of the quantile breakdown frontier for the global effect is
\begin{align}\label{qbf_estimation}
\hat c_{\tau,L}= \frac{\tau - \hat F_A\left ( \hat F_{Y}^{-1}(\tau)+ g \right )}{\delta},
\end{align}
where $ \hat F_{Y}^{-1}(\tau)$ is the empirical $\tau$-quantile of $Y$: $\hat F_Y^{-1}(\tau) :=\inf \left\{ y:\hat F_Y(y)\geq \tau  \right\}$, and $\hat F_A(y)$ is the empirical counterpart of $F_A(y)$ given in Theorem \ref{thm_bounds_global}. Detailed expressions are provided in Appendix \ref{appendix_notation}. 
This is similar to a quantile-quantile transformation (see Exercise 4 in Chapter 3.9 in \cite{vandervaart1996}). We base our proof of the asymptotic distribution of $\sqrt n(\hat c_{\tau,L}-c_{\tau,L})$ on the proof of Lemma A.1 in \cite{Beare2019}.\footnote{\cite{Beare2019} also offer some interesting historical context for the result.} We view the map $\tau\mapsto\hat c_{\tau,L}$ as a random element of $\ell^{\infty}(\delta,1-\delta)$. In that case, we denote it simply by $\hat c_{L}.$

The main assumption is the following.
\begin{assumption}[Functional CLT]\label{brownian_bridge}
The following multivariate functional central limit theorem holds 
\begin{align*} 
\sqrt n\begin{pmatrix}
\hat F_{Y} - F_Y\\
\hat F_{A}-  F_{A}
\end{pmatrix} \rightsquigarrow \begin{pmatrix}
\mathbb G_Y\\ 
\mathbb G_{A}
\end{pmatrix},
\end{align*}
where $\mathbb G_Y$ and $\mathbb G_A$ are Brownian bridges in $\ell^{\infty}(\mathcal Y)$.
\end{assumption}

\begin{remark}
Appendix \ref{app_primitive} lists primitive assumptions in terms of the CDFs that comprise $F_A$ that lead to Assumption \ref{brownian_bridge}. 
\end{remark}

The following assumption is needed to establish the Hadamard differentiable of different functions used in the construction of $\theta$.
\begin{assumption}[Conditions for Hadamard Differentiability]\label{ass_had_dif}
\item
\begin{enumerate}
\item \label{ass_f_y_quantiles}For some $\varepsilon>0$, $F_Y$ is continuously differentiable in $[F_Y^{-1}(\delta)-\varepsilon,F_Y^{-1}(1-\delta)+\varepsilon]\subset \mathcal Y$ with strictly positive derivative $f_Y$.
\item  \label{ass_f_a_uniform}$y\mapsto F_A(y)$ is differentiable, with uniformly continuous and bounded derivatives.
\end{enumerate}
\end{assumption}

The first item in Assumption \ref{ass_had_dif} concerns the support $\mathcal Y$ and the smoothness of $F_Y$. It is used to guarantee the Hadamard differentiability of the quantile process $\tau\mapsto F_Y^{-1}(\tau)$ for $\tau\in(\delta,1-\delta)$. The second item ensures that $F_A(y)$ has a uniformly continuous and bounded derivative which we denote by $f_A(y)$. It is needed to establish the Hadamard differentiability of the composition map $(F_A,F_Y^{-1})\mapsto F_A\circ (F_Y^{-1}+g)$.\footnote{Section 3.9 in \cite{vandervaart1996} studies the Hadamard differentiability of composition maps.}

\begin{theorem}\label{dist_theta}
Under Assumptions \ref{uniform_threshold}, \ref{dgp}, \ref{brownian_bridge}, and \ref{ass_had_dif}
\begin{align*}
\sqrt n\begin{pmatrix}
\hat c_{L}- c_{L}
\end{pmatrix} \rightsquigarrow 
\frac{1}{\delta}\mathbb G_A\circ (F_Y^{-1}+g)-\frac{1}{\delta}f_A\circ (F_Y^{-1}+g) \frac{\mathbb G_Y\circ F_Y^{-1}}{f_Y\circ F_Y^{-1}},
\end{align*}
a tight Gaussian element in $\ell^{\infty}(\delta,1-\delta)$.
\end{theorem}

\subsection{Global Effect - Inference}
We propose a two-step inference procedure to accommodate the results of Lemma \ref{qbf_global_lemma}. First, we test the following null hypothesis: $   H_{1,0}: \tilde F_A^{-1}(\tau-\delta)-F_Y^{-1}(\tau)< g,$
against the alternative $ H_{1,a}: \tilde F_A^{-1}(\tau-\delta)-F_Y^{-1}(\tau)\geq g.$
If the null $H_{1,0}$ is not rejected, we proceed to test: $H_{2,0}: c_{\tau,L} = \bar c$, against the alternative $ H_{2,a}: c_{\tau,L} \neq \bar c,$ for some user-specified $\bar c$.

Let $\alpha_1$ denote the \emph{size} of the first test, and let $\alpha_2$ denote the \emph{conditional size} of the second test. The family-wise error rate (FWER), defined as the probability of making at least one false rejection among all true null hypotheses, is $ \alpha_1 + (1-\alpha_1)\alpha_2$. For example, to keep the FWER at $5\%$, then setting $\alpha_1=0.025$ yields $\alpha_2\le(0.05-\alpha_1)/(1-\alpha_1)\approx 0.02564$.

The first null hypothesis requires estimation of $d_{\tau}:=\tilde F_A^{-1}(\tau-\delta)-F_Y^{-1}(\tau)$. Estimation follows by resorting to the sample analogs: $\hat d_{\tau}=\hat{\tilde F}_A^{-1}(\tau-\delta)-\hat F_Y^{-1}(\tau)$.
Details can be found in Appendix \ref{appendix_d_tau}.

Both tests are pointwise in $\tau$. They are carried out at a single quantile level $\tau=\tau^*$, and the critical values are computed at that level alone. In the empirical application, we set $\tau^*=0.3$.
For the first test, the statistic is $T_{1,n}=\sqrt n\left(\hat d_{\tau^*}-g\right)$, and $H_{1,0}$ is rejected when $T_{1,n}>\xi_{1-\alpha_1}$, where $\xi_{1-\alpha_1}$ is the $(1-\alpha_1)$ quantile of the bootstrap draws $\sqrt n\left(\hat d^{\,b}_{\tau^*}-\hat d_{\tau^*}\right)$, $b=1,\ldots,B$. For the second test, the statistic is $T_{2,n}=\sqrt n\left(\hat c_{\tau^*,L}-\bar c\right)$, and $H_{2,0}$ is rejected when $|T_{2,n}|>\eta_{1-\alpha_2}$, where $\eta_{1-\alpha_2}$ is the $(1-\alpha_2)$ quantile of $\left|\sqrt n\left(\hat c^{\,b}_{\tau^*,L}-\hat c_{\tau^*,L}\right)\right|$. Inverting the second test gives the confidence interval $\hat c_{\tau^*,L}\pm\eta_{1-\alpha_2}/\sqrt n$ reported in the empirical application, Section \ref{section_empirical}.\footnote{The standard bootstrap is valid for both tests because $\hat d_{\tau}$ and $\hat c_{\tau,L}$ have Gaussian limits as established in Appendix \ref{appendix_d_tau} and in  Theorem \ref{dist_theta} respectively.}

\begin{remark}\label{remark_grid}
While the
\emph{population} frontier, as a function of $\tau$, is smooth under Assumption \ref{ass_had_dif},
the estimated counterpart, which is built from empirical distribution functions, need not inherit
that property in any finite sample. For this reason the quantitative statements in Section
\ref{section_empirical} are made pointwise, only at the quantile levels at which the frontier is actually
evaluated. The levels are marked in Figures \ref{empirical_qbf} and \ref{empirical_qbf_marg},
and the segments connecting them are drawn for legibility rather than being estimates. The bounds in Figure
\ref{bounds_qbf} are the exception, being reported with uniform bands.
\end{remark}

\subsection{Bounds derived from the QBF}
To estimate the bounds derived from the QBF for the global effect, we use the sample counterpart of \eqref{eq:bounds_derived}. Besides the target conclusion $g$ for the global effect, we need to supply, a quantile level $\tau^*$ of interest. For 
$\tau\in(\delta,1-\delta)$, the estimator of the derived bounds is
\begin{align*}
\max\left\{\hat{\tilde F}_A^{-1}(\tau-\delta) ,\hat F_A^{-1}(\tau-\delta \hat {\tilde c}_{\tau^*,L})\right\} - \hat F_{Y}^{-1}(\tau) &\leq G_{\tau, D_\delta}\\
&\leq \min\left\{\hat{\tilde F}_A^{-1}(\tau) ,\hat F_A^{-1}(\tau+\delta \hat {\tilde c}_{\tau^*,L})\right\}- \hat F_{Y}^{-1}(\tau),
\end{align*}
where 
\begin{align}
\hat {\tilde c}_{\tau^*,L} =  \max \left \{  \min \left \{  \hat c_{\tau^*,L}, 1\right\},0 \right\}.
\end{align}
Here, $\hat c_{\tau^*,L}$ is the estimator given in \eqref{qbf_estimation} evaluated $\tau=\tau^*$. The asymptotic distribution of the bounds cannot be Gaussian due to being the composition of maps which are not Hadamard fully differentiable--only directional. 
Hence, by Corollary 3.1 in \cite{Fang2019}, the standard bootstrap will fail. This means that if we attempt to construct confidence intervals in the usual way by resampling, we will not obtain correct asymptotic coverage. Instead, we use the numerical delta method of \cite{Hong2018}. A detailed analysis of the asymptotic distribution of the bounds is in Appendix \ref{distribution_of_bounds_appendix}.

\section{Empirical application: What do unions do?}\label{section_empirical}

There is an extensive literature that studies unions and inequality. A recent contribution by \cite{Farber2020} contains a review of the literature. In our empirical application, in particular, we look at how unions affect the distribution of wages for \emph{all} workers. Unions can have a variety of effects on the distribution of wages. As argued by \cite{Freeman1980}, unions can raise the wages of unionized workers relative to non-unionized workers, possibly through more bargaining power. So, if higher paid workers unionize, the dispersion of wages can increase, but if lower paid workers unionize, the dispersion of wages can decrease. Furthermore, within a given industry, the union can reduce the dispersion of wages by standardizing the wages. This will impact the distribution of wages more or less depending on the size of the industry and the wages it pays. 

A key difficulty in identifying the causal effect of unions on wages is that selection into unions is non-random. Hence, any measurement of the union premium--the difference in wages between similar union and nonunion workers--will be biased for the causal effect. Indeed, this has been a long standing concern of labor economists.\footnote{Indeed, the opening words of \cite{Card1996} are: \begin{quotation}
\emph{Despite a large and sophisticated literature there is still substantial disagreement over the extent to which differences in the structure of wages between union and nonunion workers represent an effect of trade unions, rather than a consequence of the nonrandom selection of unionized workers.}
\end{quotation}} With respect to selection into unions, \cite{Card1996} argues that unionized workers with low observed skills, tend to have high unobserved skills. The reverse happens with high skilled unionized workers: they tend to have low unobservable skills. Due to this selection bias, it might be impossible for a policy maker to devise a policy where the \emph{newly} unionized workers are selected in a way such that they are drawn from the distribution of the \emph{already} unionized workers.

Using the techniques developed in this paper, we are going to consider the effect of both globally and marginally expanding union coverage. We will explicitly allow for non-random selection into unions. %
This allows for the \emph{newly} unionized and \emph{already} unionized workers to be drawn from different distributions. We do not use any imputation method to impute the union premium of the newly unionized workers.

Following \cite{Freeman1980}, \cite{Card2001} and \cite{Card2004} we consider a two sector economy. Each worker has a well-defined pair of potential (log) wages: $Y_i(1)$ for the unionized sector and $Y_i(0)$ for the nonunionized sector. Under Assumption \ref{assumption_policy}, and for any policy $D_\delta$, we have the following classification of individuals:
\begin{equation*}
\centering
\begin{tabular}{l|*{2}{c}r}
              & $D_{\delta }=0$ & $D_{\delta }=1$ \\
\hline
$D=0$ & \emph{nonunionized} & \emph{newly unionized}  \\
$D=1$            & - & \emph{unionized}  \\
\end{tabular}
\end{equation*}

The relevant unobserved distribution is then $F_{Y(1)|D=0, D_\delta=1}$: the union wages of the newly unionized workers. Following \eqref{eq:ks_distance}, we look at departures of $F_{Y(1)|D=0, D_\delta=1}$ from
\begin{align*}
\int_{\mathcal X}F_{Y|D=1,X=x}(y)dF_{X|D=0,D_\delta=1}(x),
\end{align*}
which is observed. This difference is what we refer to as the policy selection bias.

Using the data in \cite{Firpo2009} we estimate the quantile breakdown frontier for marginal and global effects of different type of policies on the distribution of real log hourly wages. We use the 1983-1985 Outgoing Rotation Group (ORG) Supplement of the Current Population Survey. Our sample consists of 266,956 observations on U.S. males, of which $73.8\%$ are non-unionized, and $26.2\%$  are unionized. See \cite{Lemieux2006} for more details about the data. The covariates are: years of education, age, marital status, dummy for race (nonwhite), years of experience, and union status indicator. Table \ref{tab:wage_cov_diffs} contains sample means by union status, and the difference. All of the differences are significantly different from $0$ at the $5\%$ level, possibly reflecting selection into unions. 

\begin{table}[ht]
\centering
\begin{tabular}{lccc}
\hline
Variable & Non-unionized & Unionized & Difference \\
\hline
\hline
Log(wage) & 1.718 & 1.963 & 0.245 \\
 & (0.601) & (0.407) & [0.241, 0.249] \\[1ex]
Age (years) & 34.953 & 39.783 & 4.831 \\
 & (12.455) & (11.566) & [4.729, 4.933] \\[1ex]
Education (years) & 12.966 & 12.498 & -0.468 \\
 & (2.912) & (2.650) & [-0.491, -0.444] \\[1ex]
Experience (years) & 15.987 & 21.285 & 5.298 \\
 & (12.668) & (12.252) & [5.192, 5.405] \\[1ex]
Married (share) & 0.611 & 0.742 & 0.131 \\
 & (0.488) & (0.438) & [0.127, 0.135] \\[1ex]
Nonwhite (share) & 0.102 & 0.134 & 0.032 \\
 & (0.303) & (0.341) & [0.029, 0.035] \\[1ex]
Number of Observations & 197,115 & 69,841 & \\[2ex]
\hline
\hline
\end{tabular}
\caption{Sample means by union status. Standard deviations in parentheses, and $95\%$ CIs for the differences in brackets.}
\label{tab:wage_cov_diffs}
\end{table}

The unionization rate in the dataset is $0.26$. Figure \ref{hump_union_rates} shows the typical hump-shaped pattern of the unionization rates by quantiles of the distribution of wages. For lower quantiles, unionization rates are quite low. They peak in the past the middle of the distribution and then drop at the higher quantiles.

\begin{figure}
\centering
\includegraphics[scale=0.42]{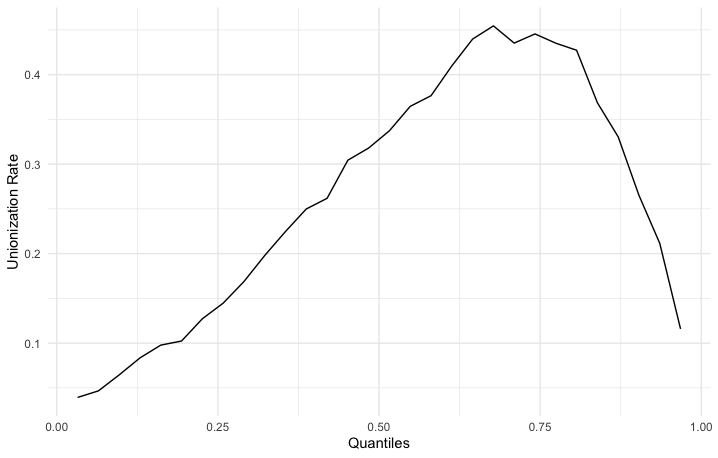}
\caption{\emph{Unionization rates by quantiles of the distribution of wages.}}
\label{hump_union_rates}
\end{figure}

Consider a policy that increase in the unionization rate by $10\%$. It consists of unionizing workers whose wages are below the $.10/(1-p)$-quantile $\approx 0.14$-quantile of the wages of the nonunionized sector.\footnote{This guarantees that the unionization rate increases by roughly 10\%. Indeed, the mean of $D_\delta$ is now $0.36$.} In the notation of this paper, we have $D=1$ if a worker is unionized, $D_\delta=1$ if a worker is unionized under the policy, $Y$ is (log) wage, and $\delta=0.1$. That is, $D_\delta$ is given by
\begin{align*}
D_\delta = 
\left\{
\begin{aligned}
1 &\quad \text{if } D = 1 \\
1 &\quad \text{if } D = 0 \text{ and } Y \le  F^{-1}_{Y|D=0}(0.14) \\
0 &\quad \text{otherwise}
\end{aligned}
\right.
\end{align*}

Figure \ref{empirical_qbf_many_2} shows a family of quantile breakdown frontiers for $g=0, 0.025, 0.05, 0.075, 0.1$, for a grid of $\tau\in (0.15,0.85)$. The top one corresponds to $g=0$. As $g$ increases, the frontiers shift down. This is in agreement with \eqref{der_qbf_g}, which states that, for a fixed $\tau$, the derivative of $c_{\tau,L}$ with respect to $g$ is negative.

Figure \ref{empirical_qbf} takes a closer look at the quantile breakdown frontier for $g=0.05$. Confidence intervals are obtained via 1,000 bootstrap replications. Since the outcome variable is log wages, this means approximately a $5\%$ increase in wages. First, we can see that for lower quantiles, the conclusion will hold as long as $c$ is below the frontier. On the other hand, for upper quantiles, the conclusion does not hold under point identification. Thus there is no point in doing a sensitivity analysis for upper quantiles.

Suppose the policy maker is interested in $\tau^*=.3$. This is indicated by the vertical dashed red line. The $.3$-quantile of the unconditional distribution of log wages is $1.465$. An increase of $0.05$ would bring the $.3$-quantile to $1.515$ which is close to the $.33$-quantile which is $1.521$. In this case, $\hat c_{\tau^*,L}=\hat c_{.3,L}=0.223$, and $[0.203,0.243]$ is the approximately $97.5\%$ pointwise confidence interval. Importantly, $\hat d_{\tau^*} = \hat d_{.3} = 0.00005$, which is far below $g=0.05$, so $H_{1,0}$ is not rejected; in fact, as shown in Figure \ref{plot_d_tau}, even its $97.5\%$ pointwise one-sided upper limit of $0.00028$ lies below $g=0.05$. This means, by parts $(iii)$ and $(iv)$ of Lemma \ref{qbf_global_lemma}, that the conclusion \textit{only} holds for $c\leq 0.223$, and hence the quantile breakdown frontier is sharp at $\tau=.3$. Using equation \eqref{eq:bounds_derived}, we can find the bounds for the global effect for all quantiles which are consistent with this amount of policy selection bias. This is shown in Figure \ref{bounds_qbf}. Note that at $\tau=0.3$, the lower bound is $.05$ by construction. At upper quantiles, the lower bound is slightly negative. The figure reports uniform confidence bands for three step sizes, $\epsilon_n\in\{n^{-1/2},n^{-1/3},n^{-1/6}\}$. Among those that satisfy the two rate conditions, $n^{-1/3}$ is our preferred choice: in the Monte Carlo study of Appendix \ref{appendix_mc} its bands reach the nominal coverage level at $n=10000$, whereas those for $n^{-1/6}$ remain well below it.\footnote{The step size must satisfy two conditions: $\epsilon_n\to 0$, and $\epsilon_n\sqrt n\to\infty$, so $\epsilon_n$ must vanish, but more slowly than $n^{-1/2}$. Note that $\epsilon_n=n^{-1/2}$ itself violates the second condition, since then $\epsilon_n\sqrt n=1$. Nevertheless, it is reported because it reproduces the standard bootstrap and so serves as a conservative reference. Both $n^{-1/3}$ and $n^{-1/6}$ satisfy the two conditions. For the lower bands, $n^{-1/3}$ and $n^{-1/6}$ are almost indistinguishable.}

\begin{figure}
\centering
\includegraphics[scale=0.18]{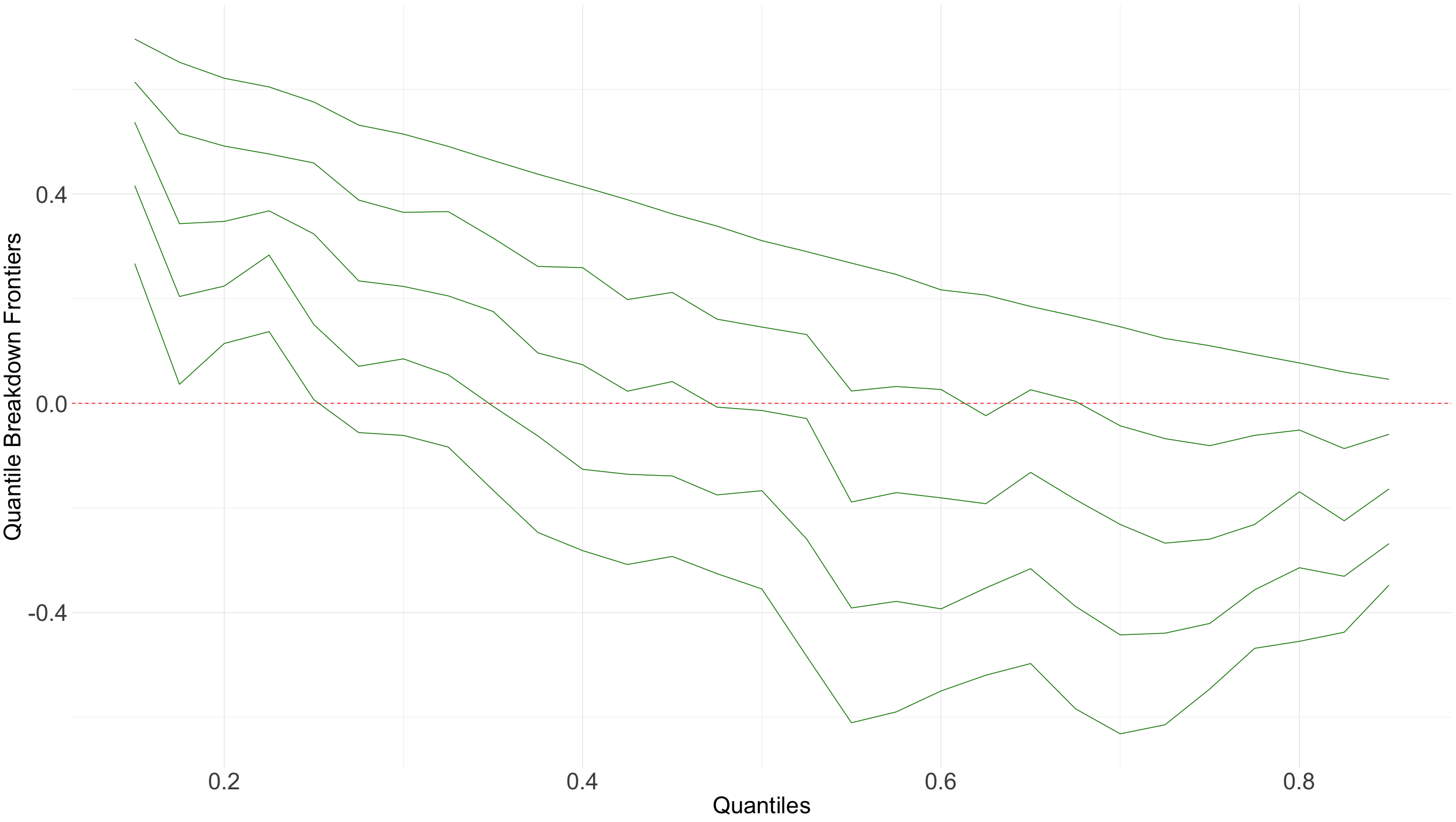}
\caption{\emph{Quantile Breakdown Frontiers for $G_{\tau,D_\delta}\geq g$ and $g=0, 0.025, 0.05, 0.075, 0.1.$ Top one is $g=0$.}}
\label{empirical_qbf_many_2}
\end{figure}

\begin{figure}
\centering
\includegraphics[scale=0.18]{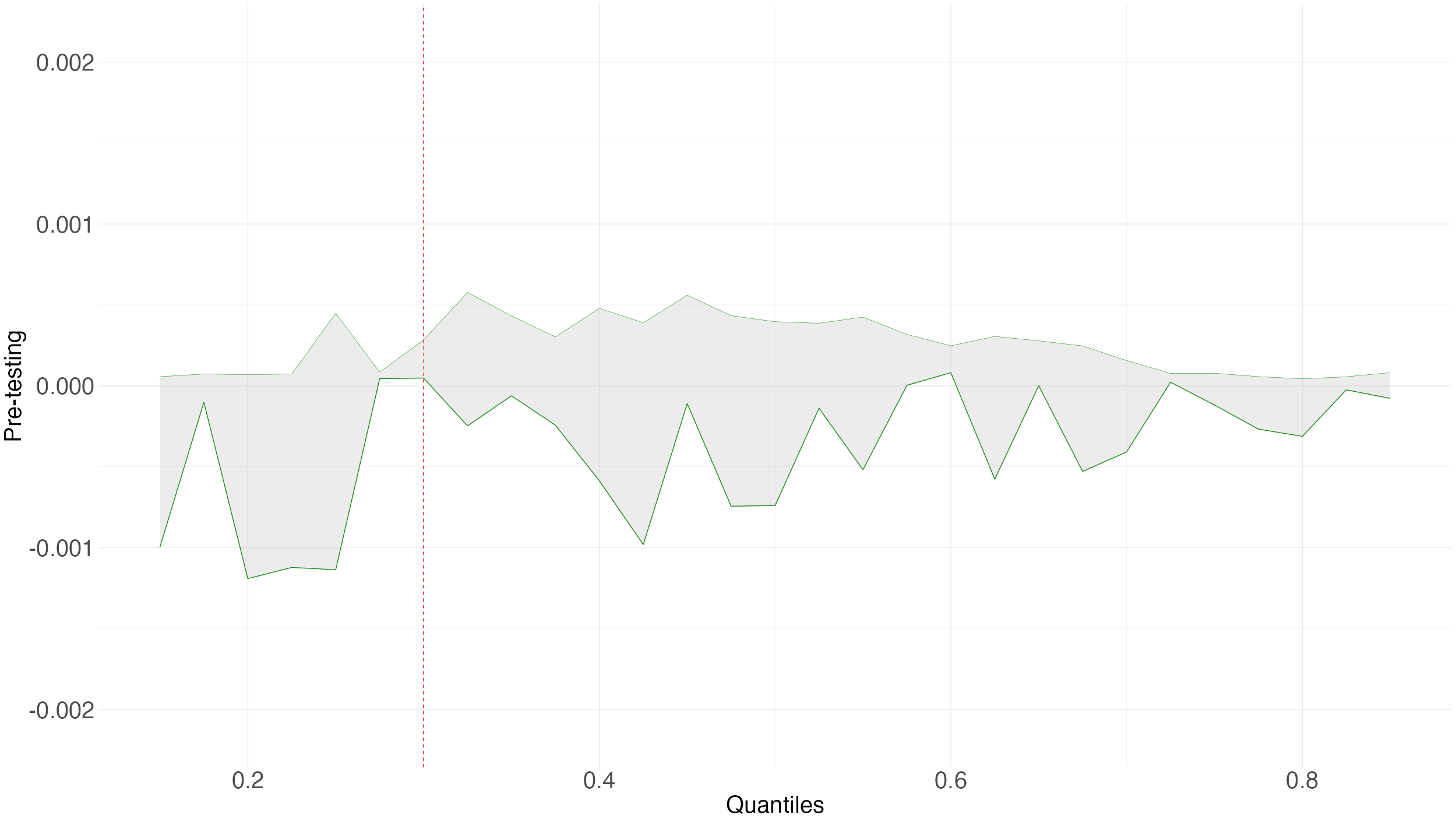}
\caption{\emph{Pre-testing: $97.5\%$ one-sided upper pointwise confidence band for $d_\tau.$}}
\label{plot_d_tau}
\end{figure}

\begin{figure}
\centering
\includegraphics[scale=0.18]{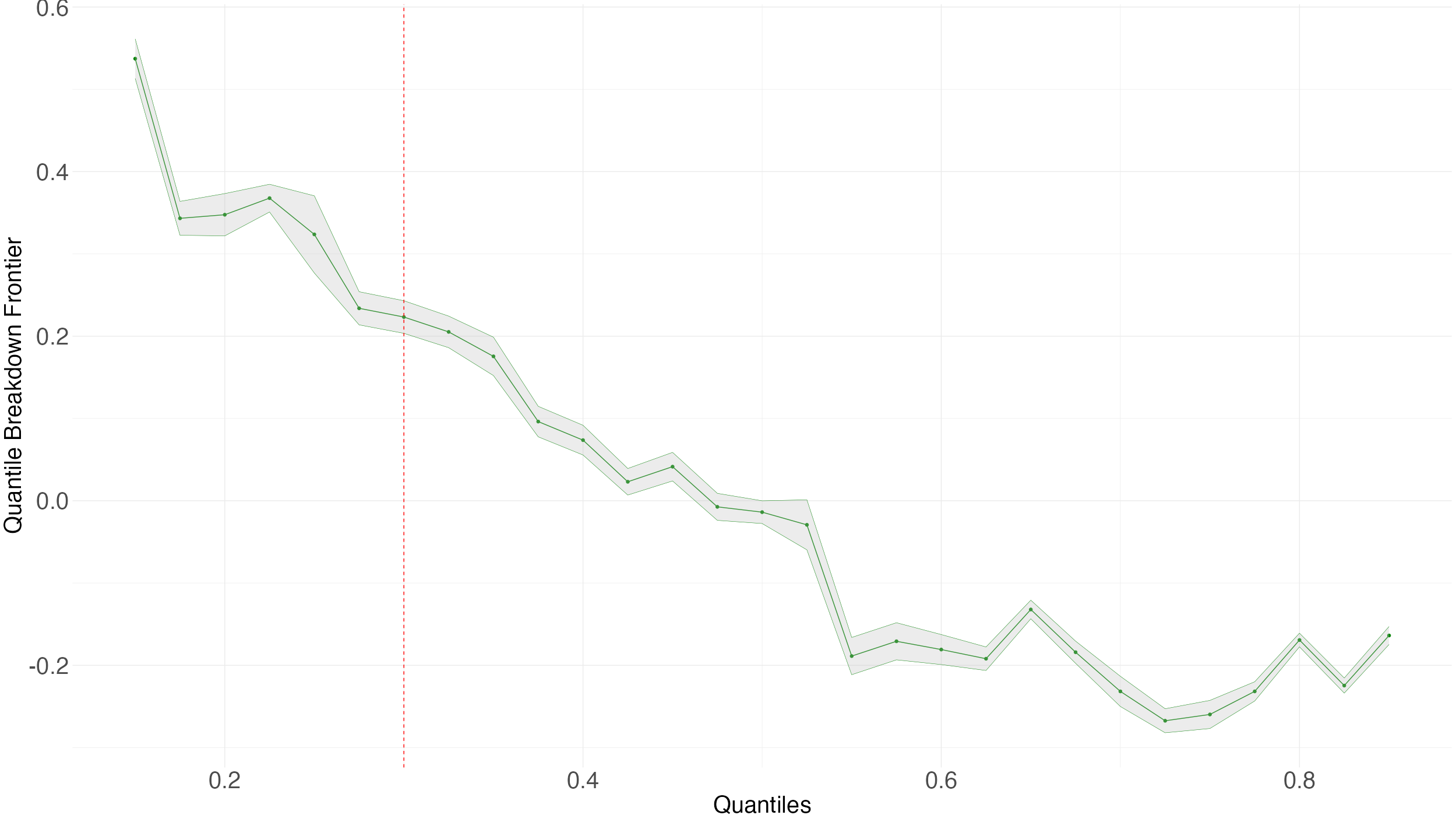}
\caption{\emph{Quantile Breakdown Frontier for $G_{\tau,D_\delta}\geq0.05$ with $\approx97.5\%$ pointwise confidence bands.}}
\label{empirical_qbf}
\end{figure}

\begin{figure}
\centering
\includegraphics[scale=0.18]{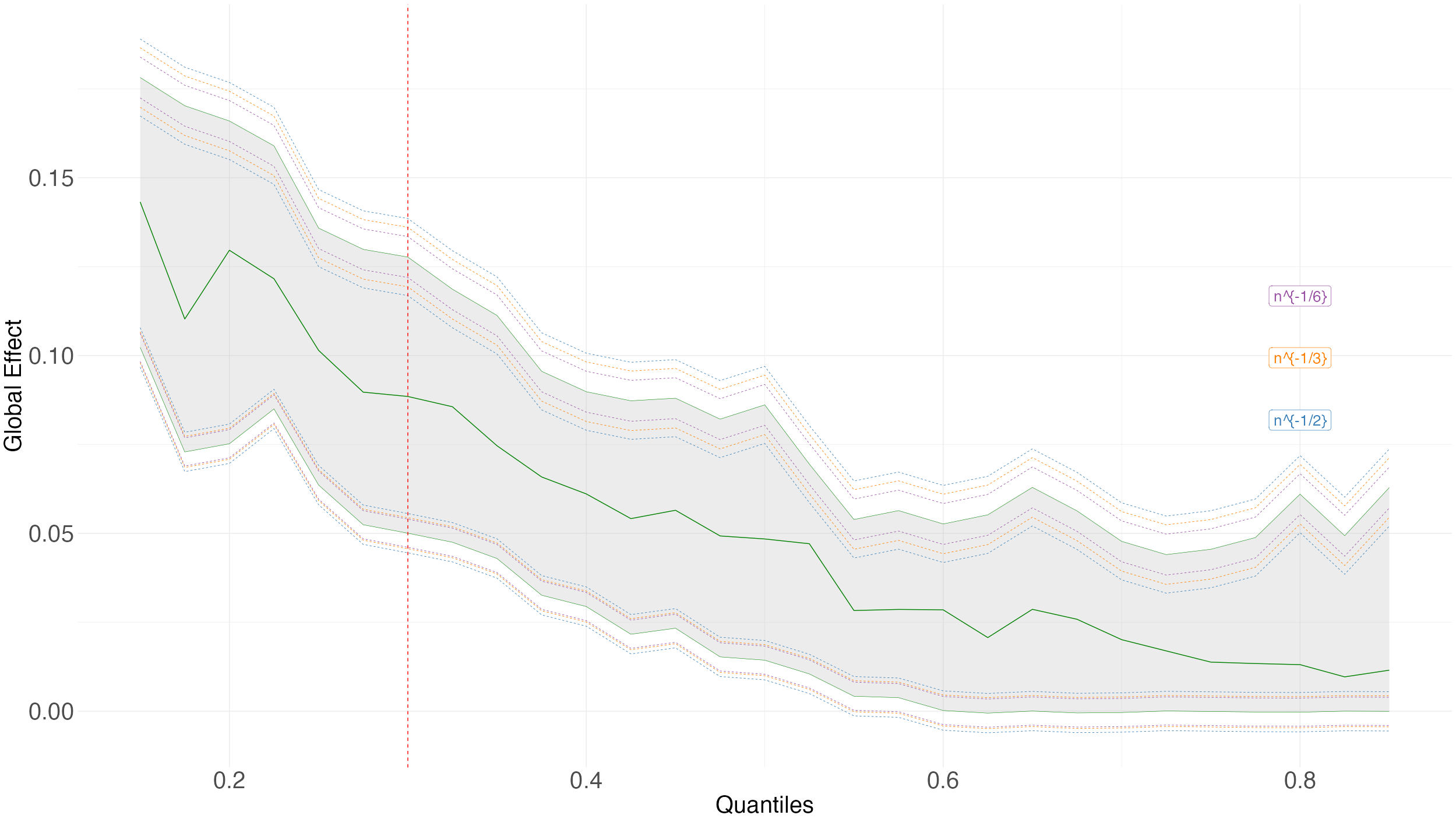}
\caption{\emph{Bounds derived from the QBF with $95\%$ uniform confidence bands (each bound at $97.5\%$, joint $95\%$ by the Bonferroni correction), and the point identified global effects.}}
\label{bounds_qbf}
\end{figure}

In order to interpret the magnitude of the QBF, we compare it to the observable difference
		\begin{align}\label{eq:obs_psb}
c_{obs}:=\sup_{y\in \mathbb R}\left |F_{Y(0)|D=0,D_\delta=1}(y)-\int_{\mathcal X}F_{Y|D=1,X=x}(y)dF_{X|D=0,D_\delta=1}(x)\right |
\end{align}	
This is different from the KS-distance in Assumption \ref{assumption_ks_distance}, which is stated in terms of $F_{Y(1)|D=0,D_\delta=1}(y)$, an unobserved distribution. In \eqref{eq:obs_psb}, we are comparing the difference between the newly-treated when they are not treated, versus the treated matched with the newly treated covariates. The empirical counterpart of \eqref{eq:obs_psb} is $\hat c_{obs}=0.243$, and $(0.231 ,0.256)$ is the $95\%$ confidence interval computed via the usual bootstrap. This is similar to the value of the QBF at $\tau=.3$ which is  $\hat c_{.3,L}=0.223$. If the treatment does not have a very large effect on the distribution of the outcome, then one could argue that the conclusion of interest might hold. However, if the treatment might have a very strong effect on the distribution, then it might be less likely that the conclusion holds.

\section{Conclusion}
This paper provides a way to perform a sensitivity analysis on the unconditional quantiles effects of policies that manipulate the proportion of treated individuals. To do so, we introduce the quantile breakdown frontier as a tool to examine, across quantiles, whether a sensitivity analysis is possible, and how much policy selection bias is compatible with a given conclusion. Next, we use the information from the curve at a particular quantile to provide bounds on the effect of a policy across quantiles. Our empirical application looks at the effect of increasing the proportion of unionized workers by unionizing lower earners. We find that an increase of $10\%$ for the $30$\textsuperscript{th} quantile is consistent under moderate values of selection bias. Moreover, this implies that the effect on lower quantiles would be even bigger, and that lower quantiles would not be hurt.

\clearpage
\bibliographystyle{aea}
\bibliography{references}

\clearpage
\appendix
\section*{Appendices}

\renewcommand{\thesubsection}{\Alph{subsection}}

\setcounter{equation}{0}
\renewcommand\theequation{A.\arabic{equation}}%

\subsection{Understanding the Policy Selection Bias}\label{app_pol_bias}
    In an RCT with perfect compliance there is no selection bias, and, if a counterfactual policy shifts people selected at random, then there is no policy selection bias either. In this case, the global effect of such policy is identified. A new policy where $D_\delta$ is correlated with $U_1$ induces policy selection bias, and we lose identification. 
    
  To understand the policy selection bias, we carry out a Monte Carlo simulation of an RCT with perfect compliance. By dispensing of the selection bias, we aim to study the policy selection bias in isolation. The effect of a policy that shifts people at random will be denoted by $G_\tau$. By Theorem \ref{thm_bounds_global}, we need to set $c=0$, and we get
    \begin{align*}
G_{\tau}=F_A^{-1}(\tau)-F_{Y}^{-1}(\tau)
\end{align*}
    where $F_A(y)=(p+\delta) F_{Y|D=1}(y)  + (1-p-\delta) F_{Y|D=0}(y)$. As mentioned above, this is identified. The global effect of another policy $D_\delta$ will be denoted by $G_{\tau, D_\delta}$, and is given by
    \begin{align*}
G_{\tau, D_\delta}&:=F_{Y_{D_\delta}}^{-1}(\tau)-F_{Y}^{-1}(\tau).
\end{align*}
In the simulation I will show the discrepancy between $G_{\tau}$ and $G_{\tau, D_\delta}$, and also compute numerically
\begin{align*}
\sup_{y\in\mathbb R}\left |F_{Y(1)|D=0,D_\delta=1}(y)-\int_{\mathcal X}F_{Y|D=1,X=x}(y)dF_{X|D=0,D_\delta=1}(x)\right |
\end{align*}
which, if there are not covariates, boils down to $\sup_{y\in\mathbb R}\left |F_{Y(1)|D=0,D_\delta=1}(y)-F_{Y|D=1}(y)\right |$. The policies I will compare will be similar to the one in the empirical application:  increase the proportion of treated units by $\delta$, by treating those untreated individuals whose outcome is below the $\delta/(1-p)$-quantile: treat everyone who was already treated, plus the $\delta$ lowest untreated Y values.

 These are details of the simulation:
     \begin{align*}
Y(0)&\sim \mathcal N(0,1)\\
Y(1)&=1+ Y(0) + U,
\end{align*}
where $U\sim \mathcal N(0,1)$. Assignment to treatment obeys $\Pr(D=1)=p=1/2$, with $D\perp Y(0), Y(1)$, and $Y = DY(1) + (1-D)Y(0)$. The counterfactual policy $D_\delta$ is defined by
 \begin{align}\label{d_delta_appendix}
D_\delta = 
\left\{
\begin{aligned}
1 &\quad \text{if } D = 1 \\
1 &\quad \text{if } D = 0 \text{ and } Y \le F^{-1}_{Y|D=0}\left(\frac{\delta}{1-p}\right) \\
0 &\quad \text{otherwise}
\end{aligned}
\right.
\end{align}
I averaged over 1,000 simulations to smooth the curves. Figure \ref{fig:gtau} plots $G_{\tau}$ vs. $G_{\tau, D_\delta}$ across $\tau$ for a fixed $\delta=0.1$. As expected,  $G_{\tau, D_\delta}$  is higher for lower quantiles and monotonically decreasing: as $\tau$ increases there are fewer and fewer individuals with very low $Y(0)$ but high $Y(1)$. On the other hand, $G_{\tau}$ shows the reverse pattern. This is also expected, because $Y(1)$ tends to be bigger $Y(0)$, bringing random individuals intro treatment shifts increases all the quantiles of $Y(1)$ while keeping the quantiles of $Y(0)$ roughly unchanged. In short, relying on (the identified) $G_{\tau}$ can be very misleading.

To see how the policy selection bias depends on $\delta$, Figure \ref{fig:sup} plots the map 
     \begin{align*}
\delta\mapsto \sup_{y\in\mathbb R}\left |F_{Y(1)|D=0,D_\delta=1}(y)-F_{Y|D=1}(y)\right |,
\end{align*}
and shows that as $\delta$ increases the policy selection bias decreases. As $\delta$ increases, because of the particular policy we are considering in \eqref{d_delta_appendix}, we start shifting more individuals that are not originally selected into treatment. In a sense then $Y(1)|D=0, D_\delta=1$ gets closer to $Y(1)|D=1$.
\begin{figure}[ht]
\centering
\begin{minipage}{0.48\textwidth}
  \centering
  \includegraphics[width=\linewidth]{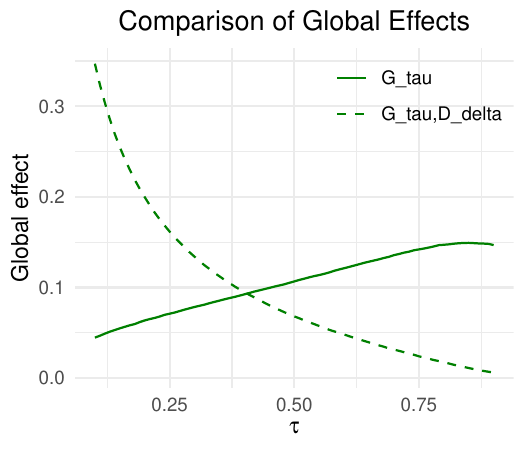}
  \caption{\textit{$G_{\tau}$ vs. $G_{\tau, D_\delta}$}.}
  \label{fig:gtau}
\end{minipage}
\hfill
\begin{minipage}{0.48\textwidth}
  \centering
  \includegraphics[width=\linewidth]{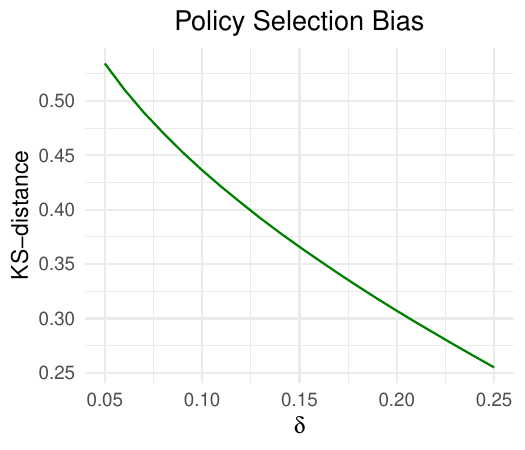}
  \caption{\textit{Policy selection bias as a function of $\delta$.}}
  \label{fig:sup}
\end{minipage}
\end{figure}

\subsection{Understanding Lemma \ref{qbf_global_lemma}}\label{understanding_qbf_global_lemma}

This appendix explains how to read the quantile breakdown frontier, using the simulation of Appendix \ref{app_pol_bias}. The frontier $c_{\tau,L}$ answers the following question: how much policy selection bias can there be before the given conclusion $G_{\tau,D_\delta}\geq g$ is no longer guaranteed? Lemma \ref{qbf_global_lemma} distinguishes three situations. The idea of this section is to illustrate these situations.
\begin{itemize}
\item \emph{The frontier is negative} (part $(ii)$). The conclusion fails to hold under point identification, this means, even if there is no policy selection bias at all. In this case, there is no sensitivity analysis to carry out.
\item \emph{The frontier is nonnegative} (part $(iii)$). The conclusion is guaranteed whenever the policy selection bias is at most $c_{\tau,L}$. Whenever the policy selection bias is larger, the conclusion is no longer guaranteed. This means that the conclusion may still be true, but the data cannot tell.
\item \emph{The conclusion holds whatever the bias} (part $(iv)$). This happens when the bounds that use the no-assumption-bounds on the newly treated already deliver the conclusion. In that case the frontier is not a point at which the conclusion breaks down.
\end{itemize}

To illustrate these three points, we follow the setting of Appendix \ref{app_pol_bias} with $\delta=0.1$: the policy treats the untreated individuals with the lowest outcomes. Because the data are simulated, we know the true effect $G_{\tau,D_\delta}$, shown in Figure \ref{fig:gtau}, and the actual policy selection bias, which is approximately $0.44$, as shown in Figure \ref{fig:sup}. Of course, a researcher would know neither. %

Fix the median. Figure \ref{fig:lb_c} plots the lower bound of Theorem \ref{thm_bounds_global} at $\tau=0.5$ as a function of $c$. With no policy selection bias, $c=0$, the lower bound is $0.20$. As more bias is allowed the lower bound falls, until at $c\approx 0.66$ it reaches zero and stays there: from that point on, the bound that uses no assumption, which is zero in this example, takes over.\footnote{It is zero because the policy treats the individuals in the left tail of the distribution. Even if treatment lowered their outcomes arbitrarily, they would remain at the bottom of the distribution and the quantiles above them would not move.} The three situations of the lemma can be read off the figure by drawing a horizontal line at the threshold $g$. The frontier is the value of $c$ at which the lower bound crosses that line. We consider three different values of the threshold.

\begin{itemize}
\item $g=0.25$. The line is above the curve already at $c=0$. The conclusion $G_{0.5,D_\delta}\geq 0.25$ fails even without policy selection bias. This is part $(ii)$: indeed, the frontier is negative, $c_{0.5,L}\approx -0.17$. This can be seen by extending Figure \ref{fig:lb_c} to the left for negative values of $c$.
\item $g=0.05$. The curve crosses the line at $c\approx 0.49$, which is the frontier. This is part $(iii)$: for any policy selection bias up to $0.49$ the conclusion is guaranteed. The actual bias, $0.44$, is below the frontier, and indeed the true effect at the median is $0.07$.
\item $g=0$. The curve never goes below the line. This is part $(iv)$: the conclusion $G_{0.5,D_\delta}\geq 0$ holds for every $c\in[0,1]$. The formula for the frontier still returns a number, $0.66$, but nothing breaks down there: beyond it the lower bound is flat in $c$ and stays at the threshold.
\end{itemize}

\begin{figure}[ht]
\centering
  \includegraphics[width=0.5\textwidth]{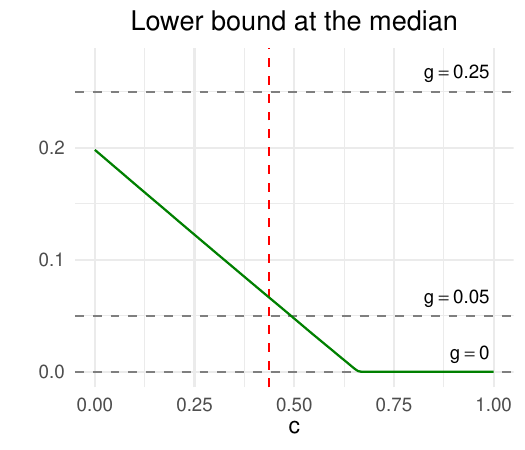}
  \caption{\textit{Lower bound of Theorem \ref{thm_bounds_global} at the median as a function of $c$. The vertical line is the actual policy selection bias.}}
  \label{fig:lb_c}
\end{figure}

Figures \ref{fig:four-markers} and \ref{fig:four-markers_2} take a different look at part $(iv)$ by showing why the lower bound becomes flat. They place on the outcome axis the two quantities that enter the lower bound of Theorem \ref{thm_bounds_global}, $\max\left\{\tilde F_A^{-1}(\tau-\delta) ,F_A^{-1}(\tau-\delta c)\right\} - F_{Y}^{-1}(\tau)$. The first does not depend on $c$. The second moves to the left as $c$ grows, and it equals $g_{\tau}+F_Y^{-1}(\tau)$ exactly at the frontier, $c=c_{\tau,L}$. In Figure \ref{fig:four-markers}, $\tilde F_A^{-1}(\tau-\delta)$ lies to the right of $g_{\tau}+F_Y^{-1}(\tau)$. For small $c$ the lower bound is given by $F_A^{-1}(\tau-\delta c)$ and decreases with $c$; once $F_A^{-1}(\tau-\delta c)$ falls below $\tilde F_A^{-1}(\tau-\delta)$, which happens before $c$ reaches the frontier, the lower bound is flat in $c$ and remains above the threshold. The conclusion holds for any $c$. In Figure \ref{fig:four-markers_2}, $\tilde F_A^{-1}(\tau-\delta)$ lies to the left of $g_{\tau}+F_Y^{-1}(\tau)$. Once $c$ exceeds the frontier, as $c_2$ does in the figure, the lower bound is below the threshold and the conclusion is no longer guaranteed. In this second case the frontier is sharp.

\begin{figure}[ht]
    \centering
    \begin{tikzpicture}[xscale=1.2, yscale=1.2, every node/.style={font=\large}]
        \draw[thick, -{Stealth[length=4pt]}] (0,0) -- (14.5,0);

        \foreach \x in {2,6,9,12}{
            \draw[thick] (\x,0.1) -- (\x,-0.1);
        }

        \node[above=4pt] at (2,0) {$g_{\tau}+F_Y^{-1}(\tau)$};

        \node[below=6pt] at (2,0) {$F_A^{-1}(\tau-\delta c_{\tau,L})$};
        \node[below=6pt] at (6,0) {$F_A^{-1}(\tau-\delta c_1)$};
        \node[below=6pt] at (9,0) {$\tilde F_A^{-1}(\tau-\delta)$};
        \node[below=6pt] at (12,0) {$F_A^{-1}(\tau)$};
    \end{tikzpicture}
    \caption{Illustration of Part $(iv)$ of Lemma \ref{qbf_global_lemma}: $\tilde F_A^{-1}(\tau-\delta)-F_Y^{-1}(\tau)\geq g_{\tau,L} $.}
    \label{fig:four-markers}
\end{figure}

\begin{figure}[ht]
    \centering
    \begin{tikzpicture}[xscale=1.2, yscale=1.2, every node/.style={font=\large}]
        \draw[thick, -{Stealth[length=4pt]}] (0,0) -- (14.5,0);

        \foreach \x in {2,6,9,12}{
            \draw[thick] (\x,0.1) -- (\x,-0.1);
        }

        \node[above=4pt] at (9,0) {$g_{\tau}+F_Y^{-1}(\tau)$};

        \node[below=6pt] at (2,0) {$\tilde F_A^{-1}(\tau-\delta)$};
        \node[below=6pt] at (6,0) {$F_A^{-1}(\tau-\delta c_2)$};
        \node[below=6pt] at (9,0) {$F_A^{-1}(\tau-\delta c_{\tau,L})$};
        \node[below=6pt] at (12,0) {$F_A^{-1}(\tau)$};
    \end{tikzpicture}
    \caption{Illustration of Part $(iv)$ of Lemma \ref{qbf_global_lemma}: $\tilde F_A^{-1}(\tau-\delta)-F_Y^{-1}(\tau)< g_{\tau,L}.$ }
    \label{fig:four-markers_2}
\end{figure}

\subsection{Analysis of the Marginal Effect}\label{app_maginal_effect}
Before we attempt to bound the marginal effect, we will first investigate the issue of existence. To that end, we introduce the following assumption.
\begin{assumption}[Limiting Distributions]\label{assumption_limit_distribution}
 For a given sequence of policies $\mathcal D$, the following maps exist and are continuous on $y\in\mathbb R$:
 \begin{align*}
y\mapsto \frac{\partial F_{Y|D_\delta=0}(y)}{\partial \delta}\bigg|_{\delta=0}:= \lim_{\delta\to0} \frac{F_{Y|D_\delta=0}(y)-F_{Y|D=0}(y)}{\delta}
  \end{align*}
and
   \begin{align*}
y\mapsto   F_{Y(1)|\partial D}(y):=\lim_{\delta\to0}F_{Y(1)|D=0,D_\delta=1}(y).
  \end{align*}
  Moreover, convergence takes place uniformly:
  \begin{align*}
\lim_{\delta\to0}\sup_{y\in\mathbb R} \left|\frac{F_{Y|D_\delta=0}(y)-F_{Y|D=0}(y)}{\delta}-\frac{\partial F_{Y|D_\delta=0}(y)}{\partial \delta}\bigg|_{\delta=0} \right|=0,
\end{align*}
and
  \begin{align*}
\lim_{\delta\to0}\sup_{y\in\mathbb R} \left|F_{Y(1)|D=0,D_\delta=1}(y) - F_{Y(1)|\partial D}(y) \right| =0.
\end{align*}
  
  \end{assumption}

\begin{theorem}[Existence of Marginal Effect]\label{theorem_existence}
Let $\mathcal D$ be a sequence of policies such that Assumptions \ref{assumption_policy}, \ref{assumption_regularity}, and \ref{assumption_limit_distribution} hold. If $f_Y(F_Y^{-1}(\tau))>0$, then $M_{\tau,\mathcal D}$ exists and is given by
\begin{equation*}
M_{\tau,\mathcal D}=-\frac{\dot F_{Y,\mathcal D}(F_Y^{-1}(\tau))}{f_Y(F_Y^{-1}(\tau))}.
\end{equation*}
where $\dot F_{Y,\mathcal D}(y)$ satisfies
\begin{align*}
\lim_{\delta \downarrow 0}\sup_{y\in\mathbb R} \left|\frac{F_{Y_{D_\delta}}(y)-F_{Y}(y)}{\delta}-\dot F_{Y,\mathcal D}(y)\right|=0.
\end{align*}
\end{theorem}

The conditions and the proof of Theorem \ref{theorem_existence} come from viewing the marginal effect as a Hadamard derivative. Inspection of \eqref{count_dist_dec} shows that the assumptions of the theorem ensure that $\lim_{\delta\to 0}F_{Y_{D_\delta}}=F_Y$. That is, in the limit, the sequence of policies lead to the observed distribution of $Y$. To better understand the smoothness conditions imposed by Assumption \ref{assumption_limit_distribution} we provide an example and a counterexample.

\begin{example}\label{example_existence_of_marg}
Suppose that $D=\mathds 1\left\{ V\geq 0 \right\}$ for $V$ continuously distributed around 0. For a sequence $\delta\downarrow 0$, consider $D_\delta=\mathds 1\left\{ V\geq F_V^{-1}(F_V(0)-\delta) \right\}$. Then, for $p:=\Pr(D=1)$, we have that $\Pr(D_\delta=1)=p+\delta$.
Under mild conditions
\begin{align*}
\lim_{\delta\to 0}\frac{F_{Y|D_\delta=0}(y)-F_{Y|D=0}(y)}{\delta}=-\left(\frac{1}{f_V(0)}+\frac{1}{p}\right)F_{Y|D=0}(y)
\end{align*}
and 
\begin{align*}
\lim_{\delta\to0} F_{Y(1)|D=0,D_\delta=1}(y) = F_{Y(1)|V=0}(y)
\end{align*}
We can interpret $V=0$ as those individuals who are indifferent with respect to treatment. That is why we denote them with $\partial D$. See the appendix for detailed calculations. 

For a counterexample, consider instead $D_\delta=\mathds 1\left\{ V\leq F_V^{-1}(F_V(0)+\delta) \right\}$,  where the inequality sign is reversed. Here, $\Pr(D_\delta=1)=p+\delta$ provided $V$ is symmetric around 0: $F_V(0)=1-F_V(0)$. However,
\begin{align*}
F_{Y|D_\delta=0}(y)= \Pr(Y\leq y| V\geq F_V^{-1}(F_V(0)+\delta)) \to \Pr(Y\leq y| V\geq 0)  = F_{Y|D=1}(y).
\end{align*}
which contradicts the existence of the derivative $\frac{\partial F_{Y|D_\delta=0}(y)}{\partial \delta}\bigg|_{\delta=0}$ required in Assumption \ref{assumption_limit_distribution}.
\end{example}

To provide sharp bounds on the marginal effect we need the following assumption.
\begin{assumption}[More Limiting Distributions]\label{assumption_indifference}
 For a given sequence of policies $\mathcal D$, 
 \begin{enumerate}
 \item $\delta\mapsto F_{Y|D_{\delta}}(y)$ is differentiable at $\delta=0$ for every $y\in\mathbb R$.
 \item $\delta \mapsto F_{X| D=0, D_{\delta}=1}(x)$ is differentiable at $\delta=0$ for every $x\in\mathcal X$. Moreover, we denote $ F_{X|\partial D}(x):=     \lim_{\delta\downarrow0} F_{X| D=0, D_{\delta}=1}(x).$
\end{enumerate}
  \end{assumption}
Here, $F_{X|\partial D}(x)$ can be seen as the observable characteristics of the those individuals who are ``indifferent'' to treatment. Indeed, for the case of Example \ref{example_existence_of_marg}, we have $F_{X|\partial D}(x)=F_{X|V}(x|0)$.

\begin{theorem}[Bounds on the Marginal Effect]\label{thm_bounds_marginal}
If Assumptions \ref{assumption_policy}, \ref{assumption_regularity}, \ref{assumption_limit_distribution}, and \ref{assumption_indifference} hold, Assumption \ref{assumption_ks_distance} holds for all $\delta$ in a neighborhood of 0, and $F_{X|\partial D}$ is known, then, for $\tau\in(0,1)$, 
\begin{align*}
\theta_{L,M} \leq M_{\tau, \mathcal D} \leq \theta_{U,M}
\end{align*}
where
\begin{align*}
\theta_{L,M}:=\max\left\{\frac{F_{Y|D=0}(F^{-1}_{Y}(\tau))-1}{f_{Y}( F^{-1}_{Y}(\tau))} ,\frac{F_{Y|D=0}(F^{-1}_{Y}(\tau))-\int_{\mathcal X}F_{Y|D=1,X=x}(F^{-1}_{Y}(\tau))dF_{X|\partial D}(x)-c}{f_{Y}( F^{-1}_{Y}(\tau))}  \right\},
\end{align*}
and
\begin{align*}
\theta_{U,M}:=\min\left\{\frac{F_{Y|D=0}(F^{-1}_{Y}(\tau))}{f_{Y}( F^{-1}_{Y}(\tau))},\frac{F_{Y|D=0}(F^{-1}_{Y}(\tau))-\int_{\mathcal X}F_{Y|D=1,X=x}(F^{-1}_{Y}(\tau))dF_{X|\partial D}(x)+c}{f_{Y}( F^{-1}_{Y}(\tau))}\right\}.
\end{align*}
For a fixed $\tau\in(0,1)$, these bounds are sharp.
\end{theorem}

\begin{remark}
The bounds for $M_{\tau, \mathcal D}$ require knowledge of $F_{X|\partial D}(x)$, \emph{i.e.}, the covariates for the individuals along the margin of indifference. This depends crucially on the type of counterfactual policy considered. Different policies may result in different individuals along the margin of indifference. One option is to set $F_{X|\partial D}(x)=F_{X| D=1}(x)$: marginal individuals have the (observable) characteristics of those with $D=1$. In this case, covariates no longer play a role in the bounds because
\begin{align*}
\int_{\mathcal X}F_{Y|D=1,X=x}(F^{-1}_{Y}(\tau))dF_{X|\partial D}(x)=\int_{\mathcal X}F_{Y|D=1,X=x}(F^{-1}_{Y}(\tau))dF_{X| D=1}(x)=\tau.
\end{align*}
A more general approach is to consider, for some $\alpha(x)\in[0,1]$
\begin{align}\label{convex_covariates}
F_{X|\partial D}(x)=\alpha(x) F_{X| D=1}(x) + (1-\alpha(x))F_{X| D=0}(x).
\end{align}
It must be kept in mind that any choice of $F_{X|\partial D}(x)$ entails a risk of misspecification. This makes the problem of bounding the marginal effect much harder.
\end{remark}

\begin{remark}
If we set $c=0$, and if Assumption \ref{assumption_ks_distance} holds for all $\delta$ in a neighborhood of 0, upon taking the limit in \eqref{eq:ks_distance} we obtain
\begin{align*}
F_{Y(1)|\partial D}(y)=\int_{\mathcal X}F_{Y|D=1,X=x}(y)dF_{X|\partial D}(x)
\end{align*}
for all $y\in\mathbb R$. This opens a number of ways to point identify the marginal effect. For example, we could require either $F_{Y(1)|\partial D}(y)=F_{Y| D=1}(y)$ or $F_{X|\partial D}(x)=F_{X| D=1}(x)$. In either case, we obtain
\begin{align*}
M_{\tau, \mathcal D} = \frac{F_{Y|D=0}(F^{-1}_{Y}(\tau))-F_{Y|D=1}(F^{-1}_{Y}(\tau))}{f_{Y}( F^{-1}_{Y}(\tau))},
\end{align*}
which is precisely the estimand of \cite{Firpo2009}. We could use \eqref{convex_covariates} for a given $\alpha(x)$ to obtain yet a different value for the marginal effect.
\end{remark}

\begin{remark}
The bounds are sharp in the sense that for any possible value of the marginal effect within the bounds, there is a corresponding sequence of newly treated distribution $F_{Y(1)|D=0,D_\delta=1}(y)$ which delivers that same marginal effect. This construction depends on $\tau$, implying that, rather than uniform, sharpness is pointwise in $\tau\in(0,1)$.
\end{remark}

For the marginal effect we can construct the quantile breakdown frontier in a similar fashion as with the global effect. However, in order to avoid the presence of the density $f_Y$ in the denominator,\footnote{As explained appendix, section \ref{app_frontier_der}, avoiding $f_Y$ allows us to retain a $\sqrt n$-consistent estimator.} we focus on conclusions $M_{\tau,\mathcal D}\leq 0$.\footnote{The case $M_{\tau,\mathcal D}\geq 0$ can be dealt with in the same way after some minor modifications. It is therefore omitted.} Thus, instead of a collection of conclusions that may vary with $\tau$, here we have a common conclusion across $\tau$. Since the marginal effect is a derivative, this amounts to looking at conclusions of whether the global effect in a neighborhood of $\delta=0$ is increasing or decreasing as $\delta$ increases. The quantile breakdown frontier for the marginal effect is the map
\begin{align}\label{qbf_def_marg}
\tau\mapsto c_{\tau,U}^M= \int_{\mathcal X}F_{Y|D=1,X=x}(F^{-1}_{Y}(\tau))dF_{X|\partial D}(x)-F_{Y|D=0}(F^{-1}_{Y}(\tau)),
\end{align}
for $\tau\in (0,1)$. The next lemma is analogous to Lemma \ref{qbf_global_lemma}.

\begin{lemma}\label{qbf_marginal_lemma}
Under the assumptions of Theorem \ref{thm_bounds_marginal}, the quantile breakdown frontier for $M_{\tau,\mathcal D}$ defined as the map in \eqref{qbf_def_marg} satisfies the following: (i) $\tau\mapsto c_{\tau,U}^M$ is continuous; (ii) if $c_{\tau,U}^M<0$, then the conclusion does not hold when $c=0$; and (iii) if $c_{\tau,U}^M\geq 0$, then for any $c\geq 0$ such that $c\leq  c_{\tau,U}^M$, the conclusion holds, and for $c> c_{\tau,U}^M$ the conclusion might not hold.
\end{lemma}

\begin{remark}
While the frontier for the global effect is valid for $\tau\in (\delta, 1-\delta)$, for the marginal effect it is valid for $\tau\in(0,1)$, provided that $Y$ has a bounded support. This difference arises because in the marginal effect we are taking the limit as $\delta\to 0$.
\end{remark}

Using sample analogs we obtain the following estimator of \eqref{qbf_def_marg}:
\begin{align}\label{qbf_estimation_marg}
\hat c_{\tau,U}^M = \sum_ {x\in\mathcal X}\hat F_{Y|D=1,X=x}(\hat F^{-1}_{Y}(\tau))\hat p_{x|\partial D}-\hat F_{Y|D=0}(\hat F^{-1}_{Y}(\tau))
\end{align}
Here, $\hat p_{x|\partial D}$ is the empirical counterpart of $\Pr(X=x|\partial D)$, which, as mentioned above, is the distribution of covariates for individuals at the margin of indifference. By Assumption \ref{assumption_indifference}, this is $ F_{X|\partial D}(x):=     \lim_{\delta\downarrow0} F_{X| D=0, D_{\delta}=1}(x).$ Following \eqref{convex_covariates}, for a \emph{user-supplied} $\alpha(x)\equiv \alpha$,\footnote{Note that $\alpha(x)$ is fixed across $x$ at $\alpha$ for simplicity.} we take $F_{X|\partial D}(x)=\alpha F_{X| D=1}(x) + (1-\alpha)F_{X| D=0}(x).$ Thus, 
\begin{align*}
\hat p_{x|\partial D} =\alpha \hat F_{X| D=1}(x) + (1-\alpha)\hat F_{X| D=0}(x).
\end{align*}
Details are given in Appendix \ref{appendix_notation}. The main assumption is 
\begin{assumption}[Functional CLT]\label{brownian_bridge_marginal}
The following multivariate functional central limit theorem holds 
\begin{align*} 
\sqrt n\begin{pmatrix}
 \sum_ {x\in\mathcal X}\hat F_{Y|D=1,X=x} \hat p_{x|\partial D} - E[F_{Y|D=1,X}|\partial D]\\
\hat F_{Y|D=0}-F_{Y|D=0} \\
\end{pmatrix} \rightsquigarrow \begin{pmatrix}
\mathbb G_{\partial D}\\ \mathbb G_{0}
\end{pmatrix},
\end{align*}
where $ \mathbb G_{\partial D}$ and $\mathbb G_{0}$ are Brownian bridges in $\ell^{\infty}(\mathcal Y)$.
\end{assumption}
Here, $E[F_{Y|D=1,X}|\partial D]$ implicitly denotes the map $y\mapsto E[F_{Y|D=1,X}(y)|\partial D]$. That is, the expectation is taken with respect to $X$ using $ F_{X|\partial D}(x)$. By Assumption \ref{dgp}, this expectation is just a 
finite convex combination of the distribution functions $F_{Y|D=1,X=x}(y)$ 

\begin{remark}
Appendix \ref{app_primitive_marginal} lists primitive conditions that result in Assumption \ref{brownian_bridge_marginal}. 
\end{remark}

The next assumption is needed to establish the Hadamard differentiability of the composition map, and the quantile process.

\begin{assumption}[Conditions for Hadamard Differentiability]\label{ass_had_dif_2}
\item
\begin{enumerate}
\item The distribution functions $ F_{Y|D=0}(y)$ and $ E[F_{Y|D=1,X}(y)|\partial D]$ are differentiable, with uniformly continuous and bounded derivatives on their support $\mathcal Y$. The derivatives are denoted by  $f_{Y|D=0}(y)$ and $E[f_{Y|D=1,X}(y)|\partial D]$ respectively.
\item The support $\mathcal Y$ is the compact set $[y_l,y_u]$.
\item $F_Y(y)$ is continuously differentiable on $\mathcal Y$ with strictly positive derivative $f_Y$.
\end{enumerate}
\end{assumption}

The next theorem establishes the asymptotic distribution of $\hat c_{\tau,U}^M$ as a process in $\ell^{\infty}(0,1)$. In such a case, we denote the process simply by $\hat c_{U}^M$.

\begin{theorem}[Asymptotic Distribution of QBF for Marginal Effect]\label{qbf_marginal}
Under Assumptions \ref{dgp}, \ref{brownian_bridge_marginal} and \ref{ass_had_dif_2}
\begin{align*}
\sqrt n(\hat c_{U}^M- c_{U}^M)&\rightsquigarrow   \mathbb G_{\partial D,Y}- \mathbb G_{0,Y},
\end{align*}
where 
\begin{align*}
\mathbb G_{\partial D,Y}&=\mathbb G_{\partial D}\circ F_Y^{-1}-E[f_{Y|D=1,X}|\partial D]\circ F_Y^{-1} \cdot \frac{\mathbb G_Y\circ F_Y^{-1}}{f_Y\circ F_Y^{-1}},\\
\mathbb G_{0,Y}&=\mathbb G_0\circ F_Y^{-1}-f_{Y|D=0}\circ F_Y^{-1} \cdot \frac{\mathbb G_Y\circ F_Y^{-1}}{f_Y\circ F_Y^{-1}}.
\end{align*}
\end{theorem}

As in the case of the quantile breakdown frontier for the global effect, the empirical bootstrap is valid and can be used to construct confidence intervals.

Using the same dataset of the empirical application in Section \ref{section_empirical}, we compute the QBF for the marginal effect. We focus on the conclusion $M_{\tau,\mathcal D}\leq 0$, and $\tau$ ranges from $0.05$ to $0.95$. For the estimation of $\hat p_{x|\partial D}$ we set $\alpha=0.5$. The frontier is shown in Figure \ref{empirical_qbf_marg}. The frontier is negative for most quantiles, except for the higher ones. According to Lemma \ref{qbf_marginal_lemma}, this means that for lower quantiles the conclusion does not hold under point identification. For upper quantiles it does hold, but the frontier shows that a rather small amount of policy selection bias can overturn the conclusion. 

\begin{figure}
\centering
\includegraphics[scale=0.18]{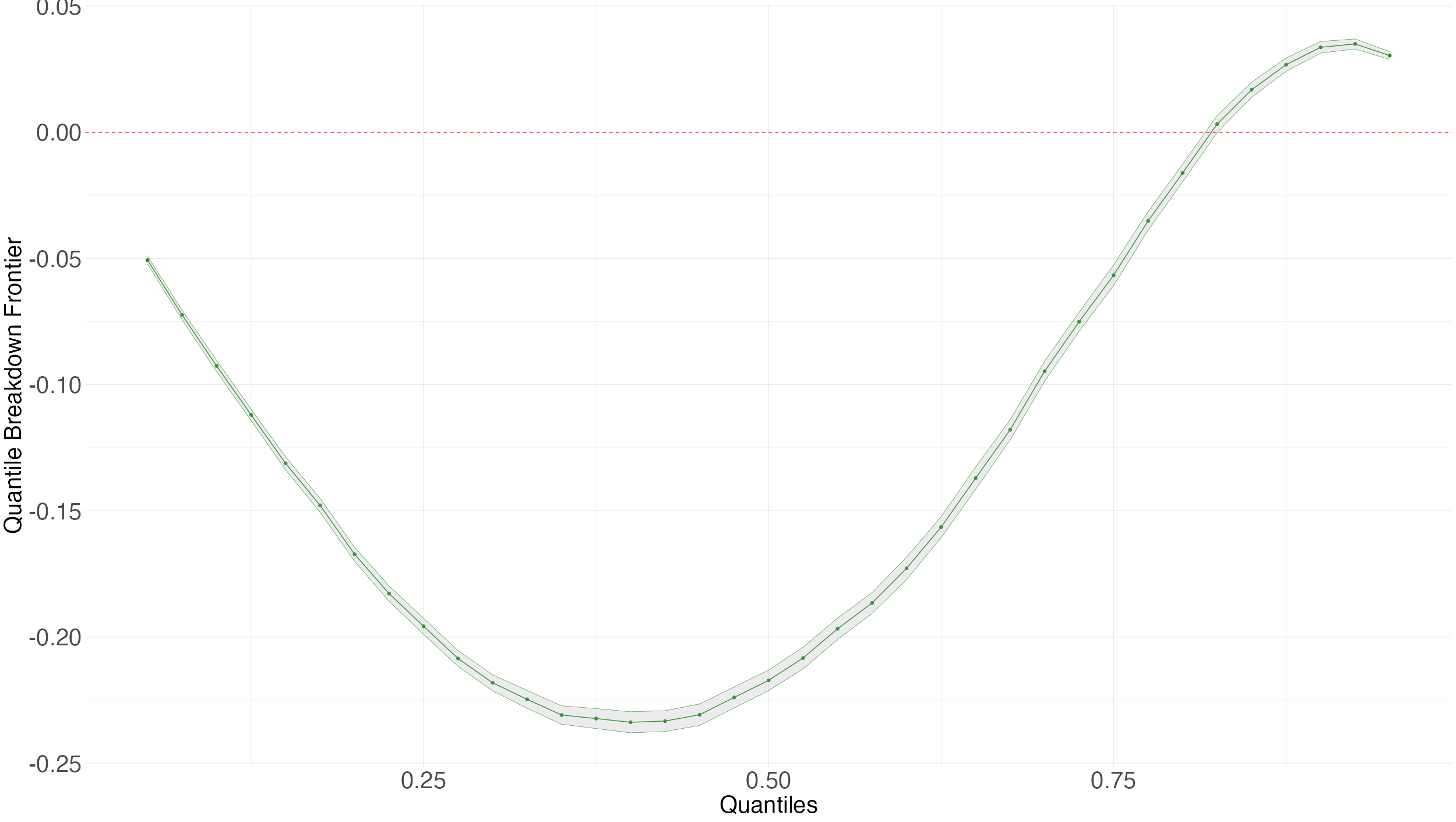}
\caption{\emph{Quantile Breakdown Frontier for $M_{\tau,\mathcal D}\leq 0$ with $95\%$ pointwise confidence bands.}}
\label{empirical_qbf_marg}
\end{figure}

\subsection{Derivation of the Quantile Breakdown Frontiers}\label{app_frontier_der}
By Theorem \ref{thm_bounds_global}, we have that, for a given $c$, the bounds on the global effect are
\begin{align*}
\max\left\{\tilde F_A^{-1}(\tau-\delta) ,F_A^{-1}(\tau-\delta c)\right\} - F_{Y}^{-1}(\tau) \leq G_{\tau, D_\delta}\leq \min\left\{\tilde F_A^{-1}(\tau) ,F_A^{-1}(\tau+\delta c)\right\}-F_{Y}^{-1}(\tau).
\end{align*}
If we are interested in a target conclusion $g_{\tau,L}\leq G_{\tau,D_\delta}  $ for some $g_{\tau,L}$, then we need the lower bound to be greater or equal than $g_{\tau,L}$. This ensures that the conclusion will hold. Consider the set of all $c$ such that the conclusion holds. That is
\begin{align*}
\left\{c:\max\left\{\tilde F_A^{-1}(\tau-\delta) ,F_A^{-1}(\tau-\delta c)\right\} - F_{Y}^{-1}(\tau) \geq g_{\tau,L}\right\}.
\end{align*}
Now, the expression on the left hand side is decreasing in $c$ provided that $\delta$ is positive. Thus, to look for the smallest $c$, denoted $c_{\tau,L}$, we need to solve $F_A^{-1}(\tau-\delta c_{\tau,L}) - F_{Y}^{-1}(\tau) = g_{\tau,L}$,
which yields 
\begin{align*}
c_{\tau,L}= \frac{\tau -F_A\left ( F_{Y}^{-1}(\tau)+ g_{\tau,L} \right )}{\delta}.
\end{align*}
When we allow this to vary with $\tau$, we obtain \eqref{qbf_def}. The presence of the $\max$ adds some nuances to the interpretation. This is what Lemma \ref{qbf_global_lemma} aims to clarify.

In order to obtain the analogous result for the marginal effect, we use Theorem \ref{thm_bounds_marginal}. Suppose we are interested in the conclusion $M_{\tau,\mathcal D}\leq m_{\tau,U}$, for some given $m_{\tau,U}$. Following the same logic as with the global effect, we need that the upper bound for $M_{\tau,\mathcal D}$ be less than or equal to $m_{\tau,U}$. Thus, the collection of $c$ compatible with the conclusion if given by the set
\begin{align*}
\left\{c:\min\left\{\frac{F_{Y|D=0}(F^{-1}_{Y}(\tau))}{f_{Y}( F^{-1}_{Y}(\tau))},\frac{F_{Y|D=0}(F^{-1}_{Y}(\tau))-\int_{\mathcal X}F_{Y|D=1,X=x}(F^{-1}_{Y}(\tau))dF_{X|\partial D}(x)+c}{f_{Y}( F^{-1}_{Y}(\tau))}\right\}\leq m_{\tau,U}\right\}.
\end{align*}
The bounds are increasing in $c$, so we look for the maximum $c$, denoted $c_{\tau,U}^M$ by solving
\begin{align*}
\frac{F_{Y|D=0}(F^{-1}_{Y}(\tau))-\int_{\mathcal X}F_{Y|D=1,X=x}(F^{-1}_{Y}(\tau))dF_{X|\partial D}(x)+c_{\tau,U}^M}{f_{Y}( F^{-1}_{Y}(\tau))}= m_{\tau,U}.
\end{align*}
That is,
\begin{align*}
c_{\tau,U}^M= f_{Y}( F^{-1}_{Y}(\tau))m_{\tau,U}+\int_{\mathcal X}F_{Y|D=1,X=x}(F^{-1}_{Y}(\tau))dF_{X|\partial D}(x)-F_{Y|D=0}(F^{-1}_{Y}(\tau)).
\end{align*}
The above display is a general version of equation \eqref{qbf_def_marg}. If we are interested in $m_{\tau,U}=0$, then we obtain \eqref{qbf_def_marg}. In this case we do not have to estimate $ f_{Y}( F^{-1}_{Y}(\tau))$, thus allowing us to retain a $\sqrt n$-consistent estimator.

\subsection{Notation for Estimation and Inference}\label{appendix_notation}
 
For $F_A(y)$ given in Theorem \ref{thm_bounds_global}, $\hat F_A(y)$ is its empirical counterpart. It is given by
\begin{align}\label{emp_app_dis}
\hat F_A(y) &:= \hat p \hat F_{Y|D=1}(y)  + (1 - \hat p - \delta) \hat F_{Y|D_\delta=0}(y) + \delta \sum_{x\in\mathcal X} \hat F_{Y|D=1,X=x}(y)\hat p_x
\end{align}
where $p_x := \Pr(X=x|D=0, D_\delta=1)$ and 
\begin{align*}
\hat p_x &= \frac{\sum_{i=1}^n  \mathds 1 \left\{X_i= x \right\}(1-D_i)D_{i,\delta} }{\sum_{i=1}^n  (1-D_i)D_{i,\delta}},\\
\hat p &= \frac{1}{n}\sum_{i=1}^n D_i,\\
\hat F_{Y|D=1}(y) &= \frac{\sum_{i=1}^n \mathds 1 \left\{Y_i\leq y \right\}D_i}{\sum_{i=1}^n D_i},\\
\hat F_{Y|D_\delta=0}(y) &= \frac{\sum_{i=1}^n \mathds 1 \left\{Y_i\leq y \right\}(1-D_{i,\delta})}{\sum_{i=1}^n (1-D_{i,\delta})},\\
 \hat F_{Y|D=1,X=x}(y) &= \frac{\sum_{i=1}^n  \mathds 1 \left\{Y_i\leq y \right\}D_i  \mathds 1 \left\{X_i=x \right\} }{\sum_{i=1}^n  D_i  \mathds 1 \left\{X_i=x \right\}  }.
\end{align*}
Moreover, we have
\begin{align}\label{emp_incomplete_app_dis}
\hat {\tilde F}_A(y) &:= \hat p \hat F_{Y|D=1}(y)  + (1 - \hat p - \delta) \hat F_{Y|D_\delta=0}(y).
\end{align}

For the marginal effect, we have the following empirical cdf:
\begin{align*}
\hat F_{Y|D=0}(y) = \frac{\sum_{i=1}^n\mathds1\{Y_i\leq y\}(1-D_i)}{\sum_{i=1}^n(1-D_i)}.
\end{align*}
and
\begin{align*}
\sum_ {x\in\mathcal X}\hat F_{Y|D=1,X=x}(y)\hat p_{x|\partial D}= \sum_{x\in \mathcal X} \frac{\sum_{i=1}^n  \mathds 1 \left\{Y_i\leq y \right\}D_i  \mathds 1 \left\{X_i=x \right\} }{\sum_{i=1}^n  D_i  \mathds 1 \left\{X_i=x \right\}  }\hat p_{x|\partial D} .
\end{align*}

\subsection{Proofs}\label{appendix_proofs}

\begin{proof}[Proof of Theorem \ref{thm_bounds_global}]

Recall that, for a given $D_\delta$, the counterfactual distribution is
\begin{align*}
F_{Y_{D_\delta}}(y) &= pF_{Y|D=1}(y)  + (1-p-\delta) F_{Y|D_\delta=0}(y)  +\delta F_{Y(1)|D=0,D_\delta=1}(y).
\end{align*}
Since $F_{Y(1)|D=0,D_\delta=1}$ is a distribution function, we naturally have $0\leq F_{Y(1)|D=0,D_\delta=1}(y)\leq 1$ for every $y\in\mathbb R$. Moreover, by Assumption \ref{assumption_ks_distance}, we also have that
\begin{align*}
 \int_{\mathcal X}F_{Y|D=1,X=x}(y)dF_{X|D=0,D_\delta=1}(x)-c\leq F_{Y(1)|D=0,D_\delta=1}(y)\leq \int_{\mathcal X}F_{Y|D=1,X=x}(y)dF_{X|D=0,D_\delta=1}(x)+c
\end{align*}
$y\in\mathbb R$. Thus combining both bounds, we obtain
\begin{align*}
\max\left\{0, \int_{\mathcal X}F_{Y|D=1,X=x}(y)dF_{X|D=0,D_\delta=1}(x)-c\right\}&\leq F_{Y(1)|D=0,D_\delta=1}(y)\\
&\leq \min\left\{1,\int_{\mathcal X}F_{Y|D=1,X=x}(y)dF_{X|D=0,D_\delta=1}(x)+c\right\}.
\end{align*}
for every $y\in\mathbb R$. To help with the notation, define
\begin{align}\label{incomplete_apparent}
\tilde F_A(y)&:= p F_{Y|D=1}(y)  + (1-p-\delta) F_{Y|D_\delta=0}(y),
\end{align}
and
\begin{align}\label{complete_apparent}
F_A(y)&:=\tilde F_A(y)  +\delta  \int_{\mathcal X}F_{Y|D=1,X=x}(y)dF_{X|D=0,D_\delta=1}(x).
\end{align}
Note that, while $F_A(y)$ is a proper distribution function, $\tilde F_A(y)$ is not. We refer to $F_A$ as the apparent distribution, and to $\tilde F_A$ as the ``incomplete'' apparent distribution. By Assumption \ref{assumption_regularity}, both $F_A(y)$ and $\tilde F_A(y)$ are strictly increasing and continuous for all $y$ such that $0<F_A(y)<1$ and $0<\tilde F_A(y)<1$. The bounds on $F_{Y(1)|D=0,D_\delta=1}(y)$ allow us to bound $F_{Y_{D_\delta}}(y)$. The upper bound is
\begin{align*}
F_{Y_{D_\delta}}(y) &\leq  pF_{Y|D=1}(y)  + (1-p-\delta) F_{Y|D_\delta=0}(y)  \\
&+\delta \min\left\{1,\int_{\mathcal X}F_{Y|D=1,X=x}(y)dF_{X|D=0,D_\delta=1}(x)+c\right\}\\
&=\min\left\{\tilde F_A(y)+\delta ,F_A(y)+\delta c\right\}.
\end{align*}
Similarly, the lower bound is
\begin{align*}
F_{Y_{D_\delta}}(y) &\geq  pF_{Y|D=1}(y)  + (1-p-\delta) F_{Y|D_\delta=0}(y)  \\
&+\delta \max\left\{0, \int_{\mathcal X}F_{Y|D=1,X=x}(y)dF_{X|D=0,D_\delta=1}(x)-c\right\}\\
&=\max\left\{\tilde F_A(y) ,F_A(y)-\delta c\right\}.
\end{align*}
Let $F_{Y_{D_\delta}}^{-1}(\tau)$ be that $\tau$-quantile of $Y_{D_\delta}$ which is unique by Assumption \ref{assumption_regularity}. Then
\begin{align*}
\max\left\{\tilde F_A(F_{Y_{D_\delta}}^{-1}(\tau)) ,F_A(F_{Y_{D_\delta}}^{-1}(\tau))-\delta c\right\} \leq \tau \leq \min\left\{\tilde F_A(F_{Y_{D_\delta}}^{-1}(\tau))+\delta ,F_A(F_{Y_{D_\delta}}^{-1}(\tau))+\delta c\right\}.
\end{align*}
For two strictly increasing continuous functions $f$ and $g$, we have that (this is best seen graphically):
\begin{align*}
\left[\min\left\{ f,g\right\}\right]^{-1}(\tau)=\max\left\{ f^{-1}(\tau),g^{-1}(\tau)\right\}.
\end{align*} 
and 
\begin{align*}
\left[\max\left\{ f,g\right\}\right]^{-1}(\tau)=\min\left\{ f^{-1}(\tau),g^{-1}(\tau)\right\}.
\end{align*} 
Moreover, for $F_A$, define the inverse to be $F_A^{-1}(t)$, where $F_A(F_A^{-1}(t))=t$ for $t\in(0,1)$, $F_A^{-1}(t)=-\infty$ if $t\leq 0$, and $F_A^{-1}(t)=\infty$ if $t\geq 1$. For the ``incomplete'' apparent distribution $\tilde F_A$, we have that $\tilde F_A:\mathbb R\mapsto (0,1-\delta)$. Thus, $\tilde F_A^{-1}(t)$ is defined as usual for $t\in (0,1-\delta)$. Whereas, $\tilde F_A^{-1}(t)=-\infty$ if $t\leq 0$, and $\tilde F_A^{-1}(t)=\infty$ if $t\geq 1-\delta$. Note that $\tilde F_A(y)\leq F_A(y).$

Thus, we can bound $F_{Y_{D_\delta}}^{-1}(\tau)$ by 
\begin{align*}
\theta_L:=\max\left\{\tilde F_A^{-1}(\tau-\delta) ,F_A^{-1}(\tau-\delta c)\right\} \leq F_{Y_{D_\delta}}^{-1}(\tau) \leq \min\left\{\tilde F_A^{-1}(\tau) ,F_A^{-1}(\tau+\delta c)\right\}=:\theta_H.
\end{align*}

First suppose that $\theta_L>-\infty$ and $\theta_H<\infty.$ To show that the bounds are sharp we follow the approach of the proof of Lemma 2.1 of \cite{Kline2013}. For any value that $F_{Y_{D_\delta}}^{-1}(\tau)$ may take in $ [\theta_L,\theta_H]$, we need to show that there exists a function $F^*_\delta:\mathbb R\mapsto (0,1)$ that satisfies the conditions imposed on $F_{Y(1)|D=0,D_\delta=1}$. Namely: 
\begin{enumerate}
\item\label{condition_1} $\tau = pF_{Y|D=1}(F_{Y_{D_\delta}}^{-1}(\tau))  + (1-p-\delta) F_{Y|D_\delta=0}(F_{Y_{D_\delta}}^{-1}(\tau))  +\delta F^*_\delta(F_{Y_{D_\delta}}^{-1}(\tau))$;
\item\label{condition_2} $\sup_{y\in\mathbb R}\left |F^*_\delta(y)-\int_{\mathcal X}F_{Y|D=1,X=x}(y)dF_{X|D=0,D_\delta=1}(x)\right |\leq c$;
\item\label{condition_3}  $F^*_\delta(y)$ is continuous and strictly increasing in $\mathbb R$ for all $y$ such that $0<F^*_\delta(y)<1.$
\end{enumerate}
By condition \ref{condition_1}., we can express $F^*_\delta(F_{Y_{D_\delta}}^{-1}(\tau))$ as
\begin{align*}
F^*_\delta(F_{Y_{D_\delta}}^{-1}(\tau)) = \frac{\tau - pF_{Y|D=1}(F_{Y_{D_\delta}}^{-1}(\tau))  - (1-p-\delta) F_{Y|D_\delta=0}(F_{Y_{D_\delta}}^{-1}(\tau))  }{\delta}
\end{align*}
To verify condition \ref{condition_2}., we start by checking that 
\begin{align*}
\left |F^*_\delta(F_{Y_{D_\delta}}^{-1}(\tau))-\int_{\mathcal X}F_{Y|D=1,X=x}(F_{Y_{D_\delta}}^{-1}(\tau))dF_{X|D=0,D_\delta=1}(x)\right |\leq c.
\end{align*}
We write
\begin{align*}
&F^*_\delta(F_{Y_{D_\delta}}^{-1}(\tau))-\int_{\mathcal X}F_{Y|D=1,X=x}(F_{Y_{D_\delta}}^{-1}(\tau))dF_{X|D=0,D_\delta=1}(x)\\ &=  \frac{\tau - pF_{Y|D=1}(F_{Y_{D_\delta}}^{-1}(\tau))  - (1-p-\delta) F_{Y|D_\delta=0}(F_{Y_{D_\delta}}^{-1}(\tau))}{\delta}  \\&-\int_{\mathcal X}F_{Y|D=1,X=x}(F_{Y_{D_\delta}}^{-1}(\tau))dF_{X|D=0,D_\delta=1}(x)\\
&=\frac{\tau - F_A(F_{Y_{D_\delta}}^{-1}(\tau))}{\delta}.
\end{align*}
Because $F_A$ is increasing, and $F_{Y_{D_\delta}}^{-1}(\tau) \in [\theta_L,\theta_H]$, we need to check all the possible limiting cases.
\begin{enumerate}[$(i)$]
\item $F_{Y_{D_\delta}}^{-1}(\tau)  =  F_A^{-1}(\tau-\delta c)$. Then,
\begin{align*}
\frac{\tau - F_A(  F_A^{-1}(\tau-\delta c))}{\delta} = c.
\end{align*}
\item $F_{Y_{D_\delta}}^{-1}(\tau) = \tilde F_A^{-1}(\tau-\delta)$. Then,
\begin{align*}
\frac{\tau - F_A( \tilde F_A^{-1}(\tau-\delta))}{\delta}\leq \frac{\tau - F_A(  F_A^{-1}(\tau-\delta c))}{\delta} = c,
\end{align*}
where the inequality follows from $\tilde F_A^{-1}(\tau-\delta)\geq F_A^{-1}(\tau-\delta c)$.
\item $F_{Y_{D_\delta}}^{-1}(\tau) = F_A^{-1}(\tau+\delta c)$. Then,
\begin{align*}
\frac{\tau - F_A(  F_A^{-1}(\tau+\delta c))}{\delta} =- c.
\end{align*}
\item $F_{Y_{D_\delta}}^{-1}(\tau) = \tilde F_A^{-1}(\tau)$. Then,
\begin{align*}
\frac{\tau - F_A( \tilde F_A^{-1}(\tau))}{\delta}\geq \frac{\tau - F_A(  F_A^{-1}(\tau+\delta c))}{\delta} =- c,
\end{align*}
where the inequality follows from $\tilde F_A^{-1}(\tau)\leq F_A^{-1}(\tau+\delta c)$.
\end{enumerate}
Thus, we have that $\left |F^*_\delta(F_{Y_{D_\delta}}^{-1}(\tau))-\int_{\mathcal X}F_{Y|D=1,X=x}(F_{Y_{D_\delta}}^{-1}(\tau))dF_{X|D=0,D_\delta=1}(x)\right |\leq c$ whenever $F_{Y_{D_\delta}}^{-1}(\tau)\in[\theta_L,\theta_H]$. 

Now we will construct $F^*_\delta(y)$ such that conditions \ref{condition_2}. and \ref{condition_3}. also hold. The starting point is the given value $F^*_\delta(F_{Y_{D_\delta}}^{-1}(\tau))$. To alleviate notation, define
\begin{align*}
\tilde F(y) := \int_{\mathcal X}F_{Y|D=1,X=x}(y)dF_{X|D=0,D_\delta=1}(x)
\end{align*}

Suppose first that 
\begin{align*}
F^*_\delta(F_{Y_{D_\delta}}^{-1}(\tau)) \leq \tilde F(F_{Y_{D_\delta}}^{-1}(\tau)).
\end{align*}
Then, we define $F^*_\delta(y)$ as
\begin{align*}
F^*_\delta(y)= 
     \begin{cases}
\max \left\{ 0, \tilde F(y) - \left [ \tilde F(F_{Y_{D_\delta}}^{-1}(\tau))- F^*_\delta(F_{Y_{D_\delta}}^{-1}(\tau))\right]  \right\} \text{ for }y< F_{Y_{D_\delta}}^{-1}(\tau)\\
 \tilde F(y) - \omega (y-F_{Y_{D_\delta}}^{-1}(\tau))\left [ \tilde F(F_{Y_{D_\delta}}^{-1}(\tau))- F^*_\delta(F_{Y_{D_\delta}}^{-1}(\tau))\right]   \text{ for }y\geq F_{Y_{D_\delta}}^{-1}(\tau)\\
     \end{cases}
\end{align*}
where $\omega(y)$ is a continuous function that satisfies: $\omega(0)=1$, $\lim_{y\to\pm\infty} \omega(y)=0$, and it is strictly decreasing for $y>0$, and strictly increasing for $y<0$. We can take $\omega(y)=\exp(-y^2)$. The role of $\omega(y)$ is that of a weighting function. It brings $F^*_\delta(y)$ closer to $ \tilde F(y)$ as we move to the right of $F_{Y_{D_\delta}}^{-1}(\tau)$. Since $\omega(y)\to 0$ as $y\to 0$, eventually the gap disappears. For the values of $y$ to the left of $F_{Y_{D_\delta}}^{-1}(\tau)$, we simply subtract the gap and use the $\max$ to keep the function non-negative. This function $F^*_\delta(y)$ satisfies \ref{condition_1}., \ref{condition_2}., and \ref{condition_3}.

If, instead, 
\begin{align*}
F^*_\delta(F_{Y_{D_\delta}}^{-1}(\tau)) \geq \tilde F(F_{Y_{D_\delta}}^{-1}(\tau))
\end{align*}
then
\begin{align*}
F^*_\delta(y)= 
     \begin{cases}
 \tilde F(y) + \omega (y-F_{Y_{D_\delta}}^{-1}(\tau))\left [ F^*_\delta(F_{Y_{D_\delta}}^{-1}(\tau)) -\tilde F(F_{Y_{D_\delta}}^{-1}(\tau))   \right]   \text{ for }y\leq F_{Y_{D_\delta}}\\
 \min \left\{ 1, \tilde F(y) + \left [ F^*_\delta(F_{Y_{D_\delta}}^{-1}(\tau))-\tilde F(F_{Y_{D_\delta}}^{-1}(\tau)) \right]  \right\} \text{ for }y> F_{Y_{D_\delta}}^{-1}(\tau)\\
     \end{cases}
\end{align*}
Here, $F^*_\delta(y)$ lies above $ \tilde F(y)$ and we use the weight function to close the gap as $y\to -\infty.$ For the values of $y>F_{Y_{D_\delta}}^{-1}(\tau)$ we simply add the gap (we are subtracting a negative number) and use the $\min$ to avoid the function from being greater than 1. This function $F^*_\delta(y)$ satisfies \ref{condition_1}., \ref{condition_2}., and \ref{condition_3}.

Therefore, sharp bounds for the global effect $G_{\tau, D_\delta}:=F_{Y_{D_\delta}}^{-1}(\tau)-F_{Y}^{-1}(\tau)$, are given by
\begin{align*}
\max\left\{\tilde F_A^{-1}(\tau-\delta) ,F_A^{-1}(\tau-\delta c)\right\} - F_{Y}^{-1}(\tau) \leq G_{\tau, D_\delta}\leq \min\left\{\tilde F_A^{-1}(\tau) ,F_A^{-1}(\tau+\delta c)\right\}-F_{Y}^{-1}(\tau).
\end{align*}
\end{proof}

\begin{proof}[Proof of Theorem \ref{theorem_existence}]

Recall that by \eqref{count_dist_dec}, the counterfactual distribution of a given policy $D_\delta$ that satisfies assumption \ref{assumption_policy} is 
\begin{align*}
F_{Y_{D_\delta}}(y) &= pF_{Y|D=1}(y)  + (1-p-\delta) F_{Y|D_\delta=0}(y)  +\delta F_{Y(1)|D=0,D_\delta=1}(y)
\end{align*}
Then, we can write the quotient as
\begin{align*}
\frac{F_{Y_{D_\delta}}(y)-F_Y(y)}{\delta} &= (1-p) \frac{F_{Y|D_\delta=0}(y)-F_{Y|D=0}(y)}{\delta} - F_{Y|D_\delta=0}(y)  + F_{Y(1)|D=0,D_\delta=1}(y).
\end{align*}
By Assumption \ref{assumption_limit_distribution}, we can take the limit $\delta\downarrow 0$ to obtain
\begin{align*}
\lim_{\delta\downarrow 0}\frac{F_{Y_{D_\delta}}(y)-F_Y(y)}{\delta} &= (1-p)\frac{\partial F_{Y|D_\delta=0}(y)}{\partial \delta}\bigg|_{\delta=0} - F_{Y|D=0}(y)  + F_{Y(1)|\partial D}(y)\\
&=:\dot F_{Y,\mathcal D}(y).
\end{align*}
The map $y\mapsto\dot F_{Y,\mathcal D}(y)$ is continuous by Assumption \ref{assumption_limit_distribution}. Now we aim to strengthen the pointwise result to a uniform result. For that, we write:
\begin{align*}
\sup_{y\in\mathbb R} \left|\frac{F_{Y_{D_\delta}}(y)-F_{Y}(y)}{\delta}-\dot F_{Y,\mathcal D}(y)\right| & \leq (1-p)\sup_{y\in\mathbb R} \left|\frac{F_{Y|D_\delta=0}(y)-F_{Y|D=0}(y)}{\delta}-\frac{\partial F_{Y|D_\delta=0}(y)}{\partial \delta}\bigg|_{\delta=0} \right|\\
&+\sup_{y\in\mathbb R} \left|F_{Y|D_\delta=0}(y)  - F_{Y|D=0}(y) \right| \\
&+ \sup_{y\in\mathbb R} \left|F_{Y(1)|D=0,D_\delta=1}(y) - F_{Y(1)|\partial D}(y) \right| 
\end{align*}
Assumption \ref{assumption_limit_distribution} ensures the uniform convergence of the first and third summands. The second summand convergence uniformly because it is analogous to weak convergence to a continuous distribution function (see Exercise 5.(b) in Chapter 8 of \cite{resnick}).

Finally, let $\Gamma_\tau[F]$ be the $\tau$-quantile of $F$. The Hadamard derivative at $F$ is (see Lemma 21.3 in \cite{vandervaart2000})
\begin{align*}
\Gamma_{\tau,F}'[h]=-\frac{h(F^{-1}(\tau))}{f(F^{-1}(\tau))}.
\end{align*}
for any $h\in \mathbb D[-\infty,\infty]$ continuous at $F^{-1}(\tau)$. We write the marginal effect as
\begin{align*}
M_{\tau,\mathcal D}&=\lim_{\delta\to 0}\frac{\Gamma_\tau\left[F_{Y_{D_\delta}}\right]-\Gamma_\tau[F_{Y}]}{\delta}=\lim_{\delta\to 0}\frac{\Gamma_\tau\left[F_{Y_{D_0}} + \delta \left(\frac{F_{Y_{D_\delta}}-F_{Y_{D_0}}}{\delta}\right)\right]-\Gamma_\tau[F_{Y}]}{\delta}\\
&=\lim_{\delta\to 0}\frac{\Gamma_\tau\left[F_Y + \delta \left(\frac{F_{Y_{D_\delta}}-F_{Y}}{\delta}\right)\right]-\Gamma_\tau[F_{Y}]}{\delta}\\
&=\Gamma_{\tau,F_Y}'[\dot F_{Y,\mathcal D}]=\frac{\dot F_{Y,\mathcal D}(F_Y^{-1}(\tau))}{f_Y(F_Y^{-1}(\tau))}.
\end{align*}

The third equality follows from $F_{Y_{D_0}}=F_Y$. This is because by \eqref{count_dist_dec}, and the fact we have assumed that $\lim_{\delta\to 0}F_{Y|D_\delta=0}=F_{Y|D=0}$ and that $\lim_{\delta\to 0}F_{Y(1)|D=0,D_\delta=1}$ is well defined. The fourth equality follows from
\begin{align*}
\lim_{\delta \downarrow 0}\sup_{y\in\mathcal Y} \left|\frac{F_{Y_{D_\delta}}(y)-F_{Y}(y)}{\delta}-\dot F_{Y,\mathcal D}(y)\right|=0,
\end{align*}
which is required by Lemma 21.3 in \cite{vandervaart2000}.

\end{proof}

\begin{proof}[Details of Example \ref{example_existence_of_marg}]

Let $D=\mathds 1\left\{ V\geq 0 \right\}$ for $V$ continuously distributed around 0, and, for a sequence $\delta\downarrow 0$, let $D_\delta=\mathds 1\left\{ V\geq F_V^{-1}(F_V(0)-\delta) \right\}$. It follows that $p:=\Pr(D=1)=1-F_V(0)$ and that $\Pr(D_\delta=1) = p+\delta$. 
Suppose that $\lim_{\delta\to 0}\Pr(Y\leq y|V\leq F_V^{-1}(F_V(0)-\delta))=\Pr(Y\leq y|V\leq 0)$. Then
\begin{align*}
\lim_{\delta\to 0}F_{Y|D_\delta=0}(y)&=\lim_{\delta\to 0}\Pr(Y\leq y|V\leq F_V^{-1}(F_V(0)-\delta))\\
&= \Pr(Y\leq y|V\leq 0)\\
&=F_{Y|D=0}(y).
\end{align*}
Now, we have
\begin{align*}
F_{Y|D_\delta=0}(y)-F_{Y|D=0}(y)&= \frac{\Pr(Y\leq y,V\leq F_V^{-1}(F_V(0)-\delta))}{p+\delta} - \frac{\Pr(Y\leq y,V\leq 0)}{p}\\
=&\frac{p\Pr(Y\leq y,V\leq F_V^{-1}(F_V(0)-\delta)) - (p+\delta)\Pr(Y\leq y,V\leq 0)}{(p+\delta)p}\\
&=\frac{1}{p+\delta}\left [\Pr(Y\leq y,V\leq F_V^{-1}(F_V(0)-\delta))- \Pr(Y\leq y,V\leq 0)  \right]\\
&-\frac{\delta}{p(p+\delta)} \Pr(Y\leq y,V\leq 0)
\end{align*}
Consider the first term:
\begin{align*}
&\frac{\Pr(Y\leq y,V\leq F_V^{-1}(F_V(0)-\delta))- \Pr(Y\leq y,V\leq 0) }{\delta}\\&=\frac{F_V^{-1}(F_V(0)-\delta)}{\delta}\\
&\times
\frac{\Pr(Y\leq y,V\leq F_V^{-1}(F_V(0)-\delta))- \Pr(Y\leq y,V\leq 0)}{F_V^{-1}(F_V(0)-\delta)} 
\end{align*}
Now, provided that $f_V(0)\neq 0$,
\begin{align*}
\lim_{\delta\to 0}\frac{F_V^{-1}(F_V(0)-\delta)}{\delta} = -\frac{1}{f_V(0)}.
\end{align*}
For the other term we have
\begin{align*}
\frac{\Pr(Y\leq y,V\leq F_V^{-1}(F_V(0)-\delta))- \Pr(Y\leq y,V\leq 0)}{F_V^{-1}(F_V(0)-\delta)} &= \frac{\int^{F_V^{-1}(F_V(0)-\delta)}_{0}\int_{-\infty}^yf_{Y,V}(u,v)du dv}{F_V^{-1}(F_V(0)-\delta)}\\
&\to \int_{-\infty}^yf_{Y,V}(u,0)du
\end{align*}
Combining all together, we obtain
\begin{align*}
\lim_{\delta\to 0}\frac{F_{Y|D_\delta=0}(y)-F_{Y|D=0}(y)}{\delta}=-\left(\frac{1}{f_V(0)}+\frac{1}{p}\right)F_{Y|D=0}(y)
\end{align*}
Thus we can take
\begin{align*}
\frac{\partial F_{Y|D_\delta=0}(y)}{\partial \delta}\bigg|_{\delta=0}=-\left(\frac{1}{f_V(0)}-\frac{1}{p}\right)F_{Y|D=0}(y)
\end{align*}
which will be continuous as long as $F_{Y|D=0}(y)$ is continuous (which we assumed).
Now, consider
\begin{align*}
F_{Y(1)|D=0,D_\delta=1}(y) &=\Pr(Y(1)\leq y|V\leq 0, V\geq F_V^{-1}(F_V(0)-\delta))\\
&=\frac{\Pr(Y(1)\leq y,V\leq 0, V\geq F_V^{-1}(F_V(0)-\delta))}{\Pr(V\leq 0, V\geq F_V^{-1}(F_V(0)-\delta))}\\
&=\frac{\int_{F_V^{-1}(F_V(0)-\delta)}^0\int_{-\infty}^0f_{Y(1),V}(u,v)dudv}{\delta}\\
&=\frac{F_V^{-1}(F_V(0)-\delta)}{\delta}\frac{\int_{F_V^{-1}(F_V(0)-\delta)}^0\int_{-\infty}^0f_{Y(1),V}(u,v)dudv}{F_V^{-1}(F_V(0)-\delta)}.
\end{align*}
Under some mild regularity conditions:
\begin{align*}
\lim_{\delta\to0} F_{Y(1)|D=0,D_\delta=1}(y) = F_{Y(1)|V=0}(y)
\end{align*}
Thus, we can take
\begin{align*}
F_{Y(1)|\partial D}(y)=F_{Y(1)|V=0}(y).
\end{align*}
\end{proof}

\begin{proof}[Proof of Theorem \ref{thm_bounds_marginal}]

To bound the marginal effect, which we assume that exists, we aim to take the limit when $\delta\downarrow 0$ of the rescaled bounds of the global effect. That is, we aim to bound the marginal effect by the limiting versions of  
\begin{align*}
\frac{\max\left\{\tilde F_A^{-1}(\tau-\delta) ,F_A^{-1}(\tau-\delta c)\right\} - F_{Y}^{-1}(\tau) }{\delta},
\end{align*}
and 
\begin{align*}
\frac{\min\left\{\tilde F_A^{-1}(\tau) ,F_A^{-1}(\tau+\delta c)\right\}-F_{Y}^{-1}(\tau)}{\delta}. 
\end{align*}

We would like to show that, when $\delta\downarrow 0$, we have
\begin{align*}
\lim_{\delta\downarrow 0}\frac{\max\left\{\tilde F_A^{-1}(\tau-\delta) ,F_A^{-1}(\tau-\delta c)\right\} - F_{Y}^{-1}(\tau)}{\delta} \leq M_{\tau, \mathcal D}\leq \lim_{\delta\downarrow 0}\frac{\min\left\{\tilde F_A^{-1}(\tau) ,F_A^{-1}(\tau+\delta c)\right\}-F_{Y}^{-1}(\tau)}{\delta}
\end{align*}
The reason we take the limit from the right, \emph{i.e.} $\delta\downarrow 0$, is because twofold the ordinary limit might fail to exist. This is the reason why we assume that the sequence of policies consists of $\delta>0$ moving towards 0. This is part of Assumption \ref{assumption_policy}. For $\delta>0$, we can write
\begin{align*}
\frac{\max\left\{\tilde F_A^{-1}(\tau-\delta) ,F_A^{-1}(\tau-\delta c)\right\} - F_{Y}^{-1}(\tau)}{\delta} &= \frac{\max\left\{\tilde F_A^{-1}(\tau-\delta) - F_{Y}^{-1}(\tau) , F_A^{-1}(\tau-\delta c)- F_{Y}^{-1}(\tau)\right\} }{\delta}\\
&= \max\left\{  \frac{\tilde F_A^{-1}(\tau-\delta) - F_{Y}^{-1}(\tau)}{\delta}, \frac{F_A^{-1}(\tau-\delta c)- F_{Y}^{-1}(\tau)}{\delta}          \right\}.
\end{align*}
The last step, bringing $1/\delta$ into the $\max$, is only valid because $\delta>0$. Similarly, for $\delta>0$, we have
\begin{align*}
\frac{\min\left\{\tilde F_A^{-1}(\tau) ,F_A^{-1}(\tau+\delta c)\right\}-F_{Y}^{-1}(\tau)}{\delta} = \min\left\{\frac{\tilde F_A^{-1}(\tau)-F_{Y}^{-1}(\tau) }{\delta},\frac{F_A^{-1}(\tau+\delta c)-F_{Y}^{-1}(\tau)}{\delta}\right\}
\end{align*}
Now, since the $\max$ and $\min$ are continuous, to pass the limit as $\delta\downarrow 0$ we need to show that the maps 
\begin{align*}
&\delta \mapsto \tilde F_A^{-1}(\tau-\delta), \\
&\delta \mapsto F_A^{-1}(\tau-\delta c), \\
&\delta \mapsto  \tilde F_A^{-1}(\tau) ,\\
&\delta \mapsto  F_A^{-1}(\tau+\delta c),
\end{align*}
are differentiable at $\delta=0$. We need to show differentiability because it follows from equations \eqref{incomplete_apparent} and \eqref{complete_apparent} that all these maps evaluated at $\delta=0$ are equal to $F_Y^{-1}(\tau)$. Provided they are differentiable (something we show below), we obtain:
\begin{align*}
\max\left\{\frac{\tilde F_A^{-1}(\tau-\delta)}{\partial \delta}\bigg |_{\delta=0} ,\frac{\partial F_A^{-1}(\tau-\delta c)}{\partial \delta}\bigg |_{\delta=0}   \right\} \leq M_{\tau, \mathcal D}\leq  \min\left\{\frac{\partial \tilde F_A^{-1}(\tau)}{\partial \delta}\bigg |_{\delta=0} ,\frac{\partial F_A^{-1}(\tau+\delta c)}{\partial \delta}\bigg |_{\delta=0}\right\}
\end{align*}

We start with the map $\delta\mapsto F_A^{-1}(\tau-\delta c)$. By equations \eqref{incomplete_apparent} and \eqref{complete_apparent}, when $\delta=0$ , $F_A=F_Y$. So, we want to find the limit as $\delta\downarrow 0$ of 
\begin{equation}\label{eqn_proof_me}
\frac{F^{-1}_{ A}(\tau-\delta c)-F_Y^{-1}(\tau)}{\delta}.
\end{equation}
We note that $\delta$ plays a triple role in the previous expression. Using \eqref{complete_apparent}, we write $F_{ A}(y)$ as
\begin{align}
F_{ A}(y;\delta_1,\delta_2)= p F_{Y|D=1}(y)  + (1-p-\delta_1) F_{Y|D_{\delta_2}=0}(y)  +   \delta_1  \int_{\mathcal X}F_{Y|D=1,X=x}(y)dF_{X|D=0,D_{\delta_2}=1}(x)
\end{align}
and we define
\begin{align*}
g(\delta_1,\delta_2,\delta_3)=F^{-1}_{A}(\tau-\delta_3 c;\delta_1,\delta_2)
\end{align*}
The map $\delta_1\mapsto g(\delta_1,\delta_2, \delta_3)$ for a fixed $\delta_2$ and $\delta_3$ is the composition\footnote{For $[a,b]\subset [-\infty,\infty]$,  $\mathbb D[a,b]$ is the set of all real-valued cadlag functions: right continuous with left limits everywhere in $[a,b]$.  $\mathbb D[a,b]$ is equipped with the uniform norm $\|\cdot\|_{\infty}$.}
\begin{align*}
\delta_1 \in\mathbb R\overset{h}{\mapsto} F_{A}(\cdot;\delta_1,\delta_2)\in \mathbb D[-\infty,\infty]\overset{\Gamma}{\mapsto} F^{-1}_{A}(\tau-\delta_3 c;\delta_1,\delta_2)\in\mathbb R. 
\end{align*}
The first map $h:\delta_1 \in\mathbb R\mapsto F_{A}(\cdot;\delta_1,\delta_2)\in \mathbb D[-\infty,\infty]$ has Hadamard derivative given by 
\begin{align*}
\int_{\mathcal X}F_{Y|D=1,X=x}(\cdot)dF_{X|D=0,D_{\delta_2}=1}(x)-F_{Y|D_{\delta_2}=0}(\cdot).
\end{align*}
The second map has Hadamard derivative given by (See Lemma 21.3 in \cite{vandervaart2000})
\begin{align*}
\Gamma'_{F_{A}(\cdot;\delta_1,\delta_2)}[G]=-\frac{G( F^{-1}_{A}(\tau-\delta_3c;\delta_1,\delta_2))}{f_{A}( F^{-1}_{A}(\tau-\delta_3c;\delta_1,\delta_2);\delta_1,\delta_2)}.
\end{align*}
for $G\in \mathbb D[-\infty,\infty]$ continuous at $ F^{-1}_{A}(\tau-\delta_3c;\delta_1,\delta_2)$.  Then, the derivative of the composite map $\delta_1\mapsto \Gamma \circ h (\delta_1)$ is $\Gamma'_{F_{A}(\cdot;\delta_1,\delta_2)}[h'(\delta_1)]$, which is at $\delta_2=\delta_3=0$
\begin{align*}
\frac{\partial  F^{-1}_{A}(\tau;\delta_1,0)}{\partial \delta_1}=-\frac{\int_{\mathcal X}F_{Y|D=1,X=x}(F^{-1}_{A}(\tau;\delta_1,0))dF_{X|\partial D}(x)-F_{Y|D=0}(F^{-1}_{A}(\tau;\delta_1,0))}{f_{A}( F^{-1}_{A}(\tau;\delta_1,0);\delta_1,0)}.
\end{align*}
which is continuous at $\delta_1=0$.

The derivative of the second map $\delta_2\mapsto g(\delta_1,\delta_2,\delta_3)$, for a fixed $\delta_1=\delta_3=0$ is exactly 0 because as \eqref{complete_apparent} shows, $\delta_1$ multiplies the expression
\begin{align*}
-F_{Y|D_{\delta_2}}(y)+\int_{\mathcal X}F_{Y|D=1,X=x}(y)dF_{X|D=0,D_{\delta_2}=1}(x).
\end{align*}
All of the terms in this expression are differentiable at $\delta_2=0$ by Assumption \ref{assumption_indifference}.

The derivative of the third map $\delta_3\mapsto g(\delta_1,\delta_2,\delta_3)$, for a fixed $\delta_1=\delta_2=0$, can be obtained via the identity
\begin{align*}
F_{A}\left(F^{-1}_{A}(\tau-\delta_3 c;\delta_1,\delta_2)\delta_1,\delta_2\right) =  \tau-\delta_3 c.
\end{align*}
Differentiating through with respect to $\delta_3$, we obtain at $\delta_1=\delta_2=0$
\begin{align*}
\frac{\partial  F^{-1}_{A}(\tau-\delta_3 c;0,0)}{\partial \delta_3}=-\frac{c}{f_Y(F_Y^{-1}(\tau-\delta_3 c))}.
\end{align*}
which is continuous with respect to $\delta_3$. 

Therefore, all the partial derivatives of the map $(\delta_1,\delta_2,\delta_3)\mapsto g(\delta_1,\delta_2,\delta_3)$ exist and are continuous, hence the limit in \eqref{eqn_proof_me} exists and is equal to
\begin{align*}
\lim_{\delta\to 0}\frac{F^{-1}_{ A}(\tau-\delta c)-F_Y^{-1}(\tau)}{\delta}&=\frac{\partial g(\delta_1,0,0)}{\partial \delta_1}\bigg|_{\delta_1=0}+\frac{\partial g(0,\delta_2,0)}{\partial \delta_2}\bigg|_{\delta_2=0}+\frac{\partial g(0,0,\delta_3)}{\partial \delta_3}\bigg|_{\delta_3=0}\\
&=-\frac{\int_{\mathcal X}F_{Y|D=1,X=x}(F^{-1}_{Y}(\tau))dF_{X|\partial D}(x)-F_{Y|D=0}(F^{-1}_{Y}(\tau))}{f_{Y}( F^{-1}_{Y}(\tau))}-\frac{c}{f_Y(F_Y^{-1}(\tau))}.
\end{align*}

Analogous arguments can be used to obtain that
\begin{align*}
\lim_{\delta\to 0}\frac{F^{-1}_{ A}(\tau+\delta c)-F_Y^{-1}(\tau)}{\delta}
&=-\frac{\int_{\mathcal X}F_{Y|D=1,X=x}(F^{-1}_{Y}(\tau))dF_{X|\partial D}(x)-F_{Y|D=0}(F^{-1}_{Y}(\tau))}{f_{Y}( F^{-1}_{Y}(\tau))}+\frac{c}{f_Y(F_Y^{-1}(\tau))}.
\end{align*}

Now we turn our attention to the incomplete apparent distribution $\tilde F_A$ and the map $\delta\mapsto \tilde F^{-1}_A(\tau-\delta)$. By equation \eqref{incomplete_apparent}, the incomplete apparent distribution is given by
\begin{align*}
\tilde F_A(y;\delta_1,\delta_2)&:= p F_{Y|D=1}(y)  + (1-p-\delta_1) F_{Y|D_{\delta_2}=0}(y),
\end{align*}
where, again, we index by $\delta_1$ and $\delta_2$ to emphasize the dual role played by $\delta.$ As before, we define
\begin{align*}
g(\delta_1,\delta_2,\delta_3)=\tilde F^{-1}_{A}(\tau-\delta_3 ;\delta_1,\delta_2).
\end{align*}
The map $\delta_1\mapsto g(\delta_1,\delta_2, \delta_3)$ for a fixed $\delta_2$ and $\delta_3$ is the composition
\begin{align*}
\delta_1 \in\mathbb R\overset{h}{\mapsto} \tilde F_{A}(\cdot;\delta_1,\delta_2)\in \mathbb D[-\infty,\infty]\overset{\Gamma}{\mapsto} \tilde F^{-1}_{A}(\tau-\delta_3 ;\delta_1,\delta_2)\in\mathbb R. 
\end{align*}
The first map $h:\delta_1 \in\mathbb R\mapsto \tilde F_{A}(\cdot;\delta_1,\delta_2)\in \mathbb D[-\infty,\infty]$ has Hadamard derivative given by 
\begin{align*}
- F_{Y|D_{\delta_2}=0}(\cdot).
\end{align*}
The second map has Hadamard derivative given by 
\begin{align*}
\Gamma'_{\tilde F_{A}(\cdot;\delta_1,\delta_2)}[G]=-\frac{G( \tilde F^{-1}_{A}(\tau-\delta_3c;\delta_1,\delta_2))}{\tilde f_{A}( \tilde F^{-1}_{A}(\tau-\delta_3;\delta_1,\delta_2);\delta_1,\delta_2)}.
\end{align*}
Then, the derivative of  $\delta_1\mapsto \Gamma \circ h (\delta_1)$ is $\Gamma'_{\tilde F_{A}(\cdot;\delta_1,\delta_2)}[h'(\delta_1)]$, which is at $\delta_2=\delta_3=0$
\begin{align*}
\frac{\partial  \tilde F^{-1}_{A}(\tau;\delta_1,0)}{\partial \delta_1}=\frac{F_{Y|D=0}( \tilde F^{-1}_{A}(\tau;\delta_1,0))}{\tilde f_{A}( \tilde F^{-1}_{A}(\tau;\delta_1,0);\delta_1,0)}
\end{align*}
which is continuous at $\delta_1=0$.

The derivative of the second map $\delta_2\mapsto g(\delta_1,\delta_2,\delta_3)$, for a fixed $\delta_1=\delta_3=0$ is exactly 0 because, as in the case of the (complete) apparent distribution, as \eqref{incomplete_apparent} shows, $\delta_1$ multiplies the expression
\begin{align*}
F_{Y|D_{\delta_2}=0}(y)
\end{align*}
which is differentiable at $\delta_2=0$ by Assumption \ref{assumption_indifference}.

The derivative of the third map $\delta_3\mapsto g(\delta_1,\delta_2,\delta_3)$, for a fixed $\delta_1=\delta_2=0$, can be obtained via the identity
\begin{align*}
\tilde F_{A}\left(\tilde F^{-1}_{A}(\tau-\delta_3 ;\delta_1,\delta_2)\delta_1,\delta_2\right) =  \tau-\delta_3 .
\end{align*}
Differentiating through with respect to $\delta_3$, we obtain at $\delta_1=\delta_2=0$
\begin{align*}
\frac{\partial  \tilde F^{-1}_{A}(\tau-\delta_3 ;0,0)}{\partial \delta_3}=-\frac{1}{ f_Y(F_Y^{-1}(\tau-\delta_3 ))}.
\end{align*}
which is continuous with respect to $\delta_3$. Therefore, all the partial derivatives of the map $(\delta_1,\delta_2,\delta_3)\mapsto g(\delta_1,\delta_2,\delta_3)$ exist and are continuous, we can take the limit
\begin{align*}
\lim_{\delta\to 0}\frac{\tilde F^{-1}_A(\tau-\delta)-F^{-1}_Y(\tau)}{\delta}&=\frac{\partial g(\delta_1,0,0)}{\partial \delta_1}\bigg|_{\delta_1=0}+\frac{\partial g(0,\delta_2,0)}{\partial \delta_2}\bigg|_{\delta_2=0}+\frac{\partial g(0,0,\delta_3)}{\partial \delta_3}\bigg|_{\delta_3=0}\\
&=\frac{F_{Y|D=0}(F^{-1}_{Y}(\tau))-1}{f_{Y}( F^{-1}_{Y}(\tau))}
\end{align*}
because $F_{Y|D_0=0}=F_{Y|D=0}$ by Assumption \ref{assumption_limit_distribution}.

For the remaining map, $\delta\mapsto \tilde F^{-1}_A(\tau)$, similar arguments show that
\begin{align*}
\lim_{\delta\to 0}\frac{\tilde F^{-1}_A(\tau)-F^{-1}_Y(\tau)}{\delta}
&=\frac{F_{Y|D=0}(F^{-1}_{Y}(\tau))}{f_{Y}( F^{-1}_{Y}(\tau))}.
\end{align*}

Therefore, we have that
\begin{align*}
\theta_{L,M} \leq M_{\tau, \mathcal D} \leq \theta_{U,M}
\end{align*}
where
\begin{align}\label{me_low_bound}
\theta_{L,M}:=\max\left\{\frac{F_{Y|D=0}(F^{-1}_{Y}(\tau))-1}{f_{Y}( F^{-1}_{Y}(\tau))} ,\frac{F_{Y|D=0}(F^{-1}_{Y}(\tau))-\int_{\mathcal X}F_{Y|D=1,X=x}(F^{-1}_{Y}(\tau))dF_{X|\partial D}(x)-c}{f_{Y}( F^{-1}_{Y}(\tau))}  \right\},
\end{align}
and
\begin{align}\label{me_high_bound}
\theta_{U,M}:=\min\left\{\frac{F_{Y|D=0}(F^{-1}_{Y}(\tau))}{f_{Y}( F^{-1}_{Y}(\tau))},\frac{F_{Y|D=0}(F^{-1}_{Y}(\tau))-\int_{\mathcal X}F_{Y|D=1,X=x}(F^{-1}_{Y}(\tau))dF_{X|\partial D}(x)+c}{f_{Y}( F^{-1}_{Y}(\tau))}\right\}.
\end{align}

Now, the question is whether these bounds are sharp. As in the proof of Theorem \ref{thm_bounds_global}, for any $\theta\in[\theta_{L,M},\theta_{U,M}]$, we need to find a sequence of CDFs $F_{\delta}^*(y)$ such that for a counterfactual distribution constructed as
\begin{align*}
F_{Y_\delta}^*(y)=pF_{Y|D=1}(y)  + (1-p-\delta) F_{Y|D_\delta=0}(y)  +\delta F_{\delta}^*(y)
\end{align*}
we have that 
\begin{align*}
\lim_{\delta \downarrow 0}\frac{F_{Y_\delta}^{*-1}(\tau)-F^{-1}(\tau)}{\delta} = \theta
\end{align*}
For a given $\theta$, we can find (more than) a sequence of global effects that converge to $\theta$. For each $\delta$, these global effects are within the bounds provided in Theorem \ref{thm_bounds_global}. For each $\delta$ these bounds are sharp. So that for each $\delta$ we can find $F_{\delta}^*(y)$ that delivers that particular value of the global effect. Thus, that particular sequence will deliver a marginal effect of $\theta$. Hence, the bounds in \eqref{me_low_bound} and \eqref{me_high_bound} are sharp.

\end{proof}

\begin{proof}[Proof of Lemma \ref{qbf_global_lemma}]
The Lemma has four statements. We prove each one at a time.

\begin{enumerate}[$(i)$]

\item  Continuity of the quantile breakdown frontier is ensured by the continuity of each of its components: $F_Y^{-1}(\tau)$ is continuous by Assumption \ref{assumption_regularity}.\ref{assumption_yd_regular}; continuity of $F_A$ is assured by Assumptions \ref{assumption_regularity}.\ref{assumption_yd_regular}, and  \ref{assumption_regularity}.\ref{assumption_yd_regular_2}, and an application of the Dominated Convergence Theorem; and continuity of the collection of conclusions of $\tau \mapsto g_{\tau,L}$ is assumed. 

\item We have $c_{\tau,L}<0$ when $\tau<F_A\left ( F_{Y}^{-1}(\tau)+ g_{\tau,L} \right )$, which in turn implies that $F_A^{-1}(\tau)-F_{Y}^{-1}(\tau)<g_{\tau,L}$. By footnote \ref{footnote}, when $c=0$, the lower and upper bounds collapse to $ F_A^{-1}(\tau)  -F_Y^{-1}(\tau)$, but $ F_A^{-1}(\tau)  -F_Y^{-1}(\tau)< g_{\tau,L}.$ Hence, the target conclusion does not hold under point identification.

\item We have $c_{\tau,L}\geq 0$ when $\tau\geq F_A\left ( F_{Y}^{-1}(\tau)+ g_{\tau,L} \right )$, which implies that 
$g_{\tau,L}\leq F_A^{-1}(\tau)  -F_Y^{-1}(\tau)$, and the conclusion holds under point identification, or $c=0$. To show that for $0\leq c\leq  c_{\tau,L}$, the conclusion holds, we invoke Theorem \ref{thm_bounds_global}, and show that the lower bound of the global effect is decreasing in $c$, and when evaluated at $c= c_{\tau,L}$ is greater than $g_{\tau,L}$. By footnote \ref{footnote}, $\tilde F_A^{-1}(\tau-\delta)\leq F_A^{-1}(\tau)$, and $\delta$ and $c$ are positive, then $c\mapsto \max\left\{\tilde F_A^{-1}(\tau-\delta) ,F_A^{-1}(\tau-\delta c)\right\}$ is weakly decreasing. Now, if $c= c_{\tau,L}$, then 
\begin{align*}
\max\left\{\tilde F_A^{-1}(\tau-\delta ) ,F_A^{-1}(\tau-\delta  c_{\tau,L})\right\}-F_{Y}^{-1}(\tau)&= \max\left\{\tilde F_A^{-1}(\tau-\delta) ,F_Y^{-1}(\tau)+g_{\tau,L}\right\}-F_{Y}^{-1}(\tau)\\
&= \max\left\{\tilde F_A^{-1}(\tau-\delta) -F_{Y}^{-1}(\tau) ,g_{\tau,L}\right\}\\
&\geq g_{\tau,L}.
\end{align*}

\item Let $c^*$ be implicitly defined as: $\tilde F_A^{-1}(\tau-\delta)=F_A^{-1}(\tau-\delta c^*)$. Then, for $c\geq c^*$ (see footnote \ref{footnote}) we have $\tilde F_A^{-1}(\tau-\delta)\geq F_A^{-1}(\tau-\delta c)$, and $c$ plays no role in the bound. So, if $c_{\tau,L}\geq c^*$, then the conclusion holds irrespective of the value of $c$. This means that $\tilde F_A^{-1}(\tau-\delta)\geq F_A^{-1}(\tau-\delta c_{\tau,L})$, which is equivalent to $\tilde F_A^{-1}(\tau-\delta)-F_Y^{-1}(\tau)\geq g_{\tau,L}$.
\end{enumerate}
\end{proof}

\begin{proof}[Proof of Lemma \ref{qbf_marginal_lemma}]
The Lemma has four statements. We prove each one at a time.
\begin{enumerate}[$(i)$]
\item The continuity of $F^{-1}_{Y}(\tau)$, $F_{Y|D=0}(y)$, and $F_{Y|D=1,X=x}(y)$ guaranteed by Assumption \ref{assumption_regularity}.\ref{assumption_yd_regular}, plus a straightforward application of the Dominated Convergence Theorem ensure that $\tau\mapsto c_{\tau,U}^M$ is continuous.
\item $c_{\tau,U}^M$ is negative whenever $\int_{\mathcal X}F_{Y|D=1,X=x}(F^{-1}_{Y}(\tau))dF_{X|\partial D}(x)<F_{Y|D=0}(F^{-1}_{Y}(\tau))$. By Theorem \ref{thm_bounds_marginal}, when $c=0$, the upper and lower bounds collapse to 
\begin{align*}
\frac{F_{Y|D=0}(F^{-1}_{Y}(\tau))-\int_{\mathcal X}F_{Y|D=1,X=x}(F^{-1}_{Y}(\tau))dF_{X|\partial D}(x)}{f_{Y}( F^{-1}_{Y}(\tau))}. 
\end{align*}
If $c_{\tau,U}^M<0$, then 
\begin{align*}
\frac{F_{Y|D=0}(F^{-1}_{Y}(\tau))-\int_{\mathcal X}F_{Y|D=1,X=x}(F^{-1}_{Y}(\tau))dF_{X|\partial D}(x)}{f_{Y}( F^{-1}_{Y}(\tau))}>0
\end{align*}
which implies that the target conclusion $M_{\tau,\mathcal D}\leq 0$ does not hold under point identification.
\item Suppose that $0\leq c\leq c_{\tau,U}^M$. We want to show that this implies that $\theta_{U,M}\leq 0$. By Theorem \ref{thm_bounds_marginal}, the upper bound on the marginal effect is (see \eqref{me_high_bound})
\begin{align*}
\theta_{U,M}:=\min\left\{\frac{F_{Y|D=0}(F^{-1}_{Y}(\tau))}{f_{Y}( F^{-1}_{Y}(\tau))},\frac{F_{Y|D=0}(F^{-1}_{Y}(\tau))-\int_{\mathcal X}F_{Y|D=1,X=x}(F^{-1}_{Y}(\tau))dF_{X|\partial D}(x)+c}{f_{Y}( F^{-1}_{Y}(\tau))}\right\}.
\end{align*}
Since $\frac{F_{Y|D=0}(F^{-1}_{Y}(\tau))}{f_{Y}( F^{-1}_{Y}(\tau))}$ is always strictly positive under Assumption \ref{assumption_regularity}.\ref{assumption_yd_regular}, then we need the second term inside the $\min$ in $\theta_{U,M}$ to be non-positive. For $c=0$, it holds trivially, since we assumed that the conclusions hold at $c=0$, or that $c_{\tau,U}^M>0$. For $0< c\leq  c_{\tau,U}^M$, we have
\begin{align*}
&F_{Y|D=0}(F^{-1}_{Y}(\tau))-\int_{\mathcal X}F_{Y|D=1,X=x}(F^{-1}_{Y}(\tau))dF_{X|\partial D}(x)+c\\
&\leq F_{Y|D=0}(F^{-1}_{Y}(\tau))-\int_{\mathcal X}F_{Y|D=1,X=x}(F^{-1}_{Y}(\tau))dF_{X|\partial D}(x)+c_{\tau,U}^M\\
&\leq 0,
\end{align*}
because $ c_{\tau,U}^M= \int_{\mathcal X}F_{Y|D=1,X=x}(F^{-1}_{Y}(\tau))dF_{X|\partial D}(x)-F_{Y|D=0}(F^{-1}_{Y}(\tau))$.
\end{enumerate}

\end{proof}

\begin{proof}[Proof of Theorem \ref{dist_theta}]

We introduce some new notation related to Assumption \ref{ass_had_dif}. Let $\mathbb D_{\delta}\subset \ell^{\infty}(\mathcal Y)$ denote the set of all restrictions of distribution functions on $\mathbb R$ to $[F_Y^{-1}(\delta)-\varepsilon,F_Y^{-1}(1-\delta)+\varepsilon]$. Additionally, $\mathbb C_{\delta}$ is the set of continuous functions on $[F_Y^{-1}(\delta)-\varepsilon,F_Y^{-1}(1-\delta)+\varepsilon]$. Also, $\mathbb U\mathbb C(\mathcal Y)$ is the set of uniformly continuous functions defined on $\mathcal Y$.

Going back to $c_L$, it can be written as the composition of two maps. The first one is $\phi:\mathbb D(\mathcal Y)\times \mathbb D_\delta\mapsto \mathbb D(\mathcal Y)\times \ell^{\infty}(\delta,1-\delta)$ given by $\phi(H_1,H_2)\mapsto (H_1,H_2^{-1})$. The second one is $\psi:\mathbb D(\mathcal Y)\times \ell^{\infty}(\delta,1-\delta)\mapsto  \mathcal \ell^{\infty}(\delta,1-\delta)$ given by $\psi(H_1,H_2)\mapsto H_1\circ (H_2+g)$. Thus
\begin{align*}
\psi \circ \phi (F_A,F_Y)=F_A ( F_Y^{-1}+g).
\end{align*}

By Assumption \ref{ass_had_dif}.\ref{ass_f_y_quantiles} and Lemma 21.4(i) in \cite{vandervaart2000}, $\phi$ has Hadamard derivative at $(F_A,F_Y)$ tangentially to $\ell^{\infty}(\mathcal Y)\times \mathbb C_\delta$ given by the map
\begin{align*}
 \phi'_ {(F_A,F_Y)}(h_1,h_2)=\left (h_1,-\frac{h_2\circ F_Y^{-1}}{f_Y\circ F_Y^{-1}}\right).
\end{align*}

The second map $\psi:\mathbb D(\mathcal Y)\times \ell^{\infty}(\delta,1-\delta)\mapsto  \mathcal \ell^{\infty}(\delta,1-\delta)$ is given by $\psi(H_1,H_2)\mapsto H_1\circ (H_2+g)$. It has Hadamard derivative tangentially to $\mathbb U\mathbb C(\mathcal Y)\times \ell^{\infty}(\delta,1-\delta)$ at any $H_1\in \mathbb U\mathbb C(\mathcal Y)$ such that its derivative is bounded and uniformly continuous on $\mathcal Y$, and any $H_2\in  \ell^{\infty}(\delta,1-\delta)$. To see, this we combine Lemmas 3.9.25 and 3.9.27 in \cite{vandervaart1996}. Let $\alpha_{t}\to \alpha$ and $\beta_t\to\beta$ in $\mathbb D(\mathcal Y)$ and $\ell^{\infty}(\delta,1-\delta)$ respectively, as $t\to 0$.
\begin{align*}
&\frac{\psi(H_1+t\alpha_t,H_2+t\beta_t)-\psi(H_1,H_2)}{t}-\alpha\circ (H_2+g)-h_1\circ (H_2+g)\cdot  \beta\\
&=\frac{H_1\circ(H_2+g+t\beta_t) +t\alpha_t\circ(H_2+g+t\beta_t) - H_1\circ (H_2+g)}{t}-\alpha\circ (H_2+g)-h_1\circ (H_2+g)\cdot  \beta\\
&=(\alpha_t-\alpha)\circ(H_2+g+t\beta_t) + \alpha \circ(H_2+g+t\beta_t) -\alpha\circ (H_2+g)\\
&+\frac{H_1\circ(H_2+g+t\beta_t)-H_1\circ (H_2+g)}{t}-h_1\circ (H_2+g)\cdot  \beta
\end{align*}

The first term, $(\alpha_t-\alpha)\circ(H_2+g+t\beta_t)$, converges to $0$ in $\mathbb D(\mathcal Y)$ (that is, uniformly) because convergence of $\alpha_{t}\to \alpha$ is uniform. The second term, $\alpha \circ(H_2+g+t\beta_t) -\alpha\circ (H_2+g)$, converges to $0$ in $\mathbb D(\mathcal Y)$ because $\alpha$ is uniformly continuous on $\mathcal Y$. For the last term, fix a $\tau\in(\delta,1-\delta)$. By the mean-value theorem
\begin{align*}
&\frac{H_1(H_2(\tau)+g+t\beta_t(\tau))-H_1 (H_2(\tau)+g)}{t}-h_1 (H_2(\tau)+g)\cdot  \beta(\tau)\\
&=h_1(\varepsilon_{\tau,t})\beta_t(\tau)-h_1 (H_2(\tau)+g)\cdot  \beta(\tau)\\
&=h_1(\varepsilon_{\tau,t})(\beta_t(\tau)-\beta(\tau)) + (h_1(\varepsilon_{\tau,t})-h_1 (H_2(\tau)+g))\cdot \beta(\tau)
\end{align*}

The first term, $h_1(\varepsilon_{\tau,t})(\beta_t(\tau)-\beta(\tau)) $, converges uniformly to $0$ because $h_1$ is bounded on $\mathcal Y$, and $\beta_t$ converges uniformly to $\beta$. The second term converges to $0$ uniformly because $h_1$ is uniformly continuous on $\mathcal Y$.

Hence, by Assumption \ref{ass_had_dif}.\ref{ass_f_a_uniform}, $\psi$ has Hadamard derivative at $(F_A,F_Y^{-1})$ tangentially to $ \mathbb U\mathbb C(\mathcal Y)\times \ell^{\infty}(\delta,1-\delta)$ given by the map
\begin{align*}
 \psi'_ {(F_A,F_Y^{-1})}(h_1,h_2)=h_1\circ (F_Y^{-1}+g)+f_A\circ (F_Y^{-1}+g) \cdot h_2.
\end{align*}

We use the chain rule (see Theorem 20.9 in \cite{vandervaart2000}) to conclude that $\psi \circ \phi$ has Hadamard derivative at $(F_A,F_Y)$ tangentially to $\mathbb U\mathbb C(\mathcal Y)\times \mathbb C_\delta$ given by the map
\begin{align*}
(\psi \circ \phi)'_{(F_A,F_Y)}(h_1,h_2)&= \psi'_ {\phi(F_A,F_Y)}\circ  \phi'_ {(F_A,F_Y)}(h_1,h_2)\\
&=\psi'_ {(F_A,F_Y^{-1})}\circ (h_1,-h_2(F_Y^{-1})/f_Y(F_Y^{-1}))\\
&= h_1\circ (F_Y^{-1}+g)-f_A\circ (F_Y^{-1}+g) \frac{h_2\circ F_Y^{-1}}{f_Y\circ F_Y^{-1}}.
\end{align*}

By the functional Delta method (see Theorem 20.8 in \cite{vandervaart2000}) and the continuous mapping theorem (because of the $-1/\delta$ factor), we have that
\begin{align}\label{eq:theta_dist}
\sqrt n(\hat c_L-c_L)&=-\frac{1}{\delta}\sqrt n\left(\hat F_A\circ (\hat F_Y^{-1}+g)- F_A\circ( F_Y^{-1}+g)\right)\notag\\
&\rightsquigarrow -\frac{1}{\delta} (\psi \circ \phi)'_{(F_A,F_Y)}(\mathbb G_A,\mathbb G_Y))\notag\\
&=-\frac{1}{\delta}\mathbb G_A\circ (F_Y^{-1}+g)+\frac{1}{\delta}f_A\circ (F_Y^{-1}+g) \frac{\mathbb G_Y\circ F_Y^{-1}}{f_Y\circ F_Y^{-1}}.
\end{align}
Convergence takes place in $\ell^{\infty}(\delta,1-\delta)$. The limiting element is indeed Gaussian, as it is the linear combination of two Gaussian elements.

\end{proof}

\begin{proof}[Proof of Theorem \ref{qbf_marginal}]
For $G(y)=E[F_{Y|D=1,X}(y)|\partial D]$ or $G(y)=F_{Y|D=0}(y)$, we find the asymptotic distribution of $\sqrt n(\hat G \circ \hat F_Y^{-1}-G\circ  F_Y^{-1})$. Consider first the map $\psi: \mathbb D(\mathcal Y)^2\to \mathbb D(\mathcal Y)\times \ell^{\infty}(0,1)$, given by $\psi(H_1,H_2)=(H_1,H_2^{-1})$. Here, $\mathbb D(\mathcal Y)$ is the set of all restrictions of distribution functions on $\mathbb R$ to $\mathcal Y=[y_l,y_u]$, such that they give mass 1 to $(y_l,y_u]$. Also, $\mathbb C(\mathcal Y)$ is the set of all (uniformly) continuous functions defined on $\mathcal Y$.

By Lemma 21.4.(ii) in \cite{vandervaart2000}, and Assumption \ref{ass_had_dif_2}, $\psi$ is Hadamard differentiable tangentially to $\ell^{\infty}(\mathcal Y)\times \mathbb C(\mathcal Y)$ at $(G, F_Y)$, with derivative given by the map
\begin{align*}
\psi'_{(G, F_Y)}(h_1,h_2)=\left (h_1,-\frac{h_2\circ F_Y^{-1}}{f_y\circ F_Y^{-1} }\right).
\end{align*}

Now, consider the map $\phi:  \mathbb D(\mathcal Y)\times \ell^{\infty}(0,1)\to \ell^{\infty}(0,1)$ given by $\phi(H_1,H_2)=H_1\circ H_2^{-1}$. By Lemmas 3.9.25 and 3.9.27 in \cite{vandervaart1996}, and Assumption \ref{ass_had_dif_2}, $\phi$ has Hadamard derivative at $(G,F_Y^{-1})$ tangentially to $ \mathbb U\mathbb C(\mathcal Y)\times \ell^{\infty}(0,1)$ given by the map
\begin{align*}
 \phi'_ {(G,F_Y^{-1})}(h_1,h_2)=h_1\circ F_Y^{-1}+g\circ F_Y^{-1} \cdot h_2,
\end{align*}
where $g$ is the density of $G$.

We use the chain rule (see Theorem 20.9 in \cite{vandervaart2000}) to conclude that $\phi \circ \psi$ has Hadamard derivative at $(G,F_Y)$ tangentially to $\mathbb U\mathbb C(\mathcal Y)\times \mathbb C(\mathcal Y)$ given by the map
\begin{align*}
(\phi \circ \psi)'_{(G,F_Y)}(h_1,h_2)&= \phi'_ {\phi(G,F_Y)}\circ  \psi'_ {(G,F_Y)}(h_1,h_2)\\
&=\phi'_ {(G,F_Y^{-1})}\circ (h_1,-h_2\circ F_Y^{-1}/f_Y\circ F_Y^{-1})\\
&= h_1\circ F_Y^{-1}-g\circ F_Y^{-1} \cdot \frac{h_2\circ F_Y^{-1}}{f_Y\circ F_Y^{-1}}.
\end{align*}

By the functional Delta method (see Theorem 20.8 in \cite{vandervaart2000}) we have that
\begin{align*}
\sqrt n(\hat G\circ \hat F_Y^{-1}-G\circ  F_Y^{-1})
&\rightsquigarrow (\phi \circ \psi)'_{(G,F_Y)}(\mathbb G_G,\mathbb G_Y))\\
&=\mathbb G_G\circ F_Y^{-1}-g\circ F_Y^{-1} \cdot \frac{\mathbb G_Y\circ F_Y^{-1}}{f_Y\circ F_Y^{-1}}\\
&=:\mathbb G_{G,Y},
\end{align*}
where $\mathbb G_G$ is either $\mathbb G_{\partial D}$ or $\mathbb G_0$. We have
\begin{align*}
\hat c_{\tau,U}^M& :=  \sum_ {x\in\mathcal X}\hat F_{Y|D=1,X=x}(\hat F^{-1}_{Y}(\tau))\hat p_{x|\partial D}-\hat F_{Y|D=0}(\hat F^{-1}_{Y}(\tau)),\\
c_{\tau,U}^M &:= E[F_{Y|D=1,X}(F^{-1}_{Y}(\tau))|\partial D]- F_{Y|D=0}( F^{-1}_{Y}(\tau)).
\end{align*}
By the continuous mapping theorem
\begin{align*}
\sqrt n(\hat c_{\tau,U}^M - c_{\tau,U}^M)
\rightsquigarrow \mathbb G_{\theta_M}:= \mathbb G_{\partial D,Y}- \mathbb G_{0,Y}.
\end{align*}

\end{proof}

\subsection{Sufficient Conditions for Assumption \ref{brownian_bridge}}\label{app_primitive}
A sufficient condition for Assumption \ref{brownian_bridge} is the joint convergence of the empirical processes for the estimators of  $(F_Y, F_{Y|D=1}, F_{Y|D_\delta=0},F_{Y|D=1,X}, p, p_X)$ because $F_A$ is a function of them. Here, $F_{Y|D=1,X}$ is short hand for the $dim(\mathcal X)$-vector of CDFs $F_{Y|D=1,X=x}$ for each value of $x\in\mathcal X$. Same comment applies to $p_X$: it is a $dim(\mathcal X)$-vector that collects the (joint) pmf of $X$. 

Now we establish the asymptotic distribution of the apparent distribution: $\sqrt n(\hat F_A - F_A)$, as a process in $\ell ^{\infty}(\mathcal Y)$ which we denote by $\mathbb G_A$ . The estimator $\hat F_A(y)$ given in \eqref{emp_app_dis} is
\begin{align*}
\hat F_A(y) &:= \hat p \hat F_{Y|D=1}(y)  + (1 - \hat p - \delta) \hat F_{Y|D_\delta=0}(y) + \delta  \int_{\mathcal X}\hat F_{Y|D=1,X=x}(y)d\hat F_{X|D=0,D_\delta=1}(x)
\end{align*}
Since the support of $X$ is finite, then this can be written as
\begin{align*}
\hat F_A(y) &:= \hat p \hat F_{Y|D=1}(y)  + (1 - \hat p - \delta) \hat F_{Y|D_\delta=0}(y) + \delta \sum_{x\in\mathcal X} \hat F_{Y|D=1,X=x}(y)\hat p_x
\end{align*}
where $p_x := \Pr(X=x|D=0, D_\delta=1)$ and $\hat p_x$ is its empirical counterpart. The apparent counterfactual is
\begin{align}
 F_A(y) &:=  p  F_{Y|D=1}(y)  + (1 -  p - \delta)  F_{Y|D_\delta=0}(y) + \delta   \sum_{x\in\mathcal X}F_{Y|D=1,X=x}(y)p_x
\end{align}
and can be written as the map $\mathbb D(\mathcal Y)^{2+dim(\mathcal X)}\times (0,1)^{1+dim(\mathcal X)}\mapsto \mathbb D(\mathcal Y)$ given by
\begin{align*}
\psi(F_{Y|D=1}, F_{Y|D_\delta=0},F_{Y|D=1,X}, p, p_X ;\delta  ) &= p  F_{Y|D=1}(\cdot)  + (1 -  p - \delta)  F_{Y|D_\delta=0}(\cdot)\\
& +\delta   \sum_{x\in\mathcal X}F_{Y|D=1,X=x}(\cdot)p_x
\end{align*}
This map is linear, so the Hadamard derivative\footnote{We do not derivate with respect to $\delta$.} tangentially to $\ell^{\infty}(\mathcal Y)^{2+dim(X)}\times (0,1)^{1+dim(\mathcal X)}$ at $(F_{Y|D=1}, F_{Y|D_\delta=0},F_{Y|D=1,X}, p, p_X)$ is the map
\begin{align*}
\psi'_{F_{Y|D=1}, F_{Y|D_\delta=0},F_{Y|D=1,X}, p, p_X}(h_1,h_2,h_X, h_3, \tilde h_X) &=ph_1 + (1-p-\delta)h_2\\
&+ (F_{Y|D=1}(\cdot)-F_{Y|D_\delta=0}(\cdot))h_3\\
&+\delta \sum_{x\in\mathcal X}(h_x p_x+ F_{Y|D=1,X=x}\tilde h_x) 
\end{align*}
Here $h_X$ is $dim(X)$-vector of functions in $\ell^{\infty}(\mathcal Y)$, and $\tilde h_X$ is a $dim(X)$-vector that belongs to $(0,1)^{dim(\mathcal X)}$. By the functional Delta method (see Theorem 20.8 in \cite{vandervaart2000}) and Assumption \ref{brownian_bridge}, we have
\begin{align*}
\sqrt n( \hat F_A-F_A) &=  \sqrt n( \psi(\hat F_{Y|D=1}, \hat F_{Y|D_\delta=0},\hat F_{Y|D=1,X}, \hat p, \hat p_X ;\delta   )-\psi(F_{Y|D=1}, F_{Y|D_\delta=0},F_{Y|D=1,X}, p, p_X ;\delta  ) )\\
&\rightsquigarrow  p \mathbb G_1 + (1-p-\delta) \mathbb G_{0,\delta} + (F_{Y|D=1}(\cdot)-F_{Y|D_\delta=0}(\cdot))\mathbb Z_p \\
&+ \delta  \sum_{x\in\mathcal X}(\mathbb G_{1,x} p_x+ F_{Y|D=1,X=x}\mathbb Z_x) :=\mathbb G_A,
\end{align*}
and convergence takes place in $\ell^{\infty}(\mathcal Y)$. The random element $\mathbb G_A$ is Gaussian.

\subsection{Sufficient Conditions for Assumption \ref{brownian_bridge_marginal}}\label{app_primitive_marginal}
A sufficient condition for Assumption \ref{brownian_bridge_marginal} is the joint convergence of the empirical processes for the estimators of  $(F_{Y|D=0}, F_{Y|D=1,X}, F_{X|D=0}, F_{X|D=1})$, and application of the functional Delta method. Here, $F_{Y|D=1,X}$ is short hand for the $dim(\mathcal X)$-vector of CDFs $F_{Y|D=1,X=x}$ for each value of $x\in\mathcal X$.

\subsection{Inference for bounds derived from the QBF}\label{distribution_of_bounds_appendix}

\subsubsection{Asymptotic distribution}
The goal is to find the joint distribution of
\begin{align*}
\hat L_{\tau,\tau^*} &:=\max\left\{\hat{\tilde F}_A^{-1}(\tau-\delta) ,\hat F_A^{-1}(\tau-\delta \hat {\tilde c}_{\tau^*,L})\right\} - \hat F_{Y}^{-1}(\tau), \\
\hat U_{\tau,\tau^*} &:= \min\left\{\hat{\tilde F}_A^{-1}(\tau) ,\hat F_A^{-1}(\tau+\delta \hat {\tilde c}_{\tau^*,L})\right\}- \hat F_{Y}^{-1}(\tau).
\end{align*}
where $\hat {\tilde c}_{\tau^*,L} =  \max \left \{  \min \left \{  \hat c_{\tau^*,L}, 1\right\},0 \right\}$, and by \eqref{qbf_estimation},  
\begin{align*}
\hat c_{\tau^*,L}= \frac{\tau^* - \hat F_A\left ( \hat F_{Y}^{-1}(\tau^*)+ g \right )}{\delta}.
\end{align*}
as a process in $\ell^\infty(\delta,1-\delta)^2$. The main assumption we need is the following.
\begin{assumption}\label{assumption_joint_bounds} The following multivariate functional central limit theorem holds
\begin{align*} 
\sqrt n\begin{pmatrix}
\hat F_{Y} - F_Y\\
\hat F_{A}-  F_{A} \\
\hat {\tilde F}_{A}-  \tilde{F}_{A} \\
\end{pmatrix} \rightsquigarrow \begin{pmatrix}
\mathbb G_Y\\ 
\mathbb G_{A}\\
\tilde{\mathbb G}_A
\end{pmatrix},
\end{align*}
where $\mathbb G_Y$ , $\mathbb G_A$, and $\tilde{\mathbb G}_A$ are Brownian bridges in $\ell^{\infty}(\mathcal Y)$.
\end{assumption}

First consider the map $\Gamma_1: \mathbb D_\delta \times \mathbb D_A \times  \tilde{\mathbb D}_A   \to \ell^{\infty}(\delta, 1-\delta)^3\times [0,1]$, given by 
\begin{align*}
\Gamma_1 (H_1,H_2,H_3)= \left (H_1^{-1},H_2^{-1},H_3^{-1}, \tilde \Gamma_1 (H_1,H_2)\right ).
\end{align*}
where 
\begin{align*}
\tilde \Gamma_1 (H_1,H_2)=  \max \left \{  \min \left \{  \frac{\tau^* - H_2(H_1^{-1}(\tau^*)+g)}{\delta}, 1\right\},0 \right\}.
\end{align*}
 
Here, $\mathbb D_{\delta}\subset \ell^{\infty}(\mathcal Y)$ denote the set of all restrictions of distribution functions on $\mathbb R$ to $[F_Y^{-1}(\delta)-\varepsilon,F_Y^{-1}(1-\delta)+\varepsilon]$. $\mathbb D_A$ and $\tilde{\mathbb D}_A$ are defined analogously but with respect to $F_A$ and $\tilde F_A$. 
Let $\mathbb C_{\delta}$ bet the set of continuous functions on $[F_Y^{-1}(\delta)-\varepsilon,F_Y^{-1}(1-\delta)+\varepsilon]$. Define $\mathbb C_A$ and $\tilde{\mathbb C}_A$ similarly but with respect to $F_A$ and $\tilde F_A$. 

The next Assumption, similar to Assumption \ref{ass_had_dif} is needed for Hadamard Differentiability.
\begin{assumption}[Conditions for Hadamard Differentiability]\label{ass_had_dif_3}
\item
\begin{enumerate}
 \item Let $G$ be either $F_Y$, $F_A$, and $\tilde F_A$. For some $\varepsilon>0$, $G$  is continuously differentiable in $[G^{-1}(\delta)-\varepsilon,G^{-1}(1-\delta)+\varepsilon]\subset \mathcal Y$ with strictly positive derivative $g$.
\end{enumerate}
\end{assumption}

By Assumption \ref{ass_had_dif_3} and Lemma 21.4(i) in \cite{vandervaart2000}, $\Gamma_1$ has Hadamard derivative at $(F_Y, F_A, \tilde F_A)$ tangentially to $\mathbb C_\delta\times \mathbb C_A\times \tilde{\mathbb C}_A$ given by the map
\begin{align*}
 \Gamma'_ {1,(F_Y,F_A,\tilde F_A)}(h_1,h_2,h_3)=\left (-\frac{h_1\circ F_Y^{-1}}{f_Y\circ F_Y^{-1}},  -\frac{h_2\circ F_A^{-1}}{f_A\circ F_A^{-1}} ,  -\frac{h_3\circ \tilde {F}_A^{-1}}{\tilde f_A\circ \tilde{F}_A^{-1}} , \tilde \Gamma'_{1,(F_Y,F_A)}(h_1,h_2)\right).
\end{align*}
where the derivative $\tilde \Gamma'_{1,(F_Y,F_A)}(h_1,h_2)$ is shown below in an auxiliary Lemma \ref{aux_lemma}.

Now consider the map $\Gamma_2: \ell^{\infty}(\delta, 1-\delta)^3\times [0,1] \to \ell^{\infty}(\delta, 1-\delta)^5$ given by
\begin{align*}
\Gamma_2(H_1,H_2,H_3,H_4) = (H_1, H_2(\cdot + \delta H_4), H_2 (\cdot - \delta H_4), H_3, H_3 (\cdot - \delta))
\end{align*}
which is a combination of identity and composition maps, so it is Hadamard differentiable at $(F_Y,F_A,\tilde F_A, \hat {\tilde c}_{\tau^*,L} )$ by assumption to \ref{ass_had_dif}.\ref{ass_f_a_uniform} tangentially to  $ \ell^{\infty}(\delta,1-\delta)\times  \mathbb U\mathbb C(\mathcal Y)^2\times [0,1]$, with derivative given by
\begin{align*}
&\Gamma'_{2,(F_Y,F_A,\tilde F_A, \hat {\tilde c}_{\tau^*,L} )}(h_1,h_2,h_3,h_4)\\
& = (h_1, h_2(\cdot + \delta  \hat {\tilde c}_{\tau^*,L} ) + f_A(\cdot + \delta \hat {\tilde c}_{\tau^*,L})\delta h_4, h_2(\cdot - \delta  \hat {\tilde c}_{\tau^*,L} ) + f_A(\cdot - \delta \hat {\tilde c}_{\tau^*,L})\delta h_4, h_3, h_3 (\cdot - \delta))
\end{align*}

Finally, consider the map $\Gamma_3: \ell^{\infty}(\delta, 1-\delta)^5 \to \ell^{\infty}(\delta, 1-\delta)^2$, given by
\begin{align*}
\Gamma_3(H_1,H_2,H_3,H_4,H_5) = \begin{pmatrix} 
\max\left\{H_5 ,H_3\right\} - H_1\\
 \min\left\{H_4 ,H_2\right\}- H_1
\end{pmatrix}
\end{align*}
In particular
\begin{align*}
\Gamma_3(F_{Y}^{-1},F_A^{-1}(\cdot +\delta \hat {\tilde c}_{\tau^*,L}),F_A^{-1}(\cdot -\delta \hat {\tilde c}_{\tau^*,L}),{\tilde F}_A^{-1},{\tilde F}_A^{-1}(\cdot -\delta)) = \begin{pmatrix} 
L_{\cdot,\tau^*} \\
U_{\cdot,\tau^*}
\end{pmatrix}
\end{align*}
This map is Hadamard directional differentiable at $(F_{Y}^{-1},F_A^{-1}(\cdot +\delta  {\tilde c}_{\tau^*,L}),F_A^{-1}(\cdot -\delta  {\tilde c}_{\tau^*,L}),{\tilde F}_A^{-1},{\tilde F}_A^{-1}(\cdot -\delta))$ tangentially to $ \ell^{\infty}(\delta, 1-\delta)^5$ with derivative given by
\begin{align*}
&\Gamma'_{3,(F_{Y}^{-1},F_A^{-1}(\cdot +\delta  {\tilde c}_{\tau^*,L}),F_A^{-1}(\cdot -\delta  {\tilde c}_{\tau^*,L}),{\tilde F}_A^{-1},{\tilde F}_A^{-1}(\cdot -\delta))}(h_1,h_2,h_3,h_4,h_5)\\
&= \begin{pmatrix} 
h_5 \mathds 1 _{\{ {\tilde F}_A^{-1}(\cdot -\delta)>F_A^{-1}(\cdot -\delta  {\tilde c}_{\tau^*,L})\}} +  h_3  \mathds 1 _{\{ {\tilde F}_A^{-1}(\cdot -\delta)<F_A^{-1}(\cdot -\delta  {\tilde c}_{\tau^*,L})\}} +  \max\left\{h_5 ,h_3\right\} \mathds 1 _{\{ {\tilde F}_A^{-1}(\cdot -\delta)=F_A^{-1}(\cdot -\delta  {\tilde c}_{\tau^*,L})\}} - h_1\\
h_4 \mathds 1 _{\{ {\tilde F}_A^{-1}<F_A^{-1}(\cdot +\delta  {\tilde c}_{\tau^*,L})\}} +  h_2  \mathds 1 _{\{ {\tilde F}_A^{-1}>F_A^{-1}(\cdot +\delta  {\tilde c}_{\tau^*,L})\}} +  \min\left\{h_4 ,h_2\right\} \mathds 1 _{\{ {\tilde F}_A^{-1}=F_A^{-1}(\cdot +\delta  {\tilde c}_{\tau^*,L})\}}- h_1
\end{pmatrix}
\end{align*}

Therefore, the map 
\begin{align*}
\begin{pmatrix} 
L_{\cdot,\tau^*} \\
U_{\cdot,\tau^*}
\end{pmatrix} = \Gamma_3 \circ \Gamma_2 \circ \Gamma_1 (F_Y, F_A, \tilde F_A)
\end{align*}
is Hadamard directional differentiable at $F_Y, F_A, \tilde F_A$ tangentially to $\mathbb C_\delta\times \mathbb C_A\times \tilde{\mathbb C}_A$
by the chain rule for Hadamard directional differentiable maps (see Lemma C2 of \cite{Masten2020}). By Theorem 2.1 in \cite{Fang2019} and Assumption \ref{assumption_joint_bounds}, we get
\begin{align*}
\sqrt n\begin{pmatrix} 
 \hat L_{\cdot,\tau^*} -  L_{\cdot,\tau^*} \\
 \hat U_{\cdot,\tau^*} -  U_{\cdot,\tau^*}
\end{pmatrix} \rightsquigarrow  \left[\Gamma_3 \circ \Gamma_2 \circ \Gamma_1\right]' _{(F_Y, F_A, \tilde F_A)}(\mathbb G_Y, \mathbb G_{A},\tilde{\mathbb G}_A )
\end{align*}
and convergence takes place in $ \ell^{\infty}(\delta, 1-\delta)^2$.

\begin{lemma}\label{aux_lemma}
Under the assumptions of Theorem \ref{dist_theta}, the map $\tilde \Gamma_1 :  \mathbb D_\delta\times  \mathbb  D_A\mapsto \mathbb R$ given by 
\begin{align*}
\tilde \Gamma_1 (H_1,H_2)\mapsto \max \left \{  \min \left \{  \frac{\tau^* - H_2(H_1^{-1}(\tau^*)+g)}{\delta}, 1\right\},0 \right\}
\end{align*}
is Hadamard directionally differentiable at $(F_Y,F_A)$ tangentially to $\mathbb C_\delta\times \ell^{\infty}(\mathcal Y)$.
\end{lemma}

\begin{proof}[Proof of Lemma \ref{aux_lemma}]
We analyze the map
\begin{align*}
\tilde \Gamma_1 (H_1,H_2)\mapsto \max \left \{  \min \left \{  \frac{\tau^* - H_2(H_1^{-1}(\tau^*)+g)}{\delta}, 1\right\},0 \right\}
\end{align*}
separately. We write it as the composition of several maps. The first one is $\phi:  \mathbb D_\delta\times  \mathbb  D_A \mapsto  \mathbb R\times \mathbb D_A$ given by $\phi(H_1,H_2)\mapsto (H_1^{-1}(\tau^*),H_2)$. The second one is $\psi: \mathbb R\times \mathbb D_A\mapsto  \mathbb R$ given by $\psi(H_1,H_2)\mapsto (\tau^*-H_2 (H_1+g))/\delta$. Thus
\begin{align*}
\psi \circ \phi (F_Y,F_A)=\frac{\tau^*- F_A(F_Y^{-1}(\tau^*)+g)}{\delta}.
\end{align*}
The map $\phi$ essentially the same map $\phi$ in Theorem \ref{dist_theta}. By the results there, it has Hadamard derivative at $(F_Y,F_A)$ tangentially to $\mathbb C_\delta\times \ell^{\infty}(\mathcal Y) $ given by the map
\begin{align*}
 \phi'_ {(F_Y,F_A)}(h_1,h_2)=\left (-\frac{h_1 (F_Y^{-1}(\tau^*))}{f_Y (F_Y^{-1}(\tau^*))},h_2\right).
\end{align*}
Likewise, the map $\psi$ is essentially the map in Theorem \ref{dist_theta}. Therefore, it has Hadamard derivative at $(F_Y^{-1}(\tau^*),F_A)$ tangentially to $\mathbb R\times  \mathbb U\mathbb C(\mathcal Y)$ given by the map
\begin{align*}
 \psi'_ {(F_Y^{-1}(\tau^*),F_A)}(h_1,h_2)=-\frac{h_2 (F_Y^{-1}(\tau^*)+g)}{\delta}-\frac{f_A (F_Y^{-1}(\tau^*)+g) \cdot h_1}{\delta}.
\end{align*}

Now define
 \begin{align*}
\tilde \phi(H) = \max\{\min\{H,1\},0\}.
\end{align*}
and it is the composition of an evaluation map $H\in\mathbb R\mapsto H$ and of the max/min composition. The evaluation map is linear, hence fully Hadamard differentiable. The composition of max/min is Hadamard directional differentiable by the chain rule for Hadamard directional differentiable maps (see Lemma C2 of \cite{Masten2020}). Hence, another application of the chain rule yields that $\phi(H)$ is Hadamard directional differentiable at any $H\in\mathbb R$ tangentially to $\mathbb R$. By direct computation, the derivative, for any $h\in\mathbb R$, is given by the map
\begin{align}\label{phi_prime_der}
\tilde\phi'_H(h) =h\mathds{1}_{\left\{0< H< 1 \right\}} + \max(0,h)\mathds{1}_{\left\{H= 0 \right\}} + \min (0,h)\mathds{1}_{\left\{H= 1 \right\}}.
\end{align}
Therefore,
 \begin{align*}
\tilde\phi \circ \psi \circ \phi (F_Y,F_A)=  \max \left \{  \min \left \{  \frac{\tau^*- F_A(F_Y^{-1}(\tau^*)+g)}{\delta}, 1\right\},0 \right\}
\end{align*}
By the chain rule, $\tilde\phi \circ \psi \circ \phi $ has Hadamard directional derivative at $(F_Y,F_A)$ tangentially to $\mathbb C_\delta\times \ell^{\infty}(\mathcal Y) $, given by
 \begin{align*}
[\tilde\phi \circ \psi \circ \phi]'_{(F_Y,F_A)}(h_1,h_2)=  \tilde\phi'_{\psi \circ \phi (F_Y,F_A)}\circ  \psi'_ {\phi (F_Y,F_A)}\circ  \phi'_ {(F_Y,F_A)}(h_1,h_2).
\end{align*}
\end{proof}

\subsubsection{Numerical delta method}
To apply the numerical delta method of \cite{Hong2018}, we start with the map 
\begin{align*}
\begin{pmatrix} 
L_{\cdot,\tau^*} \\
U_{\cdot,\tau^*}
\end{pmatrix} = \Gamma_3 \circ \Gamma_2 \circ \Gamma_1 (F_Y, F_A, \tilde F_A)
\end{align*}
which is directionally differentiable. For notational brevity, we define two functions $\Gamma_L$ and $\Gamma_U$, which take as inputs $F_Y, F_A$ and $\tilde F_A$, and 
\begin{align*}
\begin{pmatrix} 
L_{\cdot,\tau^*} \\
U_{\cdot,\tau^*}
\end{pmatrix} = \begin{pmatrix} 
\Gamma_{L} (F_Y, F_A, \tilde F_A) \\
\Gamma_{U} (F_Y, F_A, \tilde F_A)
\end{pmatrix}. 
\end{align*}
Given $B$, the number of bootstrap repetitions, and a step size $\epsilon_n$ which satisfies $\epsilon_n\to 0$ and $\epsilon_n \sqrt{n}\to \infty$ as $n\to\infty$, the algorithm is the following:
\begin{enumerate}
\item Given a grid of values $\left\{\tau_k\right\}_{k=1}^K\subset(\delta,1-\delta)$, for each bootstrap sample $b$ in $1,\ldots, B$, compute the bootstrap counterparts $\hat F_Y^b, \hat F_A^b, \hat{\tilde F}_A^b$.
\item For each $b$, compute the \textit{perturbed} counterparts
\begin{align*}
\hat F_Y^{p,b} &= \hat F_Y + \epsilon_n\sqrt n(\hat F_Y^b-\hat F_Y)\\
\hat F_A^{p,b} &= \hat F_A + \epsilon_n\sqrt n(\hat F_A^b-\hat F_A)\\
\hat{\tilde F}_A^{p,b} &= \hat{\tilde F}_A + \epsilon_n\sqrt n(\hat{\tilde F}_A^b-\hat{\tilde F}_A)
\end{align*}
\item For each $b$, compute
\begin{align*}
\hat L_{\cdot,\tau^*}^b = \max_{\tau_1, ...., \tau_K} \left |\frac{\Gamma_{L} (\hat F_Y^{p,b}, \hat F_A^{p,b}, \hat{\tilde F}_A^{p,b})-\Gamma_{L} (\hat F_Y, \hat F_A, \hat{\tilde F}_A)}{\epsilon_n}\right |
\end{align*}
and
\begin{align*}
\hat U_{\cdot,\tau^*}^b = \max_{\tau_1, ...., \tau_K} \left |\frac{\Gamma_{U} (\hat F_Y^{p,b}, \hat F_A^{p,b}, \hat{\tilde F}_A^{p,b})-\Gamma_{U} (\hat F_Y, \hat F_A, \hat{\tilde F}_A)}{\epsilon_n}\right |
\end{align*}
\item Obtain the $1-\alpha$ order statistic from $\hat L_{\cdot,\tau^*}^b$. This is denoted $\xi_{1-\alpha,L_{\tau_1}}$.
\item Obtain the $1-\alpha$ order statistic from $\hat U_{\cdot,\tau^*}^b$. This is denoted $\xi_{1-\alpha,U_{\tau_2}}$.
\end{enumerate}

The $1-\alpha$ confidence bands are then computed as 
\begin{align*}
\mathcal {CB}( L_{\tau,\tau^*} ,\alpha)&=\left[\hat L_{\tau,\tau^*}  -\frac{\xi_{1-\alpha,L_{\tau_1}}}{\sqrt n}, \hat L_{\tau,\tau^*} +\frac{\xi_{1-\alpha,L_{\tau_1}}}{\sqrt n}\right],\\
\mathcal {CB}( U_{\tau,\tau^*} ,\alpha)&=\left[\hat U_{\tau,\tau^*}  -\frac{\xi_{1-\alpha,U_{\tau_2}}}{\sqrt n}, \hat U_{\tau,\tau^*} +\frac{\xi_{1-\alpha,U_{\tau_2}}}{\sqrt n}\right].
\end{align*}

The simultaneous $1-\alpha$ confidence bands, by the Bonferroni correction, are given by the Cartesian product
\begin{align*}
\mathcal {CB}( L_{\tau,\tau^*}, U_{\tau,\tau^*},\alpha)=\mathcal {CB}(L_{\tau,\tau^*} ,\alpha/2)\times \mathcal {CB}(U_{\tau,\tau^*} ,\alpha/2).
\end{align*}

\subsection{Finite-Sample Performance of the Numerical Delta Method}\label{appendix_mc}
The confidence bands of Appendix \ref{distribution_of_bounds_appendix} rely on an application of the
numerical delta method of \citet{Hong2018}. The method requires from the user a tuning parameter, or step size, denoted by $\epsilon_n$. This section reports a small Monte Carlo study of
the resulting coverage.

 The design follows Theorem \ref{thm_bounds_global} but with no
covariates. The three building blocks are $F_{Y|D=1}=N(0.5,1)$, $F_{Y|D=0}=N(0,1)$ and
$F_{Y|D_\delta=0}=N(0.2,1)$, with $p=0.5$, $\delta=0.1$, $\tau^*=0.5$ and $g=0$, so that $F_Y = pF_{Y|D=1}+(1-p)F_{Y|D=0}$, $\tilde F_A = pF_{Y|D=1}+(1-p-\delta)F_{Y|D_\delta=0}$, and $F_A = (p+\delta)F_{Y|D=1}+(1-p-\delta)F_{Y|D_\delta=0}.$
The fact that $F_{Y|D_\delta=0}$ is different from $F_{Y|D=0}$ is key:
under a non-randomized policy that preferentially treats low-wage control workers, the control group
that remains in control has a higher mean than the full control group. These values give
$\tilde c_{\tau^*,L}=0.51$. To understand where the Hadamard directional differentiability is needed, we note that in the no-covariate case, by Theorem \ref{thm_bounds_global}, we have
$F_A(y)-\tilde F_A(y)=\delta F_{Y|D=1}(y)$, so the quantile levels at which
the $\max$ in $L_{\tau,\tau^*}$ and the $\min$ in $U_{\tau,\tau^*}$ switch their active argument can
be computed numerically. They are $\tau_L^*=0.586$ and $\tau_U^*=0.508$. 

We draw $n_1=np$ observations from $F_{Y|D=1}$, $n_0=n(1-p)$ from $F_{Y|D=0}$ and
$n_{0\delta}=n(1-p-\delta)$ from $F_{Y|D_\delta=0}$, estimate each by its empirical distribution
function, so that Assumption \ref{assumption_joint_bounds} holds by Donsker's theorem, and evaluate
the bounds on a $42$-point grid $\tau\in[0.15,0.85]$. We use $n\in\{250,500,1000,2000,10000\}$, $500$
Monte Carlo replications and $300$ bootstrap repetitions. Each bound is held to a
nominal $97.5\%$, which gives a joint $95\%$ by the Bonferroni correction. In Table \ref{tab:mc_pt_app}, we report pointwise coverage, averaged over
the grid, and uniform coverage, which is the proportion of Monte Carlo replications for which the entire grid is covered.

\begin{table}[ht]
\centering\small
\begin{tabular}{ccccccccc}
\toprule
& & & \multicolumn{2}{c}{Pointwise ($97.5\%$)} & \multicolumn{3}{c}{Uniform} \\
\cmidrule(lr){4-5}\cmidrule(lr){6-8}
$n$ & $\epsilon_n$ & $\epsilon_n\sqrt n$ & $L$ & $U$ & $L$ & $U$ & Joint ($95\%$) \\
\midrule
$250$  & $n^{-1/2}$ & $1.00$  & $0.990$ & $0.971$ & $0.996$ & $0.990$ & $0.986$ \\
$250$  & $n^{-1/3}$ & $2.51$  & $0.979$ & $0.939$ & $0.956$ & $0.922$ & $0.890$ \\
$250$  & $n^{-1/6}$ & $6.30$  & $0.954$ & $0.867$ & $0.806$ & $0.752$ & $0.638$ \\[3pt]
$500$  & $n^{-1/2}$ & $1.00$  & $0.991$ & $0.972$ & $0.992$ & $0.980$ & $0.972$ \\
$500$  & $n^{-1/3}$ & $2.82$  & $0.984$ & $0.941$ & $0.950$ & $0.912$ & $0.876$ \\
$500$  & $n^{-1/6}$ & $7.94$  & $0.966$ & $0.868$ & $0.856$ & $0.742$ & $0.678$ \\[3pt]
$1000$ & $n^{-1/2}$ & $1.00$  & $0.988$ & $0.970$ & $0.980$ & $0.974$ & $0.964$ \\
$1000$ & $n^{-1/3}$ & $3.16$  & $0.982$ & $0.943$ & $0.964$ & $0.918$ & $0.904$ \\
$1000$ & $n^{-1/6}$ & $10.00$ & $0.968$ & $0.870$ & $0.876$ & $0.720$ & $0.666$ \\[3pt]
$2000$ & $n^{-1/2}$ & $1.00$  & $0.985$ & $0.976$ & $0.982$ & $0.980$ & $0.966$ \\
$2000$ & $n^{-1/3}$ & $3.55$  & $0.980$ & $0.951$ & $0.962$ & $0.934$ & $0.908$ \\
$2000$ & $n^{-1/6}$ & $12.60$ & $0.969$ & $0.882$ & $0.884$ & $0.748$ & $0.706$ \\[3pt]
$10000$ & $n^{-1/2}$ & $1.00$  & $0.981$ & $0.981$ & $0.988$ & $0.980$ & $0.970$ \\
$10000$ & $n^{-1/3}$ & $4.64$  & $0.979$ & $0.979$ & $0.986$ & $0.964$ & $0.954$ \\
$10000$ & $n^{-1/6}$ & $21.54$ & $0.974$ & $0.914$ & $0.904$ & $0.814$ & $0.774$ \\
\bottomrule
\end{tabular}
\caption{\emph{Coverage of the bounds $L$ and $U$. Pointwise figures are averaged
over the $42$-point grid. Each bound is at a nominal $97.5\%$, giving a joint $95\%$ by the
Bonferroni correction. $500$ Monte Carlo replications, $300$ bootstrap repetitions.}}
\label{tab:mc_pt_app}
\end{table}

\begin{table}[ht]
\centering\small
\begin{tabular}{cccc}
\toprule
$n$ & $\epsilon_n$ & $\widehat{\mathrm{Cov}}\!\left(L_{\tau_L^*}\right)$ & $\widehat{\mathrm{Cov}}\!\left(U_{\tau_U^*}\right)$ \\
\midrule
$250$  & $n^{-1/2}$ & $0.984$ & $0.964$ \\
$250$  & $n^{-1/3}$ & $0.976$ & $0.912$ \\
$250$  & $n^{-1/6}$ & $0.974$ & $0.864$ \\[3pt]
$500$  & $n^{-1/2}$ & $0.982$ & $0.966$ \\
$500$  & $n^{-1/3}$ & $0.986$ & $0.916$ \\
$500$  & $n^{-1/6}$ & $0.982$ & $0.858$ \\[3pt]
$1000$ & $n^{-1/2}$ & $0.984$ & $0.968$ \\
$1000$ & $n^{-1/3}$ & $0.988$ & $0.916$ \\
$1000$ & $n^{-1/6}$ & $0.992$ & $0.872$ \\[3pt]
$2000$ & $n^{-1/2}$ & $0.980$ & $0.972$ \\
$2000$ & $n^{-1/3}$ & $0.966$ & $0.920$ \\
$2000$ & $n^{-1/6}$ & $0.972$ & $0.876$ \\[3pt]
$10000$ & $n^{-1/2}$ & $0.988$ & $0.976$ \\
$10000$ & $n^{-1/3}$ & $0.978$ & $0.978$ \\
$10000$ & $n^{-1/6}$ & $0.988$ & $0.920$ \\
\bottomrule
\end{tabular}
\caption{\emph{Pointwise coverage at the grid points closest to the two points of
non-differentiability: $\tau=0.594$ for $L$ ($\tau_L^*=0.586$), where the $\max$ in $L_{\tau,\tau^*}$
switches its active argument, and $\tau=0.509$ for $U$ ($\tau_U^*=0.508$), where the $\min$ in
$U_{\tau,\tau^*}$ does. Nominal level $97.5\%$.}}
\label{tab:mc_cross_app}
\end{table}

Looking at the results in Tables \ref{tab:mc_pt_app} and \ref{tab:mc_cross_app}, four patterns emerge. First, for $n\leq 2000$, the bands for $U$ cover less than
those for $L$ for every step size, most visibly at $\tau_U^*$. Second, with $\epsilon_n=n^{-1/2}$ joint coverage is above the nominal level
throughout. This choice of $\epsilon_n$ violates the rate condition
and reduces to the standard bootstrap, which is inconsistent by Corollary 3.1 of \cite{Fang2019}. Third, $\epsilon_n=n^{-1/6}$ seems to perturb too much at these
sample sizes, relative to $n^{-1/3}$, with $\epsilon_n\sqrt n$ between $6$ and $22$. Moreover, this choice of tuning parameter under-covers substantially.
Fourth, $\epsilon_n=n^{-1/3}$, which satisfies both rate conditions, under-covers at the smaller
sample sizes, with joint coverage between $0.88$ and $0.91$ for $n\leq 2000$, however it reaches the nominal level at $n=10000$, where joint coverage is $0.954$. Therefore, it is our preferred choice of step size.

\subsection{Asymptotic distribution of $\hat d_{\tau}$}\label{appendix_d_tau}
In this section, we are going to establish the asymptotic distribution of $\hat d_{\tau}=\hat{\tilde F}_A^{-1}(\tau-\delta)-\hat F_Y^{-1}(\tau)$. This follows from results established in Appendix \ref{distribution_of_bounds_appendix}.
By Assumptions \ref{assumption_joint_bounds} and \ref{ass_had_dif_3},
the inverse map $(H_1,H_3)\mapsto\big(H_1^{-1}(\cdot),H_3^{-1}(\cdot)\big)$
is Hadamard differentiable at $(F_Y,\tilde F_A)$ tangentially to
$\mathbb C_\delta\times\tilde{\mathbb C}_A$, with derivative given by
(Lemma~21.4(i) in \cite{vandervaart2000})
\begin{align*}
(h_1,h_3)
\mapsto
\left(
-\frac{h_1\circ F_Y^{-1}}{f_Y\circ F_Y^{-1}},
-\frac{h_3\circ \tilde F_A^{-1}}{\tilde f_A\circ \tilde F_A^{-1}}
\right).
\end{align*}

Now define the linear map
$T:\ell^\infty(\delta,1-\delta)^2\to\ell^\infty(\delta,1-\delta)$ by $T(u,v)(\tau)=v(\tau-\delta)-u(\tau).$ 
Then $\tau\mapsto d_\tau$ can be written as
\begin{align*}
d&= T\big(F_Y^{-1},\tilde F_A^{-1}\big),\\
\hat d &= T\big(\hat F_Y^{-1},\hat{\tilde F}_A^{-1}\big).
\end{align*}
Since $T$ is continuous linear, it is Hadamard differentiable with derivative $T$
itself. Hence, by the functional delta method,
\begin{align*}
\sqrt n(\hat d-d)
&\rightsquigarrow
T\left(
-\frac{\mathbb G_Y\circ F_Y^{-1}}{f_Y\circ F_Y^{-1}},
-\frac{\tilde{\mathbb G}_A\circ \tilde F_A^{-1}}{\tilde f_A\circ \tilde F_A^{-1}}
\right)\\
&=\frac{\mathbb G_Y\!\left(F_Y^{-1}(\tau)\right)}{f_Y\!\left(F_Y^{-1}(\tau)\right)}
-
\frac{\tilde{\mathbb G}_A\!\left(\tilde F_A^{-1}(\tau-\delta)\right)}
     {\tilde f_A\!\left(\tilde F_A^{-1}(\tau-\delta)\right)},
\end{align*}
in $\ell^\infty(\delta,1-\delta).$

\end{document}